\documentclass[onecolumn]{article}

\usepackage{amsthm}

\usepackage{amsmath}
\usepackage{amsfonts}
\usepackage{amssymb}
\usepackage{graphicx}
\usepackage{caption}
\usepackage{subcaption}
\usepackage{MnSymbol,wasysym}
\usepackage{cancel}
\usepackage{float}
\usepackage{splitbib}
\usepackage{color}
\usepackage{yfonts}

\begin{category}[BIO]{Biology}
  \SBentries{okubo1986dynamical,flierl1999individuals,camazine2003self,couzin2003self,sumpter2006principles,hildenbrandt2010self}
\end{category}

\begin{category}[PHY]{Physics}
  \SBentries{ben1992adaptive,vicsek1995novel,mogilner1996spatio,vicsek2012collective,cavagna2014bird}
\end{category}

\begin{category}[SYNC]{Synchronisation}
  \SBentries{mirollo1990synchronization,strogatz2000kuramoto,wang2005partial,dorfler2014synchronization}
\end{category}

\begin{category}[DC]{Distributed consensus}
  \SBentries{cybenko1989dynamic,xiao2004fast,olshevsky2009convergence}
\end{category}

\begin{category}[CG]{Computer graphics / Collective Behaviour}
  \SBentries{reynolds1987flocks,chazelle2014convergence,chazelle2015algorithmic}
\end{category}

\begin{category}[ANT]{Ants to A(ge)nts}
  \SBentries{feynman1985surely,bruckstein1993ant,bruckstein1991ants,wagner1997row,bruckstein1997probabilistic,marshall2004formations,lin2005necessary,belkhouche2005modeling,martinez2006optimal,sinha2006generalization,sinha2007generalization,hristu2007bio,oggier2012cyclic}
\end{category}

\begin{category}[CONT]{Control theory}
  \SBentries{gazi2003stability,gazi2004stability,olfati2007consensus,moreau2004,reif1999social,jadbabaie2003,ren2005consensus,ji2007,cucker2007emergent,motsch2014heterophilious}
\end{category}

\begin{category}[CS]{Computer science}
  \SBentries{suzuki1999distributed,ando1999,cieliebak2003solving,schlude2003robotics,schlude2003point,gordon2004,gordon2005,flocchini2005gathering,cohen2005convergence,agmon2006fault,cortes2006robust,martinez2007motion,gordon2008,gordon2010fundamental,cieliebak2012distributed,bellaiche2015}
\end{category}

\begin{category}[GE]{Geometry}
  \SBentries{elzinga1972geometrical}
\end{category}

\begin{category}[BOOK]{Books}
  \SBentries{mamei2006field,bullo2009distributed,mesbahi2010graph,gazi2011swarm,bonato2011game,flocchini2012distributed}
\end{category}

\date{July 19, 2016}
\title{COME TOGETHER:\\ Multi-Agent Geometric Consensus\\ \large (Gathering, Rendezvous, Clustering, Aggregation)\\ \textcolor{white}{C}\\ Ariel Barel, Rotem Manor, and Alfred M. Bruckstein \\\textcolor{white}{C}\\
Center for Intelligent Systems (CIS)\\ Multi-Agent Robotic Systems (MARS) Laboratory \\ Technion Autonomous Systems Program (TASP) \\ Computer Science Department \\ Technion, Haifa 32000, Israel.}

 \newtheorem{theorem}{\textit{Theorem}}
 \newtheorem{lemma}{\textit{Lemma}}
\newtheorem{corollary}{\textit{Corollary}}
\newtheorem{proposition}{\textit{Proposition}}
\newtheorem{definition}{\textit{Definition}}

\newtheorem{claim}{\textit{Claim}}
\begin{document}
\maketitle

\begin{abstract}
This report surveys results on distributed systems comprising mobile agents that are identical and anonymous, oblivious and interact solely by adjusting their motion according to the relative location of their neighbours. The agents are assumed capable of sensing the presence of other agents within a given sensing range and able to implement rules of motion based on full or partial information on the geometric constellation of their neighbouring agents. Eight different problems that cover assumptions of \textit{finite} vs \textit{infinite} sensing range, \textit{direction and distance} vs \textit{direction only} sensing and \textit{discrete} vs \textit{continuous} motion, are analyzed in the context of geometric consensus, clustering or gathering tasks.\\

\textcolor{white}{  \cite{reif1999social,jadbabaie2003,gazi2003stability,gazi2004stability,moreau2004,ren2005consensus,olfati2007consensus,ji2007,cucker2007emergent,motsch2014heterophilious}
\cite{reynolds1987flocks,chazelle2014convergence,chazelle2015algorithmic}
\cite{feynman1985surely,bruckstein1991ants,bruckstein1993ant,wagner1997row,bruckstein1997probabilistic,marshall2004formations,lin2005necessary,belkhouche2005modeling,martinez2006optimal,sinha2006generalization,sinha2007generalization,hristu2007bio,oggier2012cyclic}
\cite{suzuki1999distributed,ando1999,cieliebak2003solving,schlude2003robotics,schlude2003point,gordon2004,gordon2005,flocchini2005gathering,cohen2005convergence,agmon2006fault,cortes2006robust,martinez2007motion,gordon2008,gordon2010fundamental,cieliebak2012distributed,bellaiche2015}
\cite{mamei2006field,bullo2009distributed,mesbahi2010graph,gazi2011swarm,bonato2011game,flocchini2012distributed}
}
\end{abstract} 

\newpage
\setcounter{tocdepth}{2}
\tableofcontents

\newpage

\begin{flushright}
\textfrak{\huge"Come together, right now, over me \\
 Come together..." \textcolor{white}{ccccccccccccc} }
\\ \hfill The Beatles, 1969.
\end{flushright}
\addcontentsline{toc}{section}{Introduction}

\section*{Introduction}

Nature provides amazing examples of complex goal-oriented global behaviours in multi-agent systems. The coordinated, cooperative work in colonies of ants, termites and bees, the synchronized flight in flocks of birds and swarms of locusts, the coordinated swimming and intricate spatial pattern formation in schools of fish, the waves in migrations of large animal herds, the synchronized flashing of fireflies were observed, described, admired and studied by humans for ages. These natural phenomena raise a wealth of questions. The fundamental question is this: how can groups of locally interacting, often quite myopic and simple agents, perform complex and coordinated tasks without having any centralized control mechanism. The scientific, mathematical modelling-based study of flocking, swarming, and schooling behaviours is presently a very active research area. As our technology progresses we learn more and more from the ants, birds and fish on these topics, however, despite considerable progress in data collection and analysis, we are still quite far from fully understanding the laws or the mechanisms underlying nature's multi-agent systems. As scientists and engineers, we would like to know the laws of inter-agent interaction, and be able to mimic or simulate the observed collective behaviours based on these laws. Furthermore, we would like to have goal-oriented design processes for deriving rules of local interaction between agents that ensure a desired global behaviour for a multi-agent system   solving a given complex problem.\\

In robotics applications, we aim to build colonies of simple, interacting, mobile agents to solve problems like sweeping and cleaning of an area, detecting and tracking plumes of chemical materials that spread in the environment, patrolling a region and detecting intruders, searching for stationary and moving entities in uncharted areas. All such tasks can benefit from the deployment of multiple interacting agents, in order to increase efficiency and achieve reliability through redundancy.\\

In this report we analyze a fundamental multi-agent task, the task of gathering or clustering or getting together. Suppose that many identical mobile agents are dispersed in some region and have to get together in order to subsequently perform some tasks as a cohesive group. To accomplish this, the agents can rely only on their sensors which provide information about the relative location of other agents in their neighbourhood, since we assume that they lack communication capabilities, and do not share a common geometric frame of reference. The agents, in our model, are not only identical, indistinguishable (i.e. are anonymous) and cannot directly communicate with each other, but also lack the capability to collect, remember and store information about the environment and about past configurations. This means that they are memoryless, or oblivious, and their actions, i.e. their movements in the environment, will be determined solely based on what they presently "see" with their sensors. The above-discussed limitations imposed on the agents may  formidable, however, as we shall see, by assuming that the agents are capable to carry out some geometric computations based on what they sense about the constellation of their neighbours (each in their own frame of reference), one can design local rules of motion whose implementation by all agents provably cause them to get together.\\

\textbf{A brief overview of the multi-agent literature}\\

There is a vast literature dealing with multi-agent systems, spanning areas of research from sociobiology to physics, from computer graphics to robotics, from control theory to theoretical computer science.\\

Biologists realised early on that swarming behaviours in various species entail advantages in survival, such as joint foraging for food, energy saving in motion, protection from predators, better navigation, and so on. Therefore many papers in both descriptive and theoretical/mathematical biology speculate on the types of local interactions that lead to the observed phenomena, as the papers \cite{okubo1986dynamical, flierl1999individuals, camazine2003self, couzin2003self,sumpter2006principles,hildenbrandt2010self} demonstrate. The biological observations were carefully scrutinised by other scientists and by engineers as well. Physicists, by definition interested in all aspects of nature, were drawn to the challenge of explaining swarming behaviours, with scaling effects and phase transitions that occur in systems composed of large numbers of particles or agents.  Sometimes they call the agents self-driven particles, depicting tiny living cells, and the work of Vicsek and Eshel Ben-Jacob and their collaborators who were interested in modelling biological phenomena from bacterial colonies to insects and to human traffic was very influential in this direction, see e.g. \cite{ben1992adaptive, vicsek1995novel, mogilner1996spatio, vicsek2012collective, cavagna2014bird}.\\

Physicists and mathematicians were also drawn to modelling distributed synchronisation phenomena due to the observed emergence of simultaneous flashing in large crowds of fireflies and of such phenomena such as synchronized clapping of hands in concert halls \cite{mirollo1990synchronization, strogatz2000kuramoto, wang2005partial, dorfler2014synchronization}.\\

Early on in computer science research, people considered spatially distributed networks of processors with given localised (often nearest neighbour based) communication links and analyzed ways to solve complex problems by efficiently exploiting the formidable but spread-out computational capabilities available in such systems. Problems like load-balancing, distributed algorithms for consensus, averaging, gossiping, and leader election, were, and still are, central topics of study in this area \cite{cybenko1989dynamic,xiao2004fast,olshevsky2009convergence}.\\

In computer graphics, with the increase in computational power and the advance of display technologies, researchers became interested in sophisticated simulations of autonomously behaving agents in flocks of birds flying in the sky or schools of fish swimming under the sea, or herds of horses or bison running in fierce stampedes. This led to the popular BOIDS local interaction model of Craig Reynolds \cite{reynolds1987flocks} and many subsequent works, which in turn influenced all other areas of scientific investigation on collective behaviours induced by local interactions, see e.g. \cite{chazelle2014convergence,chazelle2015algorithmic}.\\

Ants, working in efficient colonies and achieving amazing coordinated feats despite the simple reactive behaviour of the anonymous and indistinguishable individuals, are a continuous source of wonder and inspiration for distributed multi-agent systems research. Inspired by a description in Feynman's "Surely You're Joking Mr. Feynman" \cite{feynman1985surely}, about his experiments in dealing with an ant colony that invaded his house, Bruckstein proposed a pursuit model for local interaction to explain the straightness of ant-trails after they find food and start to ferry it to the ant hill, see \cite{bruckstein1993ant}. Chain pursuit and cyclic pursuit were then thoroughly analyzed as models of local interaction that achieve some global results like clustering or gathering, and finding shortest navigation paths for robotic agents (see references \cite{bruckstein1991ants, wagner1997row, bruckstein1997probabilistic,  marshall2004formations, lin2005necessary, belkhouche2005modeling, martinez2006optimal, sinha2006generalization, sinha2007generalization, hristu2007bio,oggier2012cyclic}).\\

Engineers from different disciplines interested in robotics and complex systems with many interacting parts also realised quite early on that principles guiding social animals in their struggle for survival could and should be exploited in the design of artificial multi-agent robotic colonies. In systems involving large numbers of agents, one is faced with the need to control and lead them in coordinated formations, and to do so, one may encounter formidably complicated controller-to-agents and inter-agent communication issues. Such systems may greatly benefit from having components acting according to simple distributed rules of local interaction, with no need for explicit inter-agent communication, if such rules can be designed to ensure that the system's autonomous evolution will accomplish the necessary global goals. Such systems could deploy a large number of identical low cost, and rather simple agents, from the point of view of their computational power, memory capacity and sensing capabilities, operating autonomously with no need for explicit inter-agent communication. In addition to the obvious benefits of simplicity and autonomy, systems composed of such agents, programmed to carry out rules of local, neighbourhood based interaction, also achieve scalability and fault tolerance (reliability through redundancy), much in the same way as an ant colony is not affected by the elimination of scores of individual agents, and by the constant influx of new individuals born to the colony.\\

The interest of control engineers and robotics researchers in distributed multi-agent design topics led to the development of several types of multi-agent interaction rules, based on "potential functions" or "influence fields" and networked control systems. The papers \cite{gazi2003stability, gazi2004stability,olfati2007consensus,moreau2004} analyze issues of stability, emergent behaviour of consensus and coordination for various types of so-called "networked" multi-agent systems under different assumption on individual agent dynamics (integrator, or direct velocity control, unicycles, and double integrator or force control) and different types of distance-to-neighbour based "influence functions" or "potential functions". This field of research is still quite active, with many beautiful results already available, see e.g. \cite{reif1999social,jadbabaie2003,ren2005consensus, ji2007,cucker2007emergent, motsch2014heterophilious}, but there are also many outstanding research questions that remain unanswered to date.\\

Computer scientists and mathematicians interested in the topic of gathering and coordination in swarms of mobile agents have addressed many variations on geometric consensus problems assuming that the agents "see" the constellation of their neighbours as the "input" and decide where to go next as a result of some computations on this input and on the assumed motion capabilities of the agents, see \cite{suzuki1999distributed,ando1999,cieliebak2003solving,schlude2003robotics,schlude2003point,gordon2004,gordon2005,flocchini2005gathering,cohen2005convergence,agmon2006fault,cortes2006robust,martinez2007motion,gordon2008,gordon2010fundamental,cieliebak2012distributed,bellaiche2015}.
Several books devoted to the subject of multi-agent systems have already been published and summarize results from the various points of view discussed above. A partial list is provided in the bibliography, see \cite{mamei2006field,bullo2009distributed,mesbahi2010graph,gazi2011swarm,bonato2011game,flocchini2012distributed}.
\\

 \textbf{The Gathering or Geometric Consensus Problem}\\

Gathering or clustering or "coming together" is a widely studied problem by researchers interested in multi-agent systems. It is also known as a geometric consensus or "distributed agreement" problem. The studies devoted to this problem assume various types of mobile agent reactive motion control or dynamics, based on "rules" about how agents "influence" each other. In distributed computing a frequently discussed problem of consensus or agreement is the following: one has a network of $n$ computers - the agents - connected by communication links and each of these computing agents has a certain value, so that agent $i$ has value $v_i$. The aim is to have all the agents compute the same deterministic function of the values $v_1, v_2, ... , v_n$, for example their average. The computation should be done by the network as efficiently as possible by exchanging information over the communication links. For "average consensus", each agent could send its value to its "neighbours" (i.e. the agents it is directly connected to), and each agent could replace its value by some weighted average of all the values it has seen, including its own.\\

This relatively simple topic already poses quite interesting challenges in proving convergence and in selecting the weights for fastest convergence, see e.g.\cite{xiao2004fast,olshevsky2009convergence}. In our case we deal with mobile robots as agents, and aim to design their motion in response to the geometric constellation of their neighbours, in order to ensure their convergence to a point or a small region, i.e. a "consensus location" in space. The neighborhood of the robots is dependent on their location and sensing horizon, and the robots must use relative location coordinates since they often are assumed not to have access to common, or even aligned frames of reference in space.\\

Suppose a mobile agent $i$, located at $p_i$ in the environment exerts an omnidirectional "influence field" that depends on the distance $r$ from $p_i$, as described by a function $\Phi_{p_i}(r)$, and all agents sense the combined effect of the scalar fields induced by all other agents from their location. If we postulate that the agents will move in the negative gradient direction of the combined "influence" field, we may write that

$$\frac{dp_i(t)}{dt}=-\alpha \frac{\partial \sum_{j=1}^{n}\Phi(\|p_j(t)-p_i(t)\|)}{\partial p_i(t)}$$ 
and this yields
$$\frac{dp_i(t)}{dt}=+\alpha \sum_{j \neq i} \frac{\dot{\Phi} (\|p_j(t)-p_i(t)\|)}{\|p_j(t)-p_i(t)\|} (p_j(t)-p_i(t)) = \alpha \sum_{j \neq i} \dot{\Phi} (\|p_j(t)-p_i(t)\|)  \vec u_{i \to j}(t)$$
Note that we used here the fact that
$$\|p_j(t)-p_i(t)\|=((p_j(t)-p_i(t))^\intercal(p_j(t)-p_i(t)))^{1/2}$$ 
hence 
$$\frac{\partial \|p_j(t)-p_i(t)\|}{\partial p_i(t)}=\frac{1}{\|p_j(t)-p_i(t)\|} (p_j(t)-p_i(t))=\vec u_{i \to j}(t) $$
the unit vector from $p_i(t)$ to  $p_j(t)$, see Figure \ref{UnitVectorItoJ}.\\

\begin{figure}[H]
\captionsetup{width=0.8\textwidth}
  \centering
    \includegraphics[width=60mm]{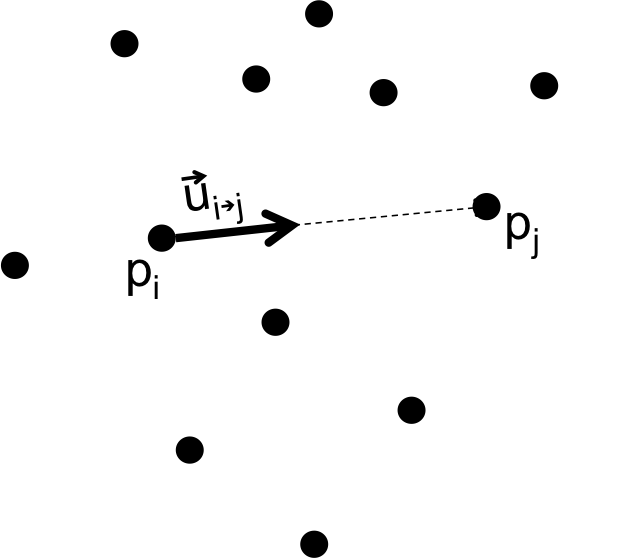}
    \caption{$\vec u_{i \to j}(t)$ is the unit vector pointing from $p_i(t)$ to $p_j(t)$.}
      \label{UnitVectorItoJ}
\end{figure}
For example we may consider (increasing) influence functions of the following particular forms:

\begin{itemize}

\item $ \Phi(r) = \frac{1}{2}r^2 $ \quad yields \quad $\frac{dp_i(t)}{dt} = \alpha \sum\limits_{j \neq i} (p_j(t) -p_i(t)) = \alpha \sum\limits_{j \neq i} \|p_j(t) -p_i(t)\|\vec u_{i \to j}(t) $

\item $\Phi(r) = r$  \quad yields \quad $ \frac{dp_i(t)}{dt} =\alpha \sum\limits_{j \neq i} \frac{p_j(t) - p_i(t)}{\|p_j(t) - p_i(t)\|} = \alpha \sum\limits_{j \neq i}\vec u_{i \to j}(t)$

\item $\Phi(r) = \ln(r+1)$ \quad yields \quad  $\frac{dp_i(t)}{dt} = \alpha \sum\limits_{j \neq i} \frac{1}{(\|p_j(t) - p_i(t)\|+1)}\vec u_{i \to j}(t)$

\end{itemize}
On the other hand, if an influence function $\Phi(r)$ decays with distance, for example as follows: 
$$\Phi(r) = \frac{1}{(r+1)^\beta}$$ 
we are led to:
$$ \frac{dp_i(t)}{dt} = -\alpha\beta \sum\limits_{j \neq i} \frac{1}{(\|p_j(t) - p_i(t)\|+1)^{\beta+1}}\vec u_{i \to j}(t)$$  
To ensure getting together in such cases, the agents should move in the positive direction of "influence field" gradients. Notice that in all the examples above, we obtain dynamics of the form:
$$\frac{dp_i(t)}{dt}=\alpha \sum_{j=1}^{n} \mathcal{F} (\|p_j(t)-p_i(t)\|)(p_j(t)-p_i(t))=\alpha \sum_{j=1}^{n} r \mathcal{F}(r)\vec u_{i \to j}(t)$$
where $\mathcal{F}$ is some scalar function.\\

As we shall see in the sequel, such coordinated motion keeps, due to symmetry, the average location of the agents (their geometric centroid) stationary. Furthermore, the required dynamics can be readily implemented provided each agent can sense the relative position of all the other agents with respect to itself, i.e. the vectors $(p_j(t)-p_i(t)) \forall j$.\\

It is also important to note that the dynamics discussed above is not the "physical" dynamics of a group of unit mass points or "particles" evolving due to their gravitational field. These would implement the following Newtonian motion rule

$$
\begin{array}{ll}
\begin{array}{l}
\end{array}
\left\{
\begin{array}{l}
\frac{dp_i(t)}{dt} = v_i(t) \\
\frac{dv_i(t)}{dt} =  \sum\limits_{j \neq i} G \frac{1}{\|p_j(t) -p_i(t)\|} \vec u_{i \to j}(t)
\end{array}
\right.
\end{array}
$$
which results in the following second order dynamics:

$$ \frac{d^2p_i(t)}{dt^2} = \sum\limits_{j \neq i} G \frac{1}{\|p_j(t) -p_i(t)\|^2}\vec u_{i \to j}(t)$$

While similar in spirit to the dynamics considered before, the unit mass points moving in their Newtonian gravitational field have their \textit{acceleration}, not their \textit{velocity} controlled by weighted sums of unit vectors pointing to their neighbours. This is a physical motion model taking into consideration the inertia of the agents and the physical impossibility to instantaneously change directions. In this report, we consider only the former, simpler dynamics, however we survey and analyze much more flexible local velocity controls that take into consideration all aspects of the geometry of the constellation of agents' neighbours. Such controls are often not expressible as results of simple "influence functions" and their dynamics cannot be written as weighted superpositions of unit vectors from each agent to its neighbours, with weights only depending on corresponding agent-to-agent distances alone.\\

\newpage
\textbf{First-order vs. Second-order models (a taxonomy)}\\

As we have seen in the previous section, a multi-agent system comprises identical anonymous and oblivious mobile agents acting according to some rules of motion determined by "what they see", i.e. by information about their neighbours (and sometimes about the environment) provided by their own sensors. We assume that the agents are located in the $\mathbb{R}^2$ plane, at points $\{p_1, p_2, ... , p_n\}$ and their geometric constellation evolves in time according to 
$$ \frac{dp_i(t)}{dt} = v_i(t) $$
where
$v_i(t)$ is determined by a function $\mathcal{F}_C$ and the constellation $\{p_1(t), p_2(t), ... , p_n(t)\}$ relative to $p_i(t)$ 
$$v_i(t)=\mathcal{F}_C\{p_1(t), p_2(t), ... , p_n(t)\, | \,p_i(t)\}$$
In a discrete-time setting the evolution will be
$$ p_i(k+1) = p_i(k) + \Delta p_i(k)$$
$$\Delta p_i(k)=\mathcal{F}_D\{p_1(k), p_2(k), ... , p_n(k)\, | \,p_i(k)\}$$

The various rules of motion are results of selecting the functions $\mathcal{F}$. If we write velocity vectors as follows (with respect to some global frame of reference)

$$ v(t) = \|v(t)\|[\cos \theta(t), \sin \theta(t)]$$
we have 
$$ \frac{dp_i(t)}{dt} = \frac{d}{dt} [x_i(t),y_i(t)] =  \|v_i(t)\|[\cos \theta_i(t), \sin \theta_i(t)]  $$
and we must specify

$$ \|v_i(t)\| = V\{p_1(t), p_2(t), ... , p_n(t)\,|\,p_i(t)\}$$
$$ \theta_i(t) = \Theta\{p_1(t), p_2(t), ... , p_n(t)\,|\,p_i(t)\}$$

Note that we require the functions $V\{*|p_i(t)\}$ and $\Theta\{*|p_i(t)\}$ to be fully determined by the instantaneous constellation of agents' locations with respect to (or as seen from the location) $p_i(t)$.\\

For the examples seen in the previous section we have

\begin{equation} \label{intro1}
\begin{array}{ll}
\frac{dp_i(t)}{dt} = \|v(t)\|[\cos \theta(t), \sin \theta(t)]^\intercal\\
\\
\left\{
\begin{array}{ll}
\|v_i(t)\| = \|\sum\limits_{j \neq i} \dot \Phi(\|p_j(t) -p_i(t)\|) \vec u_{i \to j}(t)\| \\
\\
\theta_i(t) = \; angle \; of \; \sum\limits_{j \neq i} \dot \Phi(\|p_j(t) -p_i(t)\|) \vec u_{i \to j}(t)
\end{array}
\right.
\end{array}
\end{equation} 

where $\Phi(r)$ is an influence field. In the case of Newtonian motion we had:

\begin{equation} \label{intro2}
\begin{array}{ll}
\frac{dp_i(t)}{dt} = \|v(t)\|[\cos \theta(t), \sin \theta(t)] ^\intercal\\
\\
\frac{dv_i(t)}{dt} = \sum\limits_{j \neq i}\mathcal{F}(p_j(t) -p_i(t)) \vec u_{i \to j}(t) \\
\end{array}
\end{equation} 
where $\mathcal{F}$ is some scalar function representing the "gradient of the interaction field".\\

In the examples above we need to clarify the issue of \textit{oblivion} in the behaviour of the agents. Memoryless-ness or obliviousness in the action of the agents means that their motion depends only on what they currently see from where they are. Furthermore we require their velocities to be determined without the need for a common frame of reference, and having rules of motion for $p_i$ depending only on the vectors $(p_j(t)-p_i(t))$ clearly satisfy this requirement.\\

Note however that, while the first rule above determines $\bar{v_i}$ solely based on the set of vectors $p_j(t)-p_i(t)$ (for agent $i$) the second rule requires the agent $i$ to also remember $\bar{v_i}$ at $(t-\epsilon)$ since it needs to integrate past velocities. Indeed 

$$\frac{dv_i(t)}{dt}=\mathcal{F} \{p_j(t)-p_i(t)\,|\, j \neq i\}$$
is in fact equivalent to 

$$v_i(t)= v_i(t-\epsilon) \; + \mathcal{F} \{p_j(t-\epsilon)-p_i(t-\epsilon)\, | \,j \neq i\}\cdot \epsilon$$

For this reason the rule of motion setting the velocity vector directly is memoryless and integration is only implicitly done by the changes in location of the agents while the rule of motion controlling the rate of change of the velocity vector also requires each agent to "remember" its past velocity. Hence in this case the agent implicitly performs a double integration of the velocity control to determine its location (and this implies the need to have a "state" that remembers its velocity too).\\

In the vast literature on multi-agent systems one encounters a wide variety of both single integration (i.e. velocity controlled or first order models) and second order (or double integration) models for the agents' motion. The second order models are the gravitational-like and unicycle models, and Vicsek and Reynolds type models, popular among physicists and engineers, while the first order models are usually considered by robotics researchers with background in geometry and computer science. We subsequently provide a chronologically ordered list classifying the papers mentioned in the bibliography of this report according to whether they deal with first order or second order models.\\


\begin{figure}[H]
\textbf{\underline{A Taxonomy of Multi-agent dynamic models}}\\

\textbf{{First-order models}}\\
Direct "velocity" control of the type $\dot p_i(t) = \mathcal{F}\{(p_j(t)-p_i(t))\,|\,p_i(t)\}$, or one-dimensional consensus models, are discussed in the following papers:\\
\begin{itemize}
\setlength\itemsep{-1ex}
\item Cybenco \cite{cybenko1989dynamic}
\item Mirollo and Strogatz \cite{mirollo1990synchronization} \cite{strogatz2000kuramoto}
\item Wagner, Bruckstein, et. al. \cite{bruckstein1991ants} \cite{bruckstein1993ant} \cite{wagner1997row} \cite{bruckstein1997probabilistic}
\item Oasa, Suzuki, Yamashita and Ando \cite{suzuki1999distributed} \cite{ando1999}
\item Cieliebak, Flocchini, Prencipe, Santoro and Widmayer \cite{cieliebak2003solving} \cite{flocchini2005gathering} \cite{cieliebak2012distributed}
\item Schlude \cite {schlude2003robotics} \cite{schlude2003point} 
\item Gazi and Passino \cite{gazi2003stability} \cite{gazi2004stability}
\item Xiao and Boyd \cite{xiao2004fast}
\item Moreau \cite{moreau2004} 
\item Gordon, Wagner, Elor and Bruckstein \cite{gordon2004} \cite{gordon2005} \cite{gordon2008}
\item Ren and Beard \cite{ren2005consensus}
\item Cohen, Agmon and Peleg \cite{cohen2005convergence} \cite{agmon2006fault}
\item Cort\'ez, Mart\'inez, D\"orfler and Bullo \cite{cortes2006robust} \cite{martinez2007motion} \cite{dorfler2014synchronization}
\item Olfati-Saber, Fax and Murray \cite{olfati2007consensus}
\item Ji and Egerstedt \cite{ji2007}
\item Olshevsky and Tsitsiklis \cite{olshevsky2009convergence}

\end{itemize}

\textbf{{Second-order models}}\\
Newtonian-like dynamics or Unicycle models of the type \newline $ \ddot p_i(t) = G\{(p_j(t)-p_i(t)),\dot p_i(t)\,|\,p_i(t)\}$, are studied among others in the papers:\\
\begin{itemize}
\setlength\itemsep{-1ex}
\item Reynolds \cite{reynolds1987flocks}
\item Reif and Wang \cite{reif1999social}
\item Jadbabaie, Lin and Morse \cite{jadbabaie2003}
\item Vicsek, Ben-Jacob, et. al. \cite{vicsek1995novel} \cite{vicsek2012collective}
\item Marshall, Broucke and Francis \cite{marshall2004formations}
\item Lin, Francis and Maggiore \cite{lin2005necessary}
\item F. Belkhouche and B. Belkhouche \cite{belkhouche2005modeling}
\item Hristu-Varsakelis and Shao \cite{hristu2007bio}
\item Wang and Slotine \cite{wang2005partial}
\item Cucker and Smale \cite{cucker2007emergent}
\item Motsch and Tadmor \cite{motsch2014heterophilious}
\item Chazelle \cite{chazelle2014convergence} \cite{chazelle2015algorithmic}  

\end{itemize}
\end{figure}

\newpage
In this report we survey the geometric consensus or the Gathering / Clustering / Aggregation / Rendezvous problem, for velocity controlled agents in the continuous (C) and discrete (D) time settings. The agents' sensing capabilities considered are either with full (F) or limited (L) visibility horizon, and measurements of the relative location to their neighbours will be either direction and distance, i.e. relative position (P) or direction, i.e. bearing only (B).\\

The action of the agents will be assumed synchronous (i.e. all agents are active at all times), except for the interesting cases where timing randomization is absolutely necessary to ensure convergence.\\

The three choices (C vs. D), (F vs. L), and (P vs. B) lead to the analysis of eight cases of gathering problems that provide a focused overview of multi-agent research over the years from the beginning to the forefront of research today.\\

\newpage
\addcontentsline{toc}{section}{The gathering problem}
\section*{The gathering problem}

A multi-agent system may be required to perform a variety of tasks. A basic task is gathering, i.e. using the sensing and motion capabilities to bring the agents together by designing suitable distributed interaction rules. This paper addresses the gathering problem in detail.\\

Consider a system of $n$ identical, anonymous, and memoryless agents in Euclidean space ($\mathbb{R}^d$) specified by their time varying locations $\{p_i(t)\}_{i=1,2,...,n}$. These agents interact with each other such that their motion ($\dot p_i(t)$ or $p_i(k+1)$) is determined by the constellation of their neighbours. The neighbour sets denoted $N_i(t)$, are determined by the agents' sensing range $V$, and the type of interaction between them is described by dynamic laws of the type:
$$\dot p_i(t) = \mathcal{F}\{p_j(t)\}_{p_j(t)\in N(p_i(t))}$$
$$p_i(k+1) = p_i(k) + \mathcal{F}\{p_j(k)\}_{p_j(k)\in N(p_i(k))}$$

In this paper we discuss four types of systems, which differ in the agents' sensing range and capabilities.

\begin{itemize}

\item \textbf{VISIBILITY:}\\ We consider two cases. \textit{\textbf{Full visibility}}, if the sensing range of the agents is unlimited ($V=\infty$), and each agent sees all the other agents, the graph representing the visibility between the agents being a complete graph. \textit{\textbf{Limited visibility}} if the sensing range of the agents is limited, so that each agent senses only those agents within its sensing range at any time, and the topology of the visibility graph changes according to the position of the agents.

\item  \textbf{SENSING:}\\ We call \textit{\textbf{position sensing}} the case when each agent can determine the (relative) position to all other agents within its visibility range (i.e. relative distance and relative direction), and \textit{\textbf{bearing-only sensing}} if agents sense only the direction towards their neighbours. 

\end{itemize}

Each of the above cases is discussed in two variants of temporal evolution:  \textit{\textbf{continuous-time}}  and \textit{\textbf{discrete-time}}, hence we discuss gathering of multi-agent systems for eight different cases.

\addcontentsline{toc}{subsection}{A Preliminary Observation}
\subsection*{A Preliminary Observation}
We begin our analysis with a lemma (see e.g. Gazi and Passino \cite{gazi2003stability}) showing that the average position is a system invariant in case the agents' decentralised dynamics is governed by an antisymmetric pairwise interaction function. Indeed let $f$ be a function such that 
$$\forall p_i, p_j;\  \ f(p_i(t)-p_j(t))=-f(p_j(t)-p_i(t))$$ 
Then we have the following:

\begin{lemma}\label{invariant}

Let the dynamics of a multi-agent system be 
\begin{equation}
\dot p_i(t)=\sum_{j=1}^{n}f(p_i(t)-p_j(t)) 
\label{eq:DynamicsGeneralCont}
\end{equation}
or, in discrete time,
\begin{equation}
p_i(k+1)=p_i(k)+\sum_{j=1}^{n}f(p_i(k)-p_j(k))
\label{eq:DynamicsGeneralDisc}
\end{equation}
where $f$ is the inter-agent interaction function.\\ 

If $f$ is antisymmetric, the average position denoted $\bar p$ is a system invariant i.e. $\bar p = \bar p(0) = const. $

\end{lemma}

\begin{proof}
The average position of the system is: 
 $$\bar p(t)=\frac{1}{n}\sum_{i=1}^{n}p_i(t)$$
We have clearly,
$$ \dot{\bar p}(t)=\frac{1}{n}\sum_{i=1}^{n}\dot p_i(t)=\frac{1}{n}\sum_{i=1}^{n}\sum_{j=1}^{n}f(p_i(t)-p_j(t))=0$$
since in the above summation $f(p_i(t)-p_j(t))$ cancels $f(p_j(t)-p_i(t))$ , due to the fact that $f$ is  antisymmetric.\\\\
Therefore 
$$\dot{\bar{p}}(t)=0 \implies \bar p(t) = \bar p = const. = \bar p(0)$$
and obviously the same holds for discrete time,
$$\bar p(k+1)=\bar p(k) \implies \bar p(k)= \bar p = const.=\bar p(0)$$

\end{proof}

\newpage
\section{Unlimited Visibility, Position Sensing}

Throughout this section we assume that each agent senses the relative position of all the other agents, i.e. it has unlimited visibility and the capability to measure distance and direction to all its neighbours.

\subsection{Continuous Time Dynamics (system $\mathcal{S}_1$)}

Suppose that each agent moves according to the following dynamic law:

\begin{equation}
\dot{p_i}(t)=-\sigma \sum_{j=1}^{n} (p_i (t)-p_j (t))
\label{eq:Dynamics1}
\end{equation}
where $\sigma$ is a constant positive scalar gain factor, i.e. each agent continuously moves with a velocity proportional to the sum of the vectors pointing to the positions of all other agents.\\

Since dynamics (\ref{eq:Dynamics1}) is governed by a trivially antisymmetric function, the average position of the agents in $\mathcal{S}_1$ is invariant, i.e. $\bar p=const$ (Lemma $\ref{invariant}$).

\begin{theorem} \label{LinearTheorem}
For an arbitrary initial constellation, all the agents of system $\mathcal{S}_1$ asymptotically converge to the average location of the initial constellation.
\end{theorem}

\begin{proof}
Without loss of generality, consider the positions $p_i(t)$ in a global coordinate system centered at $\bar p$, i.e. take $\bar p=0$ (a system invariant). Hence we have:
$$\dot{p_i}(t)=-\sigma \sum_{j=1}^{n} p_i (t) + \sigma \sum_{j=1}^{n}p_j (t)=-\sigma n p_i (t)$$\\
which yields

\begin{equation} 
p_i(t)=p_i (0)e^{-\sigma n t}
\label{eq:Dynamics4}
\end{equation}
i.e. all agents of system $\mathcal{S}_1$ exponentially asymptotically converge to the average position at $t=0$.
$$
 \forall{i};\  \  \lim_{t \to \infty} p_i(t) = \bar p = 0
$$

\end{proof}


\subsection{Discrete Time Dynamics (system $\mathcal{S}_2$)}

Next we assume each agent moves according to a descretized dynamic law:

\begin{equation}
p_i(k+1)=p_i(k)-\sigma\sum_{j=1}^{n}(p_i(k)-p_j(k))
\label{eq:Dynamics2_1}
\end{equation}
where $\sigma$ is a constant positive scalar gain factor, i.e. at each time-step, each agent jumps proportionally to the sum of relative position vectors to all the other agents.\newline

Dynamics (\ref{eq:Dynamics2_1}) is again an antisymmetric function, hence the average position of the agents in $\mathcal{S}_2$ is also invariant, i.e. $\bar p=const$, due to Lemma $\ref{invariant}$.

\begin{theorem} \label{LinearTheorem2}
For an arbitrary initial constellation, if $0 < \sigma < \frac{2}{n}$ all the agents of system $\mathcal{S}_2$ asymptotically converge to the average location of the initial constellation.
\end{theorem}

\begin{proof}
Without loss of generality, $\bar p$ being a system invariant, consider the positions $p_i(k)$ in a global coordinate system centered at $\bar p=0$, hence:
\begin{equation}\label{S2_dev}
p_i(k+1)=p_i(k)-\sigma\sum_{j=1}^{n}p_i(k)+\sigma\sum_{j=1}^{n}p_j(k) =p_i(k)-\sigma n p_i (k) 
\end{equation}
which yields,

\begin{equation}
p_i(k)=(1-\sigma n)^k p_i(0)
\label{eq:Dynamics2_3}
\end{equation}
Therefore 
 
\begin{equation}
 \lim_{k \to \infty} p_i(k)= \left\{\begin{alignedat}{3}
    & p_i(0)  &&\quad  \sigma =0 \\
    & \bar p(0)=0 &&\quad 0<\sigma<\frac{2}{n} \\
    & (-1)^k p_i(0) &&\quad \sigma =\frac{2}{n}, &&\quad (oscillation)\\
    & divergence &&\quad o.w.
  \end{alignedat}\right.
\label{eq:Dynamics2lim}
\end{equation}
\newline

Hence, agents of system $\mathcal{S}_2$ asymptotically converge to the average position of the initial constellation if $0 < \sigma < \frac{2}{n}$, i.e.

$$ \forall{i};\quad 0<\sigma <\frac{2}{n};  \  \  \lim_{k \to \infty} p_i(k)=\bar p = 0 $$

\end{proof}


\subsection{Discussion and Generalisations}

So far we discussed the gathering of $n$ agents with unlimited visibility. The postulated dynamics was linear and quite straightforward to analyze.\\

\subsubsection{System scalability} Notice that system $\mathcal{S}_1$ is fully scalable in the sense of convergence guarantee. Adding or removing agents affects the convergence \textit{rate}, but convergence is always guaranteed. Contrary to the system $\mathcal{S}_1$ (which is in continuous time), system $\mathcal{S}_2$ (in discrete time) is no longer scalable in the sense of guaranteed convergence. While the number of agents in system $\mathcal{S}_1$ affects only the convergence \textit{speed}, it is easy to notice in (\ref{eq:Dynamics2lim}) that for a fixed $\sigma$, adding more agents to system $\mathcal{S}_2$ may cause the agents to disperse. One way to overcome this scalability limitation is to give each agent the ability to also count the number of its neighbours, and hence the ability to also calculate the \textit{average} position of its neighbours. Then, if each agent of $\mathcal{S}_2$ adjusts its dynamics to
$$p_i(k+1)=p_i(k)-\frac{1}{n}\sigma\sum_{j=1}^{n}(p_i(k)-p_j(k)),$$
the system will be scalable with guaranteed convergence for $0<\sigma<2$, no matter what the number of agents is. In particular if $\sigma = 1$, all agents will jump to the system's  average location in one step.\\ 

\subsubsection{Travelling path}
Since the agents of system $\mathcal{S}_1$ and $\mathcal{S}_2$ do not have memory, each of them continuously calculates its motion based on the current relative location of all other agents. However, as we see in the proof of Theorem \ref{LinearTheorem} equation (\ref{eq:Dynamics4}) and Theorem \ref{LinearTheorem2} equation (\ref{eq:Dynamics2_3}), each agent moves on a \textit{straight line} from its initial position $p_i(0)$ to $\bar p$, and the travelling speed of the agents decreases as they approach $\bar p$ (see Figure \ref{straight line} for simulation results). If the agents had the ability to compute and \textit{remember} $\bar p$, they could have travelled there at a fixed velocity and gather there in finite time. Note however, that, while we assume our agents have quite extensive computational capabilities based on what they "see" at each moment, we consider them to be \textbf{oblivious}, or memory-less in the sense that they cannot recall the past at all. Hence they cannot remember $\bar p$ and move toward it.\\

\begin{figure}[H]
\captionsetup{width=0.8\textwidth}
  \centering
    \includegraphics[width=106mm]{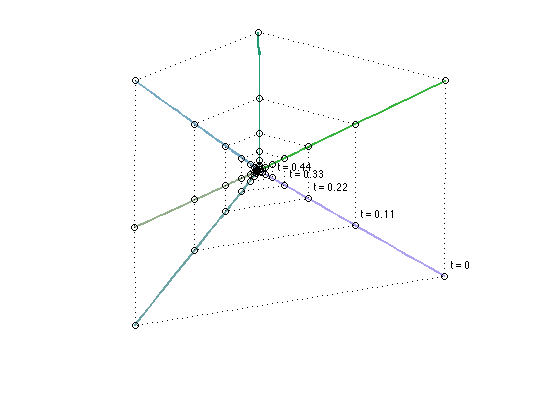}
    \caption{A typical simulation result for system $\mathcal{S}_2$ having 6 agents.}
      \label{straight line}
\end{figure}

\subsubsection{Linear Systems}
Systems $\mathcal{S}_1$ and $\mathcal{S}_2$ are simple, having linear dynamics, and they belong to a global family of systems whose properties are well known. The dynamics of such systems is briefly reviewed below in a general setting in both continuous time and discrete time. The dynamics of general linear systems in $\mathbb{R}^d$ is given by:
\begin{equation} \label{GlobalLinearDynamics}
	\begin{array}{ll}
		
		\dot{P}(t) = -\sigma LP(t) & :CT \\
		P(k+1) = (I_n - \sigma L)P(k) & :DT
	
	\end{array}
\end{equation}
where:
\begin{itemize}
\item $P(t) = [p_1(t),...,p_n(t)]^\intercal$ (with $p_i(t) = [p_i^1(t),...,p_i^d(t)]^\intercal$) is called the state vector
\item $n$ is the number of agents
\item $d$ is the space dimension
\item $L$ is an $(n \times n)$ matrix that determines the system dynamics
\end{itemize}
In particular, in the linear multi-agent systems setting, we focus on system dynamics matrices $L$ of the form $L(G)$, the graph Laplacian matrix, where the graph $G$ represents the interdependency between the agents. This matrix has some special properties which are extensively exploited, as we shall see next.\\

Interactions between agents in multi-agent systems are often mathematically described using a graph, commonly labelled as $G(\mathcal{V},\mathcal{E})$, where $\mathcal{V}=\{\nu_1, \nu_2,...,\nu_n\}$ is the set of vertices (representing the agents), and $\mathcal{E} \subseteq \mathcal{V} \times \mathcal{V}$ is the set of edges (representing connections between the agents). The neighbourhood set of a vertex is the set of vertices connected to it, i.e.  $N(\nu_i)=\{\nu_j\in \mathcal{V}\; |\; \{\nu_i,\nu_j \} \in \mathcal{E} \}$, or in short $N_i$.\\

For undirected graphs, the adjacency matrix $A(G)$ is a symmetric matrix encoding the connection between vertices so that $A_{ij}=1$ if $\epsilon_{ij} \in \mathcal{E}$ and $A_{ij}=0$ otherwise. The degree matrix $\Delta(G)$ is a diagonal matrix where $\Delta_{ij}=0$ for all $i \neq j$, and $\Delta_{ii} = deg(v_i)=|N(v_i)|$. The Laplacian matrix $L(G)$ is defined as $L(G)=\Delta(G)-A(G)$, and has the following properties:\\

\begin{enumerate}
\item $L(G)$ is symmetric and positive semi-definite, hence all its eigen-values $\lambda_i$ are real and non negative.
\item The vector $\vec{1} = [1, 1, ..., 1]^\intercal$ is in the null space of $L(G)$.
\item For connected graphs the null space of $L(G)$ is one dimensional and spanned by the vector $\vec{1}$
$$Span(null(L(G))) = \vec{1}$$
\end{enumerate}

Any symmetric positive semi-definite matrix $L(G)$ may be expressed using the eigen-decomposition:
$$L (G)= U\Lambda U^* = \sum\limits_{i=1}^n \lambda_iU_iU_i^* $$
where $U = [U_1, U_2, ..., U_n]^\intercal$ is a unitary (complex) matrix whose columns comprise the eigenvectors of $L(G)$, corresponding to the eigenvalues $\lambda_i$ - an orthonormal basis for $\mathbb{C}^n$.\\

Therefore, the dynamics (\ref{GlobalLinearDynamics}), where the matrix $L$ is the graph Laplacian $L(G)$, may be written as follows:\\

\begin{equation} \label{GlobalLinearDynamics_Dev}
	\begin{array}{ll}
		\dot{P}(t) = -\sigma LP(t) = -\sigma  \sum\limits_{i=1}^n \lambda_i U_i U_i^* P(t)  & : CT \\
		P(k+1) = (I_n - \sigma L)P(k) =  \sum\limits_{i=1}^n (1-\sigma \lambda_i) U_i U_i^* P(k) & : DT
	\end{array}
\end{equation} 
 
Since $U_i^*U_j = \delta_{ij}$, the trajectories of (\ref{GlobalLinearDynamics_Dev}) are given by:
\begin{equation} \label{GlobalLinearTrajectories}
	\begin{array}{ll}
		
		P(t) = \sum\limits_{i=1}^n e^{-\sigma \lambda_i t} \; U_i U_i^* P(0) & :CT \\
		P(k) =  \sum\limits_{i=1}^n (1-\sigma\lambda_i)^k \; U_i U_i^* P(0)  & :DT
	
	\end{array}
\end{equation}

Next, we analyze this solution using the other properties of $L(G)$.\\

Since we consider that the interconnection topology is represented by a connected graph,  the null space of $L(G)$ is spanned only by the unit vector $\frac{1}{\sqrt{n}} \vec 1$, the eigenvector corresponding to eigenvalue $\lambda_1=0$. Therefore, in (\ref{GlobalLinearTrajectories}), in the continuous time system, all terms associated to the non-zero eigenvalues vanish exponentially, hence:
$$ \lim\limits_{t \to \infty} P(t) = \frac{1}{n} \vec{1}\vec{1}^\intercal P(0) =  \vec{1} \frac{1}{n}\sum\limits_{i=1}^n p_i(0)  =  \vec{1} \frac{1}{n} \sum\limits_{i=1}^{n} p_i(t)$$
and in discrete time, if $| 1 - \sigma \max\limits_i(\lambda_i) | < 1$  then also
$$	\lim\limits_{k \to \infty} P(k) =  \frac{1}{n}\vec{1}\vec{1}^\intercal P(0) = \vec{1} \frac{1}{n} \sum\limits_{i=1}^np_i(0)      =  \vec{1}\frac{1}{n} \sum\limits_{i=1}^{n} p_i(k) $$

From (\ref{GlobalLinearTrajectories}) we see that $\frac{d}{dt}{(U_1^*P(t))} = \frac{d}{dt}{(U_1P(0))}=0$, and therefore the quantity ${U_1^*P(t)}=\frac{1}{\sqrt{n}} \sum\limits_{i=1}^n p_i(0)$ is invariant under the dynamics given by (\ref{GlobalLinearDynamics}) and this is a realisation of Lemma \ref{invariant} for the general case of linear systems with symmetric inter-agent interaction.\\

An example of the trajectories of six agents with dynamics described by a fixed, connected, but not complete graph, as the agents move towards the average location of their initial positions is  presented in Figure \ref{curved line}. Note that for general topologies the trajectories are not necessarily straight lines as in the complete graph case.\\

Figure \ref{curved line directed} presents simulation results for a linear system describing multi-agent interactions, which are represented by a directed graph  with a Laplacian matrix that is not symmetric. The example given, corresponds to the so called "Linear Cyclic Pursuit" example (in the discrete time case) as discussed in \cite{bruckstein1991ants}. Note that here too the average location of the initial positions is invariant, however, this property is not a result of a realisation of Lemma \ref{invariant}, but rather of the fact that in this case the Laplacian matrix, which is Toeplitz and circulant, has both row sums and column sums equal to zero. Hence the vectors $\vec{1}$ and $\vec{1}^\intercal$ are left and right eigenvectors corresponding to their zero eigenvalues.\\

\newpage
\begin{figure}[H]
\captionsetup{width=0.8\textwidth}
  \centering
	\includegraphics[width=100mm]{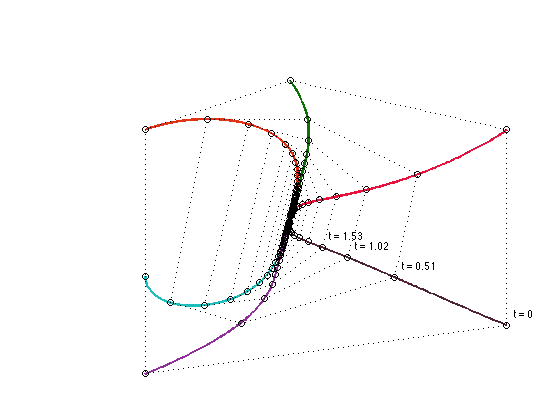}
	  \caption{Simulation results for a discrete time linear system composed of 6 agents with a connected interconnection graph and fixed topology.}
    \label{curved line}
\end{figure}
\begin{figure}[H]
\captionsetup{width=0.8\textwidth}
  \centering
	\includegraphics[width=100mm]{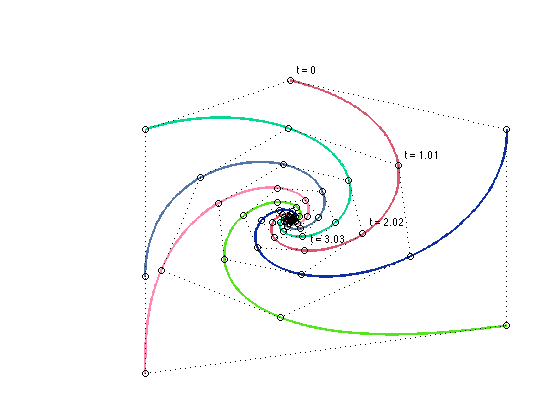}
	  \caption{Simulation results of linear cyclic pursuit of a discrete time linear system with directed, balanced and cyclic interconnection graph.}
    \label{curved line directed}
\end{figure}

\subsubsection{Convergence analysis with Lyapunov functions} A function is called Lyapunov if it maps the state of the system to a non-negative value in such a way that the system dynamics causes a monotonic decrease of this value. If the Lyapunov function reaches zero only at desirable states of the system and we prove that the dynamics leads the Lyapunov function to zero, we can argue that the system converges to a desirable state.\\

As a first example of using a Lyapunov function, let us prove the convergence of a multi-agent system to the average location of the initial constellation of the agents with the dynamics given by
\begin{equation} \label{LapllacianLinearDynamics}
	\dot{P}(t) = -\sigma L(G)P(t) 
\end{equation}
As a  Lyapunov function $\mathcal{L}_1$ let us first choose the sum of squared distances of the agents from the invariant average location as follows:
\begin{equation}\label{SquareDistancesFromMean}
	\mathcal{L}_1(P(t)) = \sum\limits_{i=1}^n \|p_i(t) - \bar{p}\|^2
\end{equation}\
   
By Lemma \ref{invariant} the average position of the agents in the system described by dynamics (\ref{LapllacianLinearDynamics}) is invariant. Without loss of generality, let us set this invariant to be $\bar{p}(t) = [0,...,0]^\intercal $. It is easy to see that under the assumed dynamics (\ref{LapllacianLinearDynamics}), we have 
 $$ \frac{d}{dt} \mathcal{L}_1(P(t)) = \sum\limits_{i=1}^n(\nabla_{p_i} \mathcal{L}_1(P(t)))^\intercal \dot{p}_i(t) = 2 \sum\limits_{i=1}^n {p_i}^\intercal(t) \dot{p}_i(t) =$$
 $$= -2 \sigma \sum\limits_{i=1}^n {p_i}^\intercal(t) \sum\limits_{j \in N_i}(p_i(t) - p_j(t)) = -2 \sigma \sum\limits_{i=1}^n\sum\limits_{j \in N_i}(\|p_i(t)\|^2 - {p_i}^\intercal(t)p_j(t)) =$$
 $$= - \sigma \left[  \sum\limits_{i=1}^n\sum\limits_{j=1}^n A_{ij}(\|p_i(t)\|^2 - {p_i}^\intercal(t)p_j(t)) +  \sum\limits_{i=1}^n\sum\limits_{j=1}^n A_{ij}(\|p_i(t)\|^2 - {p_i}^\intercal(t)p_j(t))\right] = $$
 $$= - \sigma \left[  \sum\limits_{i=1}^n\sum\limits_{j=1}^n A_{ij}(\|p_i(t)\|^2 - {p_i}^\intercal(t)p_j(t)) +  \sum\limits_{j=1}^n\sum\limits_{i=1}^n A_{ji}(\|p_j(t)\|^2 - {p_i}^\intercal(t)p_j(t))\right]$$
Recall that $A$ is the Adjacency matrix, which in the undirected graph case like ours is symmetric, i.e $A_{ij} = A_{ji}$. Notice, that the indices $i$ and $j$ were switched in the last term, and since we can rewrite $\sum\limits_{j=1}^n\sum\limits_{i=1}^n$ as $\sum\limits_{i=1}^n\sum\limits_{j=1}^n$, we have that
$$\frac{d}{dt} \mathcal{L}_1(P(t)) =$$
 $$= - \sigma \left[  \sum\limits_{i=1}^n\sum\limits_{j=1}^n A_{ij}(\|p_i(t)\|^2 - {p_i}^\intercal(t)p_j(t)) +  \sum\limits_{i=1}^n\sum\limits_{j=1}^n A_{ij}(\|p_j(t)\|^2 - {p_i}^\intercal(t)p_j(t))\right] = $$
 $$ = -\sigma \sum\limits_{i=1}^n\sum\limits_{j=1}^n A_{ij}\left( \|p_i(t)\|^2 - 2{p_i}^\intercal(t)p_j(t) + \|p_j(t)\|^2 \right) = - \sigma \sum\limits_{i=1}^n\sum\limits_{j \in N_i}^n\| p_i(t) - p_j(t) \|^2$$ 
 
Therefore, for a connected interaction topology of $G(\mathcal{V},\mathcal{E})$, the time derivative of $\mathcal{L}_1(P(t))$ equals zero if and only if all agents are at the same position, otherwise the value of the time derivative is strictly negative. Hence, the system asymptotically converges to the invariant location $\bar{p}(t) = [0,...,0]^\intercal$ as claimed.\\

\subsubsection{"Potential-like" Lyapunov function}
Another approach for convergence analysis is choosing a function $\mathcal{L}_2(P(t))$ whose gradient descent yields the dynamics of the system, if possible.\\\\
The gradient descent for $\mathcal{L}_2(P(t))$ is given by

$$ \dot{P}(t) = -\sigma\nabla_{P}\mathcal{L}_2(P(t)) $$\\
which yields the following dynamics for $\mathcal{L}_2(P(t))$:

$$ \frac{d}{dt} \mathcal{L}_2(P(t)) = \nabla_P\mathcal{L}_2(P(t))^\intercal \dot{P}(t) = -\sigma\|\nabla_P \mathcal{L}_2(P(t)) \|^2 $$\\
For example, the continuous time  dynamic law: 
$$ \dot{p_i}(t)= - \sigma\sum_{j \in N_i} (p_i (t)-p_j (t))  $$
is the gradient descent of the function
\begin{equation} \label{PotentialFunc}
	\mathcal{L}_2(P(t)) = \frac{\sigma}{2} \sum\limits_{i=1}^n \sum\limits_{j \in N_i}^n \|p_i(t) -p_j(t)\|^2
\end{equation}

Therefore, if the topology of the agents' interconnection graph is connected, a linear system described by (\ref{GlobalLinearDynamics}) asymptotically converges to a point, and the time derivative of function (\ref{PotentialFunc}) reaching zero if and only if all agents are collocated, otherwise the value of the time derivative is strictly negative.\\

This method may also be used to analyze a wider family of systems that have gradient descent dynamics corresponding to "potential functions" of the form:

\begin{equation} \label{PotentialFunction}
	\mathcal{L}_2^\alpha(P(t)) = \frac{\sigma}{\alpha} \sum\limits_{i=1}^n \sum\limits_{j \in N_i}^n \|p_i(t) - p_j(t)\|^{\alpha}
\end{equation}
where $\alpha$ is a positive scalar.\\

Considering the Lyapunov potential function as above, we are led to gradient descents of the form:

\begin{equation}
\dot{p_i}(t)= -\nabla_{p_i}\mathcal{L}_2^\alpha(P(t)) = -\sigma \sum_{j \in N_i}^n \frac{p_i (t)-p_j (t)}{\|p_i (t)-p_j (t)\|^{2-\alpha}}
\label{eq:GeneralDynamics}
\end{equation}
where for example, $\alpha=2$ recovers linear dynamics as in system $\mathcal{S}_1$. Note however that with $\alpha=1$ we have a dynamical system that needs only the bearing information from an agent to all its neighbours (and we shall analyze such systems in the sequel as systems $\mathcal{S}_3$ and $\mathcal{S}_7$).\\

\subsubsection{Convex-hull based Lyapunov functions}
For systems with agents in the plain, analyzing the dynamics of the convex-hull is another useful way to determine convergence. Let $\mathcal{L}_3(P(t))$ be the perimeter of $CH(P(t))$, the convex-hull of the agents' positions, and $l_i(t)$ be the length of the convex-hull edge connecting corners $p_i(t)$ and $p_{i+1}(t)$. Let $\varphi_i(t)$ be the internal angle of the $i$ th corner of $CH(P(t))$, let $\alpha_i(t)$ be the direction of motion of the agent located at corner $i$ (relative to its direction to corner $i+1$), and let $v_i(t)$ be the speed of that agent (as shown in Figure \ref{CHshrink_S1}).\\
\begin{figure}[h!]
\captionsetup{width=0.8\textwidth}
  \centering
        \includegraphics[width=100mm]{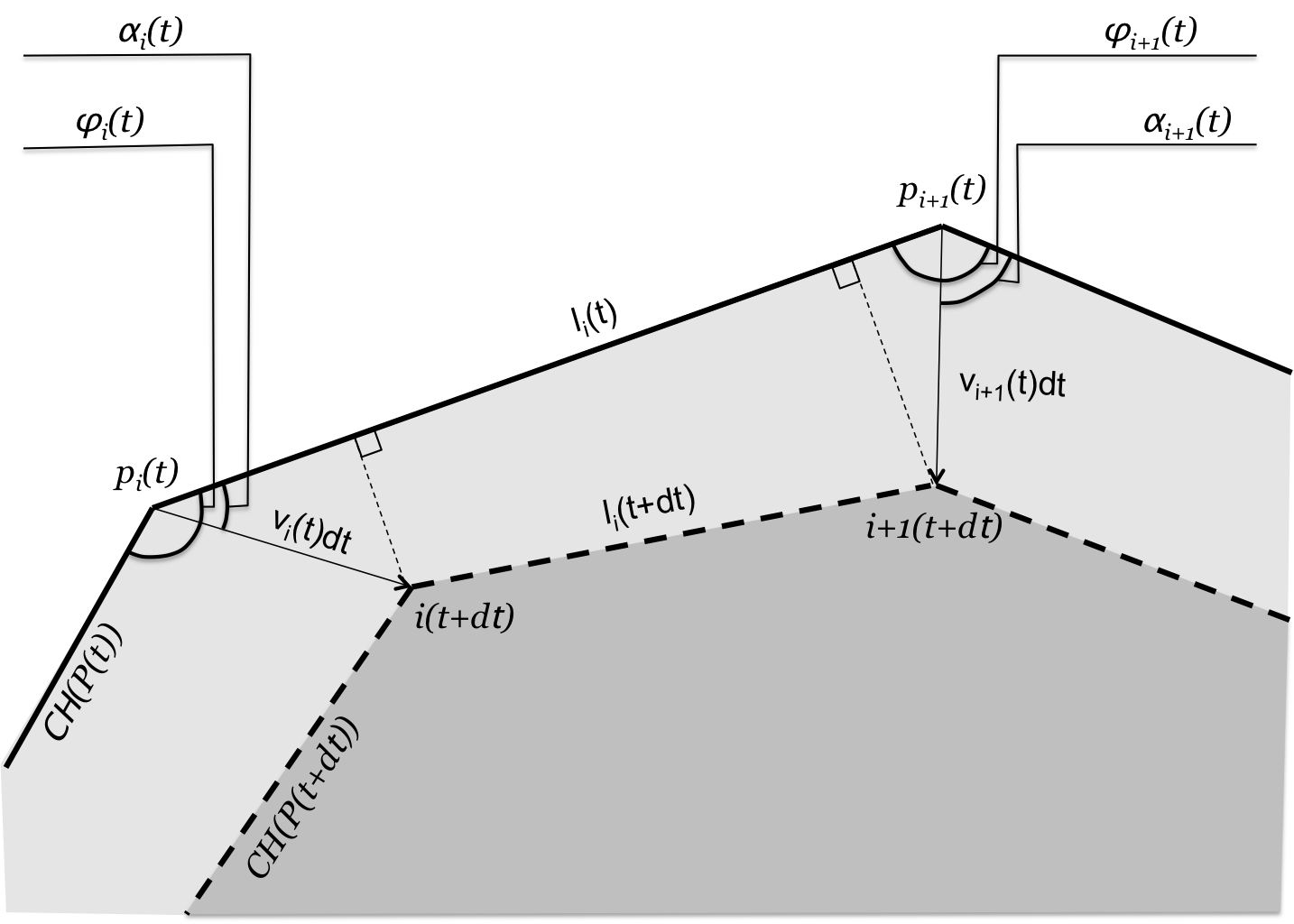}
     \caption{Convex-hull shrinkage}
     \label{CHshrink_S1}
\end{figure}

Geometric considerations then immediately yield: 
$$l_i(t+dt)-l_i(t) = -(v_i(t) \cos \alpha_i(t) + v_{i+1}(t)\cos (\varphi_{i+1} (t) -\alpha_{i+1}(t))) dt + \mathcal{O}(dt)$$
hence we have by re-ordering terms,
\begin{equation}
\begin{array}{l} \label{CH_Conv}
\dot{\mathcal{L}_3}(P(t)) =-\sum\limits_{i=1}^n v_i(t)(\cos \alpha_i(t)+\cos (\varphi_{i} (t) -\alpha_{i}(t))) \\
= -2\sum\limits_{i=1}^n v_i(t) \cos(\frac{\varphi_{i} (t)}{2}) \cos(\frac{\varphi_{i}(t) - 2\alpha_{i}(t)}{2}) 
\end{array}
\end{equation} 
However, we have $0 \leq \varphi_i(t) \leq \pi$ and $0 \leq \alpha_{i}(t) \leq \varphi_i(t)$. Furthermore, since $v_i(t) \ge 0$, we have that each element of the above sum is non-negative, and therefore $\dot{\mathcal{L}_3}(P(t)) \le 0$, i.e. the perimeter of the convex-hull of the agents positions never increases.\\

A system converges to a point if the following conditions are satisfied:
\begin{equation}
\begin{array}{l}
		
		\mathcal{L}_3(P(t)) > 0 \iff \dot{\mathcal{L}_3}(P(t)) < 0 \\
		and \\
		\mathcal{L}_3(P(t)) = 0 \iff \dot{\mathcal{L}_3}(P(t)) = 0
	
	\end{array}
\end{equation} 

In cases $\dot{\mathcal{L}_3}(P(t))$ is always negative and bounded away (by a constant) from zero, the system converges to a point in finite time. However, if $\dot{\mathcal{L}_3}(P(t))$ is always negative but its absolute value decreases to zero as the perimeter of the convex-hull decreases to zero, the system converges asymptotically.\\

Let us analyze system $\mathcal{S}_1$. For any positive $\Delta t$, the system's convex-hull at time $t+\Delta t$ is contained in its convex-hull at time $t$, since by motion law (\ref{eq:Dynamics1}) each agent on the perimeter of the convex-hull continuously moves to the sum of the vectors pointing to the positions of all other agents. These agents are clearly located inside or along the convex-hull boundary, hence each agent on the perimeter of the convex-hull can never move outside of the convex-hull.\\

The sum of angles of any convex polygon is $\pi(m-2)$, therefore every convex polygon of $m$ corners has at least one corner whose associated angle is smaller than or equal to $\pi(1-2/m)$. Since the convex-hull of system $\mathcal{S}_1$ is a convex polygon, and the maximal number of corners it may have is $n$, we shall denote the maximal possible value of its sharpest angle by $\varphi_* = \pi(1-2/n)$.\\

Let $s$ be the agent located at the sharpest corner of the convex-hull, let $\vec U_{\varphi_s}$ be a unit vector in the direction of the bisector of that corner, and let $\theta_{sj}$ be the angle between $\vec U_{\varphi_s}$ and the vector pointing from agent $s$ to some agent $j$. Then, obviously for all $j$ we have $0 \le \theta_{sj} \le \varphi_*/2$, and the velocity of agent $s$ is bounded from below as follows:

$$ v_s(t) = \|\dot p_s(t)\| \ge \vec U_{\varphi_s}^\intercal \dot p_s(t)$$
since the projection of a vector on a unit vector has a smaller or same length, and\\

$$ \vec U_{\varphi_s}^\intercal \dot p_s(t) = \vec U_{\varphi_s}^\intercal \sum\limits_{j \neq s} \left( p_j(t) - p_s(t) \right) = \sum\limits_{j \neq s} \| p_j(t) - p_s(t) \|\cos(\theta_{sj}) \ge $$
$$ \ge \sum\limits_{j \neq s} \| p_j(t) - p_s(t) \|\cos(\varphi_*/2) $$
since $0 \le \theta_{sj} \le \varphi_*/2 \leq \pi/2$\\

Furthermore, agent $s$ may either be one of the agents which defines $D(P(t))$, the diameter of the convex-hull, or there are two other agents that define $D(P(t))$. By the triangle inequality the distance from agent $s$ to at least one of those two agents defining the convex-hull diameter is greater than or equal to $D(P(t))/2$. Denote this agent by $q$, and we have that
$$v_s(t) \ge \sum\limits_{j \neq s} \| p_j(t) - p_s(t) \|\cos(\varphi_*/2) \ge \| p_q(t) - p_s(t) \|\cos(\varphi_*/2) \ge \frac{D(P(t))}{2}\cos(\varphi_*/2)$$

Using (\ref{CH_Conv}) we have that: 
$$\dot{\mathcal{L}_3}(P(t)) \le -2\sum\limits_{i=1}^n v_i(t) \cos(\frac{\varphi_{i} (t)}{2}) \cos(\frac{\varphi_{i}(t) - 2\alpha_{i}(t)}{2})  \le $$
$$ \le -2v_s(t) \cos(\frac{\varphi_{s} (t)}{2}) \cos(\frac{\varphi_{s}(t) - 2\alpha_{s}(t)}{2})$$ 
Given that $\varphi_{s}(t) \leq \varphi_*  < \pi$, and that $\alpha_{s}(t) \leq \varphi_{s}(t)$, we have that 
$$-\pi < -\varphi_{s}(t) \leq \varphi_{s}(t) - 2\alpha_{s}(t) \leq \varphi_{s}(t) < \pi$$ 
and we obtain
$$\dot{\mathcal{L}_3}(P(t)) \le -2 v_s(t) \cos^2(\frac{\varphi_*}{2}) \le - D(P(t))\cos^3(\frac{\varphi_*}{2}) $$

To analyze the relation between $\dot{\mathcal{L}_3}(P(t))$ and $\mathcal{L}_3(P(t))$ let us consider the connection between the diameter of a convex-hull and its perimeter. If a convex-hull has a diameter $D(P(t))$, the convex-hull is confined to the intersection area of two discs of radius $D(P(t))$ centered at the end points of the diameter. This sets a rough (but good enough) upper bound for any convex-hull perimeter as follows: 
$$\mathcal{L}_3(P(t)) < \frac{4}{3}\pi D(P(t))$$
and we have that

$$\dot{\mathcal{L}_3}(P(t)) \le - D(P(t))\cos^3(\frac{\varphi_*}{2}) \le -\frac{3\mathcal{L}_3(P(t))}{4\pi}\cos^3(\frac{\varphi_*}{2})$$

Since $\mathcal{L}_3(P(t))$, is non-negative, and all the other entries in the right side of above inequality are positive and constant, we have that  $\dot{\mathcal{L}_3}(P(t))$ is always negative and equals zero only when $\mathcal{L}_3(P(t))$ equals zero. This relation between the rate of decrease of the convex-hull perimeter and the length of the convex-hull perimeter shows that the agents of system $\mathcal{S}_1$ converge to a point, and the convergence rate is at least exponential. Note that similar arguments in the discrete case are possible, however they require a more delicate analysis as we shall see in the sequel.\\

\subsubsection{A detour on motion synchronisation}
This report focuses on synchronized systems, where all agents move together, i.e. in synchrony at all time. Another motion modality we shall encounter later is the semi-synchronised model, where agents in the system are not necessarily active in each cycle (i.e. there is a probability that an agent will sleep and stay put for a while, rather than move according to the dynamic rule). The analysis of systems implementing such semi-synchronised motions requires mathematical tools from stochastic processes in probability theory.\\

For example, a semi-synchronised model based on system $\mathcal{S}_2$ could be defined as having the following dynamics:

\begin{equation}\label{UpdatedS_2}
	p_i(k+1)=p_i(k)-\sigma\chi_{i}(k)\sum_{j=1}^{n}(p_i(k)-p_j(k))
\end{equation}
where the variable $\chi_{i}(k)=\{0,1\}$ has a random boolean value at each time-step, and the probability that $\chi_{i}(k) = 1$ is $\rho \ge 1/2$. We here assume that $0<\sigma<\frac{1}{n}$.\\

In order to prove convergence of such a system, we use  the following Lyapunov function:
$$ \mathcal L(P(k)) = \sum\limits_{i=1}^n \| p_i(k) - \bar p(k) \| $$
At the next time-step this function will be
$$ \mathcal L(P(k+1)) = \sum\limits_{i=1}^n \| p_i(k+1) - \bar p(k+1) \| $$
Note that due to randomization, Lemma \ref{invariant} does not hold and $\bar{p}(k)$ may change as the system evolves.\\

Without loss of generality, let us set the origin at time-step $k$ at $\bar p(k)$, then using (\ref{UpdatedS_2}), we can write:
$$ p_i(k+1) = p_i(k) - n\sigma\chi_{i}(k)p_i(k) $$
hence:
$$ \bar p(k+1) = \bar p(k) -\frac{1}{n} \sum\limits_{i=1}^n n\sigma\chi_{i}(k)p_i(k) =  -\frac{1}{n} \sum\limits_{i=1}^n n\sigma\chi_{i}(k)p_i(k)$$
and
$$ \mathcal L(P(k+1)) = \sum\limits_{i=1}^n \| p_i(k)(1 - n\sigma\chi_{i}(k)) + \frac{1}{n} \sum\limits_{j=1}^n n\sigma\chi_{j}(k)p_j(k) \| $$
We have
$$ \mathcal L(P(k+1)) \le \sum\limits_{i=1}^n \| p_i(k)(1 - n\sigma\chi_{i}(k)) \| + \frac{1}{n}\sum\limits_{i=1}^n \| \sum\limits_{j=1}^n n\sigma\chi_{j}(k)p_j(k) \|  $$
Recall that the gain factor $\sigma$ is bounded by $0<\sigma< 1/n$, which yields:
\begin{equation} \label{SemiSyncIneq}
\begin{array}{l}
	\mathcal L(P(k+1)) \le \sum\limits_{i=1}^n \|p_i(k)\| \left(1 - n\sigma\chi_{i}(k)\right) + \| \sum\limits_{j=1}^n n\sigma\chi_{j}(k)p_j(k) \|= \\
	= \mathcal L (P(k)) - n\sigma \left[\sum\limits_{i=1}^n \chi_{i}(k)\|p_i(k)\| - \| \sum\limits_{j=1}^n \chi_{j}(k)p_j(k) \| \right]
\end{array}
\end{equation}
Since, by the generalised triangle inequality, we have that
$$ \sum\limits_{i=1}^n \chi_{i}(k)\|p_i(k)\| \ge \| \sum\limits_{j=1}^n \chi_{j}(k)p_j(k) \| $$
we obtain for arbitrary $\chi_i(k)$'s:
$$ \mathcal L(P(k+1)) \le \mathcal L (P(k)) $$\\

In order to show that we have a strictly positive probability that, for any $\mathcal L(P(k)) > 0$, it decreases by a strictly positive value, we use the \textbf{"strong asynchronicity assumption"} as introduced by Gordon et. al. \cite{gordon2004,gordon2005,gordon2008,gordon2010fundamental}.

\begin{definition}
\textbf{"Strong asynchronicity assumption"} : There exist a strictly positive constant $\delta$ such that for any subset $S$ of agents, at each time-step $k$, the probability that $S$ will be the set of active agents is at least $\delta$.
\end{definition}

In our case, the probability that all agents of the system are active at a given time-step $k$ is at least $\rho$ (here $\delta=(1-\rho)^n$), and the probability of all agents to be active is $\rho^n \ge \delta$). At these time-steps the dynamics of the system is similar to the dynamics of system $\mathcal S_2$, therefore, by Lemma \ref{invariant}  we have that $\bar p(k+1) = \bar p(k)$ and by (\ref{S2_dev}) we have that $\mathcal L(P(k))$ drops as follows:
$$ \mathcal L(P(k+1)) = \sum\limits_{i=1}^n \| p_i(k+1) - \bar p(k+1) \| =$$ 
$$= \sum\limits_{i=1}^n \| (1 - n\sigma)(p_i(k) - \bar p(k)) \|  = (1 - n\sigma)\mathcal L(P(k))$$\\

Let $\tilde{\mathcal{L}}$ be a small positive constant $\tilde{\mathcal{L}} \ll \mathcal{L}(P(0))$. We know that $\mathcal L(P(k))$ is a non increasing sequence of numbers starting at $\mathcal L(P(0))$, and by the strong asynchronicity property, we know that at each time-step $k$ there is a probability $\delta$ that $\mathcal L(P(k+1)) = (1-n\sigma)\mathcal L(P(k))$ where $0 < (1-n\sigma) < 1$. After $M$ such successful steps, we shall have that $\mathcal L(P(k))$ will reach $\mathcal L(P(0))(1-n\sigma)^M$.\\\\
In order to have
$$ \mathcal L(P(0))(1-n\sigma)^M < \tilde{\mathcal L}$$
we need to have had at least $M \ge M_0$ steps, where
$$ M_0 = \frac{\ln(\frac{\tilde{\mathcal{L}}}{\mathcal L(P(0))})}{\ln(1-n\sigma)}$$
In a Bernoulli process where the probability of "success" is $\delta$ and of "failure" is $(1-\delta)$, the expected number of time-steps for the first "success" is given by
$$ \sum\limits_{k=1}^{\infty} \tilde{k}(1-\delta)^{k-1}\delta = \frac{1}{\delta} $$
hence the expected time of $M_0$ successes will be $M_0 / \delta$. Therefore the expected number of time-steps to reach a value $\tilde{\mathcal{L}} \ll \mathcal L(P(0)) $ is given by
$$ \mathbb E (\tilde{k}(\tilde{\mathcal{L}})) = \frac{M_0}{\delta} = \frac{1}{\delta}\frac{\ln(\frac{\tilde{\mathcal{L}}}{\mathcal L(P(0))})}{\ln(1-n\sigma)} $$
Hence, the strong asynchronicity assumption ensures gathering of all the agents to an $\epsilon$-disk within a finite expected number of time-steps which is upper bounded by $\mathbb E (\tilde{k}(\tilde{\mathcal{L}}=\epsilon))$ given above.
\newpage
\section{Unlimited Visibility, Bearing-only Sensing}

We here assume every agent has information on the \textit{bearing} or direction only to all the other agents in the system (but cannot measure their relative distance).

\subsection{Continuous Time Dynamics (system $\mathcal{S}_3$)}
Consider that agents move according to the following dynamic law:

\begin{equation}
\dot{p_i}(t)=-\sigma\sum_{j=1}^{n} f^{\mathcal{S}_3} (p_i(t)-p_j(t))
\label{eq:Dynamics5}
\end{equation}
where,
\begin{equation}
f^{\mathcal{S}_3}(p_i(t)-p_j(t))= \left\{\begin{alignedat}{2}
   & \frac{{p_i(t)-p_j(t)}}{\|p_i(t)-p_j(t)\|}, &&\quad{p_i(t)}\neq {p_j(t)} \\
   & 0,  && \quad o.w. \\
 \end{alignedat}\right.
\end{equation}
and $\sigma$ is a positive scalar gain factor.\\

Here $\dot p_i(t)$ is proportional to the vector-sum of unit-vectors pointing from $p_i(t)$ to all other agents.\newline

The function $f^{\mathcal{S}_3}$ is antisymmetric, therefore the average position of agents in the system $\mathcal{S}_3$ is invariant, i.e. $\bar p=const$, due to Lemma $\ref{invariant}$. So, without loss of generality, we consider the positions $p_i(t)$ in a coordinate system centered at $\bar{p} = 0$.

\begin{theorem}
For any initial constellation, the agents of system $\mathcal{S}_3$ converge to the average location of the initial constellation in finite time.
\end{theorem}
Notice that the upper bound on the convergence time presented below depends only on the initial constellation (via the largest distance between the agents in the constellation) and the gain factor of the system $\sigma$, regardless of the number of agents $n$.
\begin{proof}

Let $\mathcal{L}(P(t))$ be the sum of distances between all pairs of agents at time $t$:

\begin{equation}
\mathcal{L}(P(t))=\frac{1}{2}\sum_{\substack {i=1}}^{n}\sum_{\substack {j=1}}^{n} {\|p_i (t) -p_j (t)\|}
\label{eq:SumOfDistances}
\end{equation}

As a sum of non-negative elements, the value $\mathcal{L}(P(t))$ is strictly positive, unless the system converged, i.e.
$$\forall i,t;\  \ p_i(t)=\bar{p}=0 \iff \mathcal{L}(P(t))=0$$
We next prove that as long as $\mathcal{L}(P(t)) \neq 0$ it decreases at a rate bounded away from zero. Indeed, we have that the dynamics given by (\ref{eq:Dynamics5}) turns out to be the gradient descent flow of the Lyapunov function $\mathcal{L}(P(t))$.

$$ \dot p_i(t) = -\sigma\nabla_{p_i(t)} \mathcal{L}(P(t)) = - \sigma\sum_{\substack {j=1 \\ j \neq i}}^{n} \frac{p_i (t) -p_j (t)}{\|p_i (t) -p_j (t)\|} $$
$$ \dot{\mathcal{L}}(P(t)) = \nabla_P \mathcal{L}(P(t)) \cdot \dot{P}(t) = \nabla_P \mathcal{L}(P(t)) \cdot \sigma \nabla_P \mathcal{L}(P(t)) =$$
$$ = -\frac{1}{\sigma}\|\sigma \nabla_P \mathcal{L}(P(t))\|^2 = -\frac{1}{\sigma}\sum_{i=1}^{n}\|\dot p_i(t)\|^2$$

The corresponding temporal evolution of the $\mathcal{L}(P(t))$ function is therefore given by
$$ \dot{\mathcal{L}}(P(t)) = - \frac{1}{\sigma}\sum_{i=1}^{n} \| \dot p_i(t) \|^2  $$

Denoting by $p_{max}(t)$ the location of the farthest agent from $\bar p$ (i.e. from the origin) at time $t$, and denote by $\hat p_{max}^\intercal$ the unit vector pointing from the origin to the direction of $p_{max}(t)$, then we readily have\\

\begin{equation}                              
\dot{\mathcal{L}}(P(t)) = - \frac{1}{\sigma}\sum_{i=1}^{n}\|\dot p_i(t) \|^2 \leq - \frac{1}{\sigma} \|\dot p_{max}(t) \|^2 \leq - \frac{1}{\sigma} | \hat p_{max}^\intercal \dot p_{max}(t) |^2
\label{eq:FuthestAgentLyapunov}
\end{equation}
since $\frac{dp_{max}(t)}{dt} \in \{\frac{dp_i(t)}{dt}\}_{i=1,2,...,n}$, and the length of any projection of a vector is equal to or smaller than its norm.\\

The next claim bounds the value of
$$ | \hat p_{max}^\intercal \dot p_{max}(t) | = \left|\sigma \hat p_{max}^{\intercal}(t) \sum_{\substack {j=1 \\ j \neq max}}^{n} \frac{p_{max}(t)-p_j(t)}{\|p_{max}(t)-p_j(t)\|}\right| $$
by a strictly positive constant.
 
\begin{proposition} \label{Prop4}
$$ \sum_{\substack {j=1 \\ j \neq max}}^{n}\hat p_{max}^{\intercal}(t)\frac{p_{max}(t)-p_j(t)}{\|p_{max}(t)-p_j(t)\|} \geq \frac{n}{2} $$

\end{proposition}

\begin{proof}
Since we have that
$$
\sum_{j=1}^{n}p_j(t)=0
$$
$$
\sum_{\substack {j=1}}^{n}(p_j(t)-p_i(t))=-n p_i(t)
$$
or re-writing this in terms of unit-vectors multiplied by their lengths:
$$
	\sum_{\substack {j=1 \\ j \neq i}}^{n}\frac{p_i(t)-p_j(t)}{\|p_j(t)-p_i(t)\|}\|p_j(t)-p_i(t)\| =n\hat p_i(t)\|p_i(t)\|
$$
Projecting both sides on the unit vector $\hat p_i(t)$ we get

\begin{equation} \label{Eq:Prop5} 
	\sum_{\substack {j=1 \\ j \neq i}}^{n}\hat p_i^{\intercal}(t)\frac{p_i(t)-p_j(t)}{\|p_j(t)-p_i(t)\|}\|p_j(t)-p_i(t)\| =n\|p_i(t)\|
\end{equation}

Let $i$ be the index of the farthest agent from the origin, i.e. $p_i(t)=p_{max}(t)$, then by dividing both sides of (\ref{Eq:Prop5}) by $\|p_i(t)\|=\|p_{max}(t)\|$ we get

\begin{equation} \label{Eq:Prop5_1} 
\sum_{\substack {j=1 \\ j \neq max}}^{n}\hat p_{max}^{\intercal}(t)\frac{p_{max}(t)-p_j(t)}{\|p_{max}(t)-p_j(t)\|}\frac{\|p_{max}(t)-p_j(t)\|}{\|p_{max}(t)\|} =n
\end{equation}

Since the distance between two points in a disc of radius $\|p_{max}\|$ is always less than or equal twice the radius, we get that 

$$\{ \forall j \; : \; 0 \le \frac{\|p_j(t)-p_{max}(t)\|}{\|p_{max}(t)\|} \le 2 \}$$
and since all the angles between the vectors $p_{max}-p_j$ pointing from any point inside a disc of radius $p_{max}$ to the point $p_{max}$ on the boundary, and the vector $p_{max}$ are between $-\frac{\pi}{2}$ and $\frac{\pi}{2}$, we get that
$$\{ \forall j \; : \; 0 \le \hat{p}_{max}^{\intercal}(t)\frac{p_{max}(t)-p_j(t)}{\|p_j(t)-p_{max}(t)\|} \le 1 \}$$

Therefore, since all terms in the sum are positive, rewriting (\ref{Eq:Prop5_1}) we get that
$$
2 \sum_{\substack {j=1 \\ j \neq max}}^{n}\hat{p}_{max}^{\intercal}(t)\frac{p_{max}(t)-p_j(t)}{\|p_{max}(t)-p_j(t)\|} \geq n
$$
or
$$
\sum_{\substack {j=1 \\ j \neq max}}^{n}\hat{p}_{max}^{\intercal}(t)\frac{p_{max}(t)-p_j(t)}{\|p_{max}(t)-p_j(t)\|} \geq \frac{n}{2}
$$
as claimed.

\end{proof}

By Proposition \ref{Prop4} we obtain from (\ref{eq:FuthestAgentLyapunov}) that 
$$ \dot{\mathcal{L}}(P(t)) \leq-\sigma \frac{n^2}{4} $$
Using the fact that
$$ \mathcal{L}(P(t))=\mathcal{L}(P(0))+\int_0^t \dot{\mathcal{L}} (P(t')) dt'$$
and
$$ \mathcal{L}(P(0))\leq n^2 \smash{\displaystyle \max_{i,j}} \|p_i(0)-p_j(0)\| $$
we get
$$ \mathcal{L}(P(t))\leq n^2 \smash{\displaystyle \max_{i,j}} \|p_i(0)-p_j(0)\| - \sigma \frac{n^2}{4} t = $$ $$ = n^2( \smash{\displaystyle \max_{i,j}} \|p_i(0)-p_j(0)\| - \sigma \frac{1}{4} t)  $$
Hence, the time for system $\mathcal{S}_3$ to converge is upper bounded by a finite value $t_{ub}$
$$ t_{ub} = \frac{4}{\sigma} \smash{\displaystyle \max_{i,j}} \|p_i(0)-p_j(0)\| $$
i.e. the upper bound for the system convergence time to $\bar p$ depends on the gain factor of the system $\sigma$ and the initial constellation. Note that this bound is not affected by the number of agents in the system.\\

\end{proof}

Simulation of gathering with $n=6$ agents is presented in Figure \ref{FigBearingOnlySim}. Notice that the agents' trajectories are no longer linear and that agents may collide and continue to move on their gathering path together.\\

\begin{figure}[H]
\captionsetup{width=0.8\textwidth}
  \centering
	\includegraphics[width=110mm]{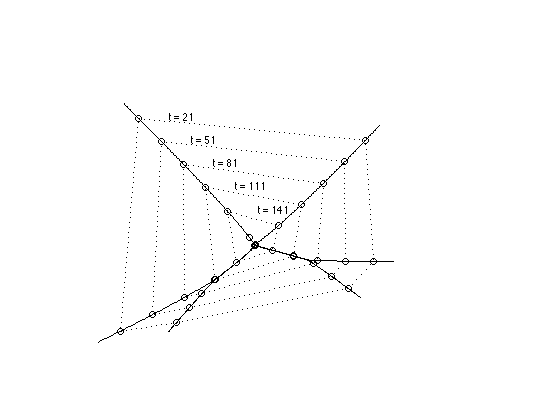}
	  \caption{Simulation results for a complete and undirected graph with bearing-only sensing. The system's dynamics is non-linear, and the movement direction of the agents is discontinuous.}
    \label{FigBearingOnlySim}
\end{figure}


\subsection{Discrete Time Dynamics (system $\mathcal{S}_4$)}
We here assume each agent moves according to the following dynamic law:

\begin{equation}
{p_i}(k+1)={p_i}(k)-\sigma\sum_{j=1}^{n} f^{\mathcal{S}_4} (p_i(k)-p_j(k))
\label{eq:Dynamics6}
\end{equation}
where,
\[
f^{\mathcal{S}_4}(p_i(k)-p_j(k))= \left\{\begin{alignedat}{2}
   & \frac{{p_i(k)-p_j(k)}}{\|p_i(k)-p_j(k)\|}, \quad&&{p_i(k)}\neq {p_j(k)} \\
   & 0,  \quad&&  o.w. \\
 \end{alignedat}\right.
\]

Here the step of agent $i$ at time $k$ is the vector-sum of all relative direction unit-vectors between agent $i$ and the rest of the agents, multiplied by a gain factor $\sigma$.\\

Since $f^{\mathcal{S}_4}$ is an antisymmetric function, the average position in $\mathcal{S}_4$ is invariant, i.e. $\bar p=const$, due to Lemma $\ref{invariant}$.

\begin{theorem} \label{Theorem6}
For any arbitrary initial constellation, all agents of system $\mathcal{S}_4$ gather to a disc centered at $\bar p$, with a radius of the order $\sigma n^2$, in a finite number of time-steps. 
\end{theorem}

\begin{proof}

Let $\mathcal{L}(P(k))$ be the sum of squared distances of all agents from $\bar p$ at time-step $k$.
Without loss of generality, since $\bar p$ is invariant, consider the positions $p_i(k)$ in a coordinate system centered at $\bar p$, i.e. $\bar p = 0$, and write:

\begin{equation}
\mathcal{L}(P(k))=\sum_{i=1}^{n} \|p_i(k)\|^2
\label{eq:SumOfSquareDistances}
\end{equation}

In the sequael we shall prove Theorem \ref{Theorem6} in four steps:
\begin{enumerate}
\item The function $\mathcal{L}(P(k))$ decreases while the sum of agents' distances from  $\bar p$ is greater than $\sigma (n-1)^2$ (Lemma \ref{DecreseMax}).
\item When the sum of agents' distances from $\bar p$ is greater than $\sigma(n-1)^2$, rule (\ref{eq:Dynamics6}) ensures that $\mathcal{L}(P(k))-\mathcal{L}(P(k+1))$ is strictly positive, hence the size of the decrease is bounded away from zero (Lemma \ref{FinitePositive4}).
\item The value of $\mathcal{L}(P(k))$ is upper bounded after finite number of time-steps (Lemma \ref{kmax}).
\item Using the bound of Lemma \ref{kmax} the positions of all agents are confined within a circle with radius of the order $\sigma n^2$.
\end{enumerate}

\begin{lemma}\label{DecreseMax}
The function $\mathcal{L}(P(k))$ decreases while the sum of all agents' distances from $\bar p$ is greater than $\sigma (n-1)^2$.
\end{lemma}

\begin{proof}
Let us develop $\mathcal{L}(P(k+1))$ using (\ref{eq:Dynamics6}) and (\ref{eq:SumOfSquareDistances}). Assuming $p_i(k) \neq p_j(k)$ for all $\{i,j\}$ at any time step $k$ (otherwise the two agents act as one), we can write:

\begin{equation} \label{DecreaseCond} 
\mathcal{L}(P(k+1)) = \sum_{i=1}^{n} \left\| p_i(k) - \sigma\sum_{\substack {j=1 \\ j \neq i}}^{n}  \frac{p_i(k) -p_j(k)}{\|p_i(k) -p_j(k)\|} \right\|^2=
\end{equation}
$$=\sum_{i=1}^{n}\| p_i(k)\|^2 -  \sigma\sum_{i=1}^{n} \sum_{\substack {j=1 }}^{n} \left\|p_i(k) -p_j(k)\right\| +{\sigma}^2\sum_{i=1}^{n}  \left\| \sum_{\substack {j=1 \\ j \neq i}}^{n}  \frac{p_i(k) -p_j(k)}{\|p_i(k) -p_j(k)\|} \right\|^2$$
Note that we here used the same method as in the development of (\ref{SquareDistancesFromMean}) in order to state that
$$  2\sum_{i=1}^{n} p_i(k)^\intercal \sum_{\substack {j=1 \\ j \neq i}}^{n}  \frac{p_i(k) -p_j(k)}{\|p_i(k) -p_j(k)\|} = \sum_{i=1}^{n} \sum_{\substack {j=1 }}^{n} \left\|p_i(k) -p_j(k)\right\| $$

Hence, for $\mathcal{L}(P(k+1)) < \mathcal{L}(P(k))$ to hold we need
\begin{equation} \label{inequality4}
	\sigma\sum_{i=1}^{n}  \left\| \sum_{\substack {j=1 \\ j \neq i}}^{n}  \frac{p_i(k) -p_j(k)}{\|p_i(k) -p_j(k)\|} \right\|^2 < \sum_{i=1}^{n}  \sum_{\substack {j=1 \\ }}^{n}  \left\|p_i(k) -p_j(k)\right\|
\end{equation} 

We now use the property of transitivity to find the conditions for the inequality $\mathcal{L}(P(k+1)) < \mathcal{L}(P(k))$ to hold. For that purpose, let us define the quantities $(a)$, $(b)$, $(c)$, and $(d)$:\\
$$
\begin{array}{l}
(a): \sigma\sum_{i=1}^{n}  \left\| \sum_{\substack {j=1 \\ j \neq i}}^{n}  \frac{p_i(k) -p_j(k)}{\|p_i(k) -p_j(k)\|} \right\|^2 \text{ the left term of (\ref{inequality4})},\\ \\
(b): \sigma n(n-1)^2\\ \\
(c): n \sum\limits_{i=1}^n \|p_i(k)\|  \\ \\
(d): \sum_{i=1}^{n}  \sum_{\substack {j=1 \\ j \neq i}}^{n}  \left\|p_i(k) -p_j(k)\right\| \text{ the right term of (\ref{inequality4})}.\\
\end{array}	
$$

Now, we have that $(a) \leq (b)$ since the left term of (\ref{inequality4}) is $\sigma$ times the sum of $n$ norms of sums of $(n-1)$ unit vectors squared, and each has total length less than or equal to $(n-1)^2$.\\

Additionally, we have that $(c) \leq (d)$ since

$$ \sum_{\substack {j=1}}^{n} p_j = 0 \Rightarrow \sum_{\substack {j=1}}^{n}(p_i - p_j) = np_i \Rightarrow \sum_{\substack {j=1}}^{n}\hat p_i^{\intercal}(p_i - p_j) = n\|p_i\| \Rightarrow$$
$$ \Rightarrow\sum_{\substack {j=1 }}^{n}\|p_i - p_j\| \ge n\|p_i\| \Rightarrow \sum_{\substack {i=1}}^{n}\sum_{\substack {j=1 }}^{n} \|p_i - p_j\| \ge n\sum_{i=1}^{n} \|p_i \| $$

Therefore, if  $(b) < (c)$, i.e. the following condition is satisfied
\begin{equation} \label{DecreaseCond2_1}
\sigma n(n-1)^2 < n \sum\limits_{i=1}^n \|p_i(k)\|
\end{equation} 
then necessarily $(a) < (d)$, implying 

\begin{equation} \label{DecreaseCond2}
	\mathcal{L}(P(k+1)) < \mathcal{L}(P(k))
\end{equation}
as claimed in Lemma \ref{DecreseMax}.\\

\end{proof}


Let us discuss the dynamics of $\mathcal{L}(P(k))$ while condition (\ref{DecreaseCond2_1}) is satisfied. if $n=2$ the system converges to oscillates around $\bar p$ with a smaller than $\sigma$ amplitude , and if $n>2$, we have the following result:

\begin{lemma}\label{FinitePositive4} 
If $n>2$ and condition (\ref{DecreaseCond2_1}) is satisfied, then $\mathcal{L}(P(k)) - \mathcal{L}(P(k+1)) > \delta$ where $\delta$ is a positive and bounded away from zero constant.\\
\end{lemma}

\begin{proof}
Recall that we showed $(a)\leq(b)$, i.e. 

$$\sigma \sum_{i=1}^{n}  \left\| \sum_{\substack {j=1 \\ j \neq i}}^{n}  \frac{p_i(k) -p_j(k)}{\|p_i(k) -p_j(k)\|} \right\|^2 \leq \sigma n(n-1)^2$$

Denote the internal variable of the left term as $\|\Delta p_i(k)\|$, so that
$$\|\Delta p_i(k)\| = \left\| \sigma\sum_{\substack {j=1 \\ j \neq i}}^{n}  \frac{p_i(k) -p_j(k)}{\|p_i(k) -p_j(k)\|} \right\|$$ 
Notice that $\|\Delta p_i(k)\|$ is the step-size of agent $i$ at time-step $k$, which is smaller than or equal to $(n-1)$.\\

Next, we show that the sum of all squared step sizes is strictly upper bounded by $\sigma n(n-1)^2$ by a constant.
$$ \sigma n(n-1)^2 - \sum_{i=1}^{n} \|\Delta p_i(k)\|^2 > \delta  $$ 
Here we assume each agent jumps its maximal possible step-size.\\
Let $d1$ and $d2$ be a pair of agents which define the diameter of system, i.e
$$ \{d1,d2\} = arg\max\limits_{i,j} \|p_i(k) - p_j(k)\|$$
Assume those agents current steps have the maximal possible length, i.e
$$\|\Delta p_{d1}(k)\| = \|\Delta p_{d2}(k)\| = \sigma(n-1) $$

For each other agents $i$ ($ \neq d1,d2$), angle $\angle p_{d1}(k) p_i(k) p_{d2}(k)$ is greater than or equal to $\pi/3$, otherwise agents $d1$ and $d2$ would not define the diameter of the system. Denote this angle as $\theta_i$. Let $\theta_{i1}$ and $\theta_{i2}$ be the angle between $\Delta p_i(k)$ agent's $i$ step to the vectors pointing from $p_i(k)$ to $p_{d1}(k)$ and to $p_{d2}(k)$ respectively, then $ \theta_i = \theta_{i1}+\theta_{i2}$.\\

Assume each agent contribute the maximal possible contribution to $\|\Delta p_i(k)\|$, so that agents $d1$ and $d2$ contribute the maximal possible value as follows:

$$\max \left\{ \sigma\Delta p_i^\intercal(k)\frac{p_{d1}(k) - p_i(k)}{\|p_{d1}(k) - p_i(k)\|} + \sigma\Delta p_i^\intercal(k)\frac{p_{d2}(k) - p_i(k)}{\|p_{d2}(k) - p_i(k)\|} \right\} = $$
$$ = \max \left\{\sigma\cos(\theta_{i1}) + \sigma\cos(\theta_{i2}) \right\} $$
$$ \mbox{s.t.  }  \frac{\pi}{3} \le \theta_{i1}+\theta_{i2} \le \pi $$

which is $2\sigma\cos(\pi/6) = \sigma\sqrt{3}$.\\
Furthermore, each one of the other agents of the system contributes $\sigma$ to $\|\Delta p_i(k)\|$, which is the maximal possible contribution. (This may happen if all those agents are located at positions with the direction of $\Delta p_i(k)$ relative to agent's $i$ position). Therefore, $\|\Delta p_i(k)\|$ is bounded from above by $\sigma(n-3 + \sqrt{3})$, and we have that
$$ \sigma^2 \sum_{i=1}^{n}  \left\| \sum_{\substack {j=1 \\ j \neq i}}^{n}  \frac{p_i(k) -p_j(k)}{\|p_i(k) -p_j(k)\|} \right\|^2\le \sigma^2\left(2(n-1)^2 + (n-2)\left((n-3 + \sqrt{3})\right)^2 \right) < \sigma^2 n(n-1)^2 $$
Denote $\delta$ as the following difference:
$$\delta = \sigma^2 n(n-1)^2 - \sigma^2 \left(2(n-1)^2 + (n-2)\left((n-3 + \sqrt{3})\right)^2 \right)$$ Then, we have that

$$ \mathcal{L}(P(k+1)) - \mathcal{L}(P(k)) \le -  \sigma\sum_{i=1}^{n} \sum_{\substack {j=1 }}^{n} \left\|p_i(k) -p_j(k)\right\| +{\sigma}^2\sum_{i=1}^{n}  \left\| \sum_{\substack {j=1 \\ j \neq i}}^{n}  \frac{p_i(k) -p_j(k)}{\|p_i(k) -p_j(k)\|} \right\|^2 \le  $$
$$ \le - \sigma n\sum\limits_{i=1}^n \|p_i(k)\| + \sigma^2 n(n-1)^2 - \delta $$

Now, we may fix Lemma \ref{DecreseMax}, which suggests that dynamics (\ref{DecreaseCond2}) realized as a result of satisfying condition (\ref{DecreaseCond2_1}), i.e. if\\

$$\sigma (n-1)^2 \le \sum\limits_{i=1}^n \|p_i(k)\|$$

then\\
 
$$\mathcal{L}(P(k+1)) - \mathcal{L}(P(k)) < 0$$

as follows:\\

If  condition (\ref{DecreaseCond2_1}) is satisfied
$$\sigma (n-1)^2 \le \sum\limits_{i=1}^n \|p_i(k)\|$$
then
$$ \mathcal{L}(P(k+1)) - \mathcal{L}(P(k)) \le - \delta $$
and therefore $\mathcal{L}(P(k))$ decreases as claimed in Lemma \ref{FinitePositive4}.\\
\end{proof}


The gathering process causes far away agents to approach $\bar p$ at first, but eventually they may disperse within a region with a limited radius from the invariant average location.\\

We next bound the size of the region where all agents of $\mathcal{S}_4$ will necessarily gather within finite number of time-steps.\\

\begin{lemma}\label{kmax}
For any arbitrary initial constellation the value of $\mathcal{L}(P(k))$ is upper bounded by a value of order $n^4$ after finite number of time-steps.
\end{lemma}

\begin{proof}

Let $k_c$ be the first time-step when the sum of distances of all agents from $\bar p$ is smaller than or equal to $\sigma(n-1)^2$, i.e when the following condition holds:
\begin{equation}\label{k_cCond}
\sum\limits_{i=1}^n \|p_i(k)\| \le \sigma(n-1)^2
\end{equation}
so that condition (\ref{DecreaseCond2}) is not satisfied, and the function $\mathcal{L}(P(k)$ is not guaranteed to decrease.\\

Let us bound the maximal value of $\mathcal{L}(P(k))$ for all time-steps $k>k_c$. If condition (\ref{k_cCond}) holds, the value of $\mathcal{L}(P(k))$ is upper bounded as follows
$$\mathcal{L}(P(k)) = \sum\limits_{i=1}^n \|p_i(k)\|^2 \le \left(\sum\limits_{i=1}^n \|p_i(k)\| \right)^2 \le (\sigma(n-1)^2)^2$$

Let us examine the first time-step when $k>k_c$, i.e. the first time-step when condition (\ref{k_cCond}) breaks. This is the first time-step when one or more agents may jump to positions that possibly increase the sum of agents' distances from $\bar p$ to above $\sigma(n-1)^2$. By Lemma \ref{DecreseMax} we have that at the next time-step, $\mathcal{L}(P(k+1))$ will decrease by a strictly positive value. Since the maximal step-size of an agent is $\sigma(n-1)$, we get that for all time steps $k>k_c$
$$ \mathcal{L}(P(k+1)) = \sum\limits_{i=1}^n \|p_i(k+1)\|^2 \le \left(\sum\limits_{i=1}^n \|p_i(k+1)\| \right)^2 \le (\sigma(n-1)^2 + n\sigma(n-1))^2 $$
Hence after time step $k_c$ the function $\mathcal{L}(P(k))$ is bounded from above by a value of order of $n^4$ as claimed.\\

\end{proof}

Let us now use the above results to prove Theorem \ref{Theorem6}.\\

In order to bound the maximal distance of an agent from $\bar p$ after time step $k_c$, assume that one agent, denoted by $\|p_{max}(k)\|$, travels as far as possible from $\bar p=0$ while all other agents gather at $\bar p=0$.\\

By definition of $\mathcal{L}(P(k))$, we get from the upper bound of Lemma \ref{kmax}
$$\|p_{max}(k)\|^2 \le \mathcal{L}(P(k)) \triangleq \sum_{i=1}^n \|p_i(k)\|^2  \leq \left(2\sigma (n-1)^2+\sigma(n-1)\right)^2$$
and therefore for any $k>k_c$, we get
$$ \|p_{max}(k)\| \le 2\sigma (n-1)^2+\sigma(n-1)$$

Therefore for any $k>k_c$, all the agents of $\mathcal{S}_4$ remain confined in a disc centered at $\bar p$ with a radius of order $\sigma n^2$ as claimed. \\
\end{proof}


\subsection{Discussion} 

Notice that, like for the system $\mathcal{S}_1$, system $\mathcal{S}_3$ is fully scalable in the sense of convergence guarantee. Adding or removing agents may affect the convergence rate, but convergence itself is always guaranteed.\\

As long as the sensing capabilities of all agents are identical and omni-directional so that their interconnection graph is undirected, by Lemma \ref{invariant}, the average position of the agents is invariant. Since the agents do not have memory, they react to what they "see" to update their motion.  The dynamics of systems $\mathcal{S}_3$ and $\mathcal{S}_4$ may be represented as $\dot P(t) = -L(P(t))P(t)$ and in the discrete case $P(k+1) = (I - L(P(k)))P(k)$ where $L(P(t))$ is a weighted graph Laplacian with weights continuously dependent on distances between the agents. Since the transition matrices $L(P)$ depend on agents' constellation (which changes in time), the dynamics is clearly non-linear.\\

\subsubsection{Gathering of non-linear systems}
L. Moreau in \cite{moreau2004} proves the following: If for every time period $\Delta t = [t_1,t_2]$ the matrix $L_w$ which is $L_w = \int_{t_1}^{t_2} L(t)dt$,  represents a weighted connected graph, then a system with dynamics $\dot p(t)=L(t)p(t)$ will (asymptotically) converges to a point. In fact, one can also use this result to prove that system $\mathcal{S}_3$ converges.\\
 
\subsubsection{Gathering region}
Unlike Systems $\mathcal{S}_1$, $\mathcal{S}_2$, and $\mathcal{S}_3$ whose agents converge to a single point, the agents of system $\mathcal{S}_4$ gather to a bounded region. The lack of information on neighbours' distance, together with the discrete dynamics of $\mathcal{S}_4$ and the finite step-size, prevent convergence to a point (except for some singular initial constellations).\\

The best upper limit we could find for the size of the gathering region of system $\mathcal{S}_4$ is of order $\sigma n^2$. We believe that this may be improved in the sense of getting a realistic bound, since simulation results show that the size of the gathering region is always bounded by a disc of radius less than $\sigma n$.\\

We analysed the dynamics of system $\mathcal{S}_4$, considering simulation results with a large number of  random initial constellations, and with various numbers of agents. An example of results from a simulation with 10 agents with 400 random initial constellations is presented in Figure \ref{FigMonteCarlo}.\\
\begin{figure}[h!]
\captionsetup{width=0.8\textwidth}
  \centering
	\includegraphics[width=120mm]{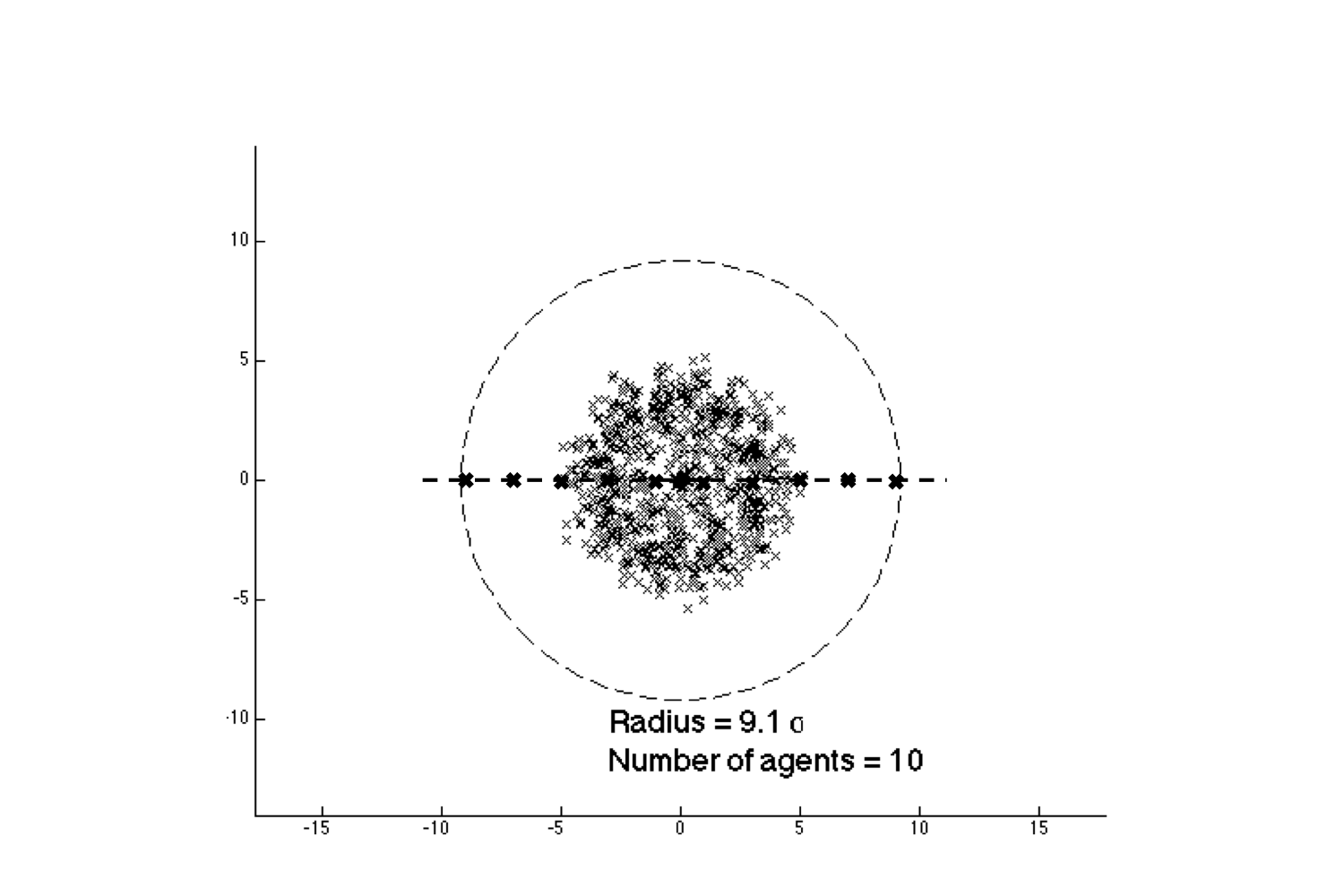}
	  \caption{Simulation results for 400 random initial constellations of system $\mathcal{S}_4$ with 10 agents. The horizontal 1D constellation yielded the worst case scenario and it was set deterministically for comparison.}
    \label{FigMonteCarlo}
\end{figure}

All empirical results suggest that the agents gather to a disc of radius of order $n\sigma$ rather than $n^2\sigma$, however proving this remain an open research problem.\\

\subsubsection{Convergence time}
We found an upper bound for the convergence time of system $\mathcal{S}_3$, which depends on the initial configuration only, rather than on the number of agents in the system. However, in simulations the convergence time of system $\mathcal{S}_3$ was also observed to become shorter as the number of agents increases.\\

We note that in \cite{bruckstein1991ants} similar results are discussed for the case of agents in cyclic pursuit. There, however, graph topology is both directed and cyclic. Despite the differences, the results are similar: for continuous time it has been shown that the time of convergence is finite, while for discrete time gathering region was proved to have a finite size.

\newpage
\section{Finite Visibility, Position Sensing}
In this and the next sections we work in $\mathbb{R}^2$ only and restrict the visibility of the agents to a finite horizon, and therefore the connection graph topology will be dependent on the geometric constellation of the agents' location. An agent $j$ will be acknowledged as a neighbour of agent $i$ at time $t$ if it is located within its \textit{visibility range} $V$, i.e.
$$ j \in N_i(t) \iff \| p_i(t) - p_j(t) \|<V $$
Since the functions describing the dynamics become \textit{discontinuous}, the resulting multi-agent systems are harder to analyze.\\


\subsection{Continuous Time Dynamics (system $\mathcal{S}_{5JE}$)}

In this section we present  a gathering process based on \textit{relative position} with limited visibility as proposed by Ji and Egerstedt in 2007 \cite{ji2007}.\\

Assume the agents move according to the law:

\begin{equation} \label{eq:Dynamics55}
	\dot{p}_i (t)= -\sigma\sum\limits_{j \in N_i(t)}\frac{2V -  l_{ij}(t)}{(V - l_{ij}(t))^2}(p_i(t) - p_j(t))
\end{equation}
where $ l_{ij}(t) = \|p_i(t) -p_j(t)\|$ is the distance between agents $i$ and $j$, and $\sigma$ is a positive scalar gain factor.\\

In order to avoid problems of dividing by zero in the agents' dynamics (\ref{eq:Dynamics55}) at $l_{ij}(t)=V$, Ji and Egerstedt modify the meaning of the neighborhood of an agent $i$ by adding a \textit{hysteresis rule}, as follows:  an agent $j$ is "acknowledged" at time $t$ to be a neighbour of agent $i$ if and only if it was within the range $V$ from $i$ at $t=0$, or it entered the range $V$ from $i$ and already crossed the range $V-\delta$ at some prior time, where $0 < \delta < V$ is small and bounded away from zero positive value. Note that this way of defining neighbours of agent $i$ requires the tracking of neighbours and remembering their past trajectories, capabilities that we do not assume our agents to have. Nevertheless we present Ji and Egerstedt's results to highlight the complexities in analyzing "potential functions"-based nonlinear systems, and for comparison purposes with the oblivious geometric approach to be presented in the sequel (System $\mathcal{S}_5$).\\

The hysteresis rule defining the neighbours of agent $i$ at time $t$ is given by

\begin{equation} \label{Hysteresis}
 j \in N_i(t) \iff \left\{\begin{alignedat}{2}
    & \{\exists \tilde t < t \; : \; \|p_i(\tilde t) - p_j(\tilde t) \| \le V - \delta  \} \\
    & or \\
    & \|p_i(0) - p_j(0) \| < V
  \end{alignedat}\right.
\end{equation}

Note that the rule (\ref{Hysteresis}) implies that once the agent $j$ is acknowledged as a neighbour of $i$, it will remain its neighbour forever, i.e. the distance between them will never exceed $V$. This is an immediate consequence of Ji and Egerstedt's dynamic rule that prevents any neighbour to overcome the infinity barrier due to division by $(V - l_{ij}(t))^2$ in (\ref{eq:Dynamics55}), which approaches zero as $j$ approaches the boundary of $i$.\\

The hysteresis rule enables new agents to join the neighbours of $i$ without inducing singularities.\\

Since the dynamics (\ref{eq:Dynamics55}) with the added hysteresis rule (\ref{Hysteresis}) remains an antisymmetric function, the average location of the agents is invariant, i.e. $\bar p=const$, due to Lemma $\ref{invariant}$.\\

Define $G(P(t))$ as the neighborhood graph of system $\mathcal{S}_{5JE}$ at time $t$.

\begin{theorem} \label{Theorem55}
For any initial constellation corresponding to a connected neighborhood graph, agents of system $\mathcal{S}_{5 JE}$ gather at the average position of their initial constellation.
\end{theorem}

\begin{proof} 
Let $\mathcal{L}(P(t),G(P(t)))$ be a Lyapunov function of the dynamic system $\mathcal{S}_{5JE}$ while the topology of $G(P(t))$ is fixed. Without loss of generality, since $\bar p$ is invariant, consider the positions $p_i(t)$ in a coordinate system centered at $\bar p=0$,

\begin{equation} \label{Lyapunov}
	\begin{array}{l}  
		\mathcal{L}(P(t),G(P(t))) =  \sum\limits_{i=1}^{n} \sum\limits_{j=1}^{n} \nu_{ij}(t) \\
		\\
		with \\
		\nu_{ij}(t) =  \left\{
					\begin{array}{ll}
						\frac{l_{ij}^2(t)}{V - l_{ij}(t) }  &  l_{ij}(t) \in G(P(t)) \\						0  & o.w.
					\end{array}
				\right.
	\end{array}
\end{equation}

\begin{lemma} \label{Lyapunov55}
If $G(P(t))$ is connected and its topology is fixed, all agents in system $\mathcal{S}_{5JE}$ asymptotically converge to the average position $\bar{p}$. \\
\end{lemma}

\begin{proof}
As long as not all agents are located at the same place (at $\bar p$) we have that $\mathcal{L}(P(t),G(P(t))) > 0$. The derivative of $\mathcal{L}(P(t),G(P(t)))$ is given by:

\begin{equation} \label{Eq:Dev55}
	\begin{array}{l}
		\dot{\mathcal{L}}(P(t),G(P(t))) = \sum\limits_{i=1}^{n} \left\{ \left( \frac{\partial \mathcal{L}(P(t),G(P(t)))}{\partial p_i}\right)^\intercal \dot{p}_i(t)\right\} = \\ \\
		= -\sum\limits_{i=1}^{n} \left( \sum\limits_{j \in N_i} \frac{ 2V - l_{ij}(t)}{(V - l_{ij}(t))^2}\|p_j(t) - p_i(t)\| \right)^2
	\end{array}
\end{equation}

Clearly, the derivative cannot be positive. Furthermore, only when all existing pairwise distances between agents $l_{ij}$ equal zero (i.e all agents are located at the same position), the function $\mathcal{L}$ itself equals zero. Since by Lemma \ref{invariant} the average position is an invariant, at $\bar p = 0$ by definition, and the derivative of the Lyapunov function ($\ref{Eq:Dev55}$) is strictly negative if agents are not all at $0$, all the agents of system $\mathcal{S}_{5JE}$ asymptotically converge to $\bar p$ if the graph topology does not change.

\end{proof}

\begin{lemma} \label{NeverLoseFriends_55}
If at time $t$ all neighbouring pairs in $G(P(t))$ of system $\mathcal{S}_{5JE}$ are closer than $V - \tilde \delta$, for some $\tilde \delta > 0$, these pairs of neighbours will remain neighbours forever.
 \end{lemma}
 
\begin{proof}
We prove this claim whether new pairs of neighbours are currently generated or not.\\

We start with the assumption that no new neighbouring pairs are currently generated, so at time $t$ we have that $\mathcal L(P(t))$ is bounded by a finite number $\mathcal{L}_{lim1}$
$$\mathcal{L}_{lim1} = \frac{M(V-\tilde{\delta})^2}{\tilde{\delta}}$$
where $M$ is the current number of edges in $G(P(t))$.\\

If a neighbouring pair tends to disconnect at $t_{dis}$ the division by a value close to zero in (\ref{Lyapunov}) would increase $\mathcal L(P(t_{dis}))$ beyond any finite limit. However, by Lemma \ref{Lyapunov55} the value of $\mathcal L(P(t))$ cannot increase since it is bounded by  $\mathcal{L}_{lim1}$. Therefore, while no new neighbouring pairs are generated, such disconnections cannot occur.\\

Next we assume that $m>0$ new neighbouring pairs are curently generated. The number of agents in the system is finite so the number of new pairs that may be generated is finite as well. Assume that $M$ is the current number of edges in $G(P(t))$ and at time step $t_{gen}$, $m$ new neighbouring pairs are generated. The value of $\mathcal L(P(t_{gen}))$ is again limited by another finite number $\mathcal{L}_{lim2}$
$$\mathcal{L}_{lim2} = \frac{M(V-\tilde{\delta})^2}{\tilde{\delta}} + \frac{m(V-\delta)^2}{\delta}$$
and here too no disconnection may occur, for the same reason as in the previous case. Since no disconnection may occur, pairs of neighbours will remain neighbours forever as claimed.

\end{proof}

To prove Theorem \ref{Theorem55}, Ji and Egerstet argue that, since no pair of neighbours can disconnect (by Lemma \ref{NeverLoseFriends_55}), system $\mathcal{S}_{5JE}$ has only two dynamical states: either the system preforms gathering to $\bar p$ (by Lemma \ref{Lyapunov55}) or new edges are added to the system. Since the maximal number of possible edges in the system is finite, the system necessarily gathers to the complete graph, and then all agents converge asymptotically to $\bar p$.\\
\end{proof} 


\textbf{An oblivious dynamic law:}\\

The aim of our survey is to analyze gathering of oblivious and anonymous agents. Ji and Egerstedt's algorithm requires each agent to be aware of whether another agent just entered its visibility range or is about to leave it. This requires tracking neighbours, and the capability to remember and mark certain neighbours as special, violating the obliviousness and anonymity paradigms. Therefore we present a different algorithm which causes the multi-agent system $\mathcal{S}_5$ to preform gathering without violating these assumptions.\\

\subsection{Continuous Time Dynamics (system $\mathcal{S}_5$)}

We shall denote by $\psi_i(t)$ the angular span of the minimal section containing all neighbours of agent $i$ which are near the visibility bound, i.e. neighbours at a distance between $V-\delta$ to $V$ from agent $i$ (where $\delta$ is a small positive value). Let us denote the set of agents in the band with distance between $V$ to $V-\delta$ from $p_i(t)$ by $B_i(t)$. Denote by $U_i^-(t)$ and $U_i^+(t)$ the unit vectors pointing from $p_i(t)$ to the current position of the right and left extremal agents in the set $B_i(t)$, i.e the agents defining the minimal section containing all agents from set $B_i(t)$. Notice that, the set $B_i(t)$ may contain only one agent, and then $U_i^-(t) = U_i^+(t)$.\\

Assume each agent moves according to the following oblivious dynamic law (see Figure \ref{EgerstedtAlternative}):
\begin{equation} \label{eq:Dynamics55_2}
\dot{p}_i (t) = 
\left\{
\begin{array}{ll}
\sigma \sum\limits_{j \in N_i(t)}(p_j(t) - p_i(t)) & B_i(t) = \emptyset \\
\\
\sigma (U_i^-(t) + U_i^+(t)) & B_i(t) \neq \emptyset \mbox{ \textit{and} } \psi_i(t) < \pi\\
\\
0 & o.w
\end{array}
\right.
\end{equation}
i.e. if the band $B_i(t)$ is empty - agent $i$ moves towards the average position of its current neighbours. Otherwise, if $\psi_i(t) < \pi$ - agent $i$ moves along its bisector, or stay put if $\psi_i(t) \geq \pi$.\\
\begin{figure}[h]
\captionsetup{width=0.8\textwidth}
  \centering
    \includegraphics[width=80mm]{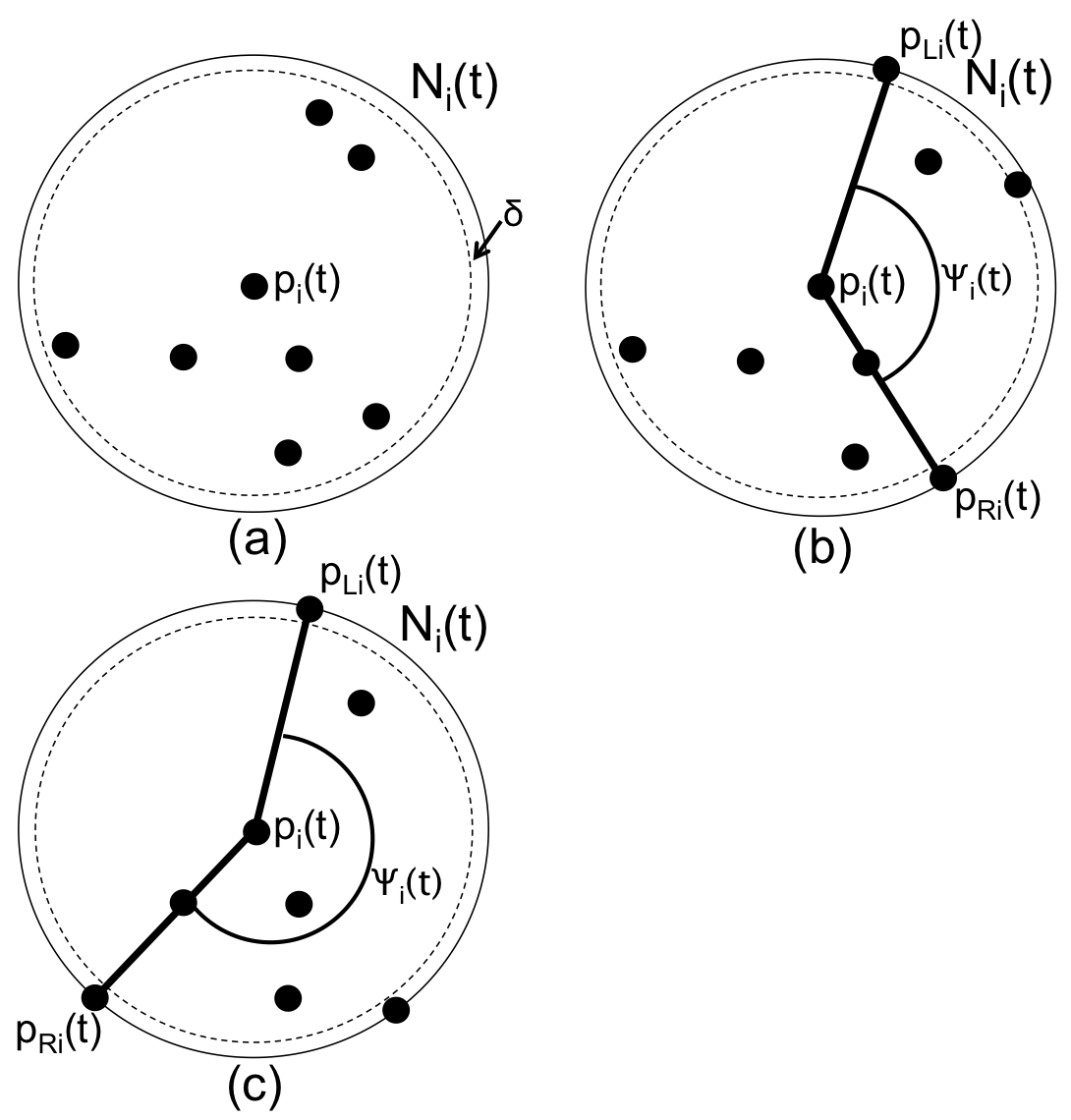}
     \caption{Dynamics (\ref{eq:Dynamics55_2}) depends on the neighborhood states, which can be:\\
$  \quad \quad \quad \quad \quad (a) \quad B_i(t) = \emptyset$\\
$ \quad \quad \quad \quad \quad (b)  \quad B_i(t) \neq \emptyset ; \quad \psi_i(t) < \pi$\\
$ \quad \quad \quad \quad \quad (c)  \quad B_i(t) \neq \emptyset ; \quad\psi_i(t) \ge \pi$}
     \label{EgerstedtAlternative}
\end{figure}

\begin{theorem} \label{Theorem55_2}
For any initial constellation corresponding to a connected neighborhood graph, agents of system $\mathcal{S}_5$ asymptotically gather to a point.
\end{theorem}

\begin{proof} 
We prove Theorem \ref{Theorem55_2} using four lemmas:
\begin{enumerate}
\item Neighbours remain neighbours (Lemma \ref{lemma55_2_1}).
\item The convex-hull of system $\mathcal{S}_5$ never increases (Lemma \ref{ChShrinks}).
\item If all the agents are within a distance of $V-\delta$ from each other, they will asymptotically converge (Lemma \ref{CloseAgentsConcerge}).
\item Any initial constellation with connected (but not fully necessarily fully connected) neighborhood graph evolves to the case above in finite time (Lemma \ref{lemma55_2_4}).
\end{enumerate}
These four lemmas prove that the agents of system $\mathcal{S}_5$ asymptotically gather to a point as claimed.\\

\begin{lemma} \label{lemma55_2_1}
Dynamics (\ref{eq:Dynamics55_2}) ensures that if agents $i$ and $j$ are neighbours at time $t'$ they will remain neighbours for all times $t>t'$.
\end{lemma}

\begin{proof}
Let $\theta_{ij}(t)$ be the current angle between the velocity vector $\dot p_i(t)$ of agent $i$, and the line defined by $p_i(t)$ and $p_j(t)$, and let $v_i(t)$ and $v_j(t)$ be the norm of $\dot p_i(t)$ and $\dot p_j(t)$ respectively (see Figure \ref{EgerstedtAlternativeLemmaProof}).\\

\begin{figure}[h!]
\captionsetup{width=0.8\textwidth}
  \centering
    \includegraphics[width=60mm]{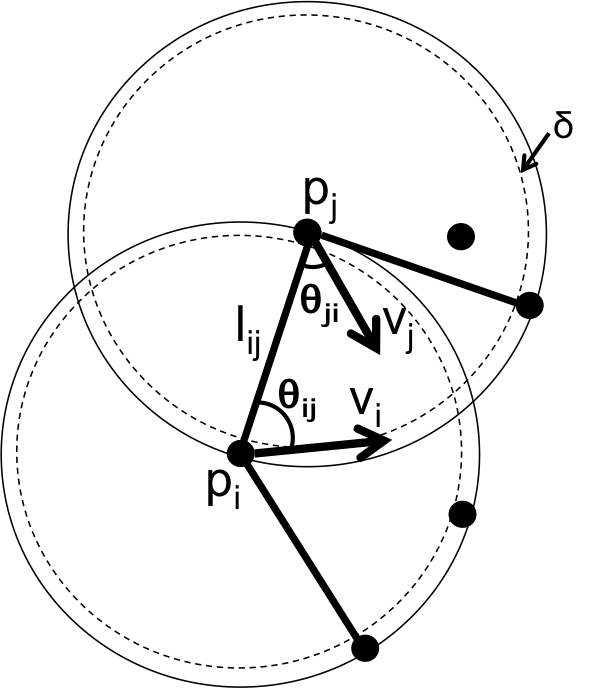}
     \caption{Drawing for Lemma \ref{lemma55_2_1}: visibility maintenance.}
     \label{EgerstedtAlternativeLemmaProof}
\end{figure}

Let $l_{ij}(t)$ be the current distance between agents $i$ and $j$ so that $l_{ij}(t) = \|p_i(t)-p_j(t)\|$. Its rate of change is as follows:
$$ \dot l_{ij}(t) = \frac{p_i(t)-p_j(t)}{ \|p_i(t)-p_j(t)\|}(\dot p_i(t)-\dot p_j(t)) = - \left(v_i(t)\cos(\theta_{ij}(t)) + v_j(t)\cos(\theta_{ji}(t))\right)$$

Let us focus only on agents located in the surrounding band $B_i(t)$, since only they may potentially lose visibility with agent $i$. If $\psi_i(t) < \pi$ then agent $i$ moves in the direction of the bisector of $\psi_i(t)$, hence 
$$\{ \forall j \in B_i(t) : \theta_{ij}(t) < \frac{\pi}{2} \}$$
otherwise, if $\psi_i(t) \geq \pi$ agent $i$ does not move.\\

If $j \in B_i(t)$ then necessarily $i \in B_j(t)$, therefore the same argument holds for agent $j$. Hence, the distance between agents $i$ and $j$ can not increase.\\
\end{proof}

\begin{lemma} \label{ChShrinks}
Let $CH(P(t))$ be the current convex-hull of all the agents of system $\mathcal{S}_5$. Then: 
$$\{\forall t , \Delta t \ge 0 : \; CH(P(t+\Delta t)) \subseteq CH(P(t)) \}$$
\end{lemma}

\begin{proof}
By (\ref{eq:Dynamics55_2}), the velocity of each agent is a linear combination of the relative positions of its neighbours multiplied by a non negative factor. Since all the neighbours are located either on the perimeter of the convex-hull or inside it, the motion of each agent on the perimeter of the convex-hull is either towards the interior of the convex-hull or along its perimeter.

\end{proof}

\begin{lemma}\label{CloseAgentsConcerge}
If all agents of system $\mathcal{S}_5$ are at a distance less then $V-\delta$ from each other, they asymptotically converge to their (current) average position.
\end{lemma}

\begin{proof}
In this case by Lemma \ref{ChShrinks}, the distance between any pair of agents will never exceed $V-\delta$. Therefore the system has a simple linear dynamics (as in system $\mathcal{S}_1$), which was proven in Theorem \ref{LinearTheorem} to asymptotically gather.

\end{proof}

\begin{lemma} \label{lemma55_2_4}
Any initial constellation of system $\mathcal{S}_5$ with a connected neighborhood graph, necessarily evolves to the state of Lemma \ref{CloseAgentsConcerge} where the agents are within a distance of $V-\delta$ from each other in finite time.
\end{lemma}

\begin{proof}

We shall show that the perimeter of $CH(P(t))$ drops at a strictly positive rate as long as the "diameter" of the system is strictly positive.\\
 
The proof is based on the dynamics of the agent (or agents) $s$, located at the current sharpest corner of the system's convex-hull. Let $\varphi_s$ be the inner angle of this corner.\\

The sum of angles of any convex polygon is  $\pi(m-2)$, where $m$ is the number of its corners, therefore the angle of its sharpest corner $\varphi_s$ is necessarily smaller than or equal to $\pi(1-\frac{2}{m})$. System $\mathcal{S}_5$ contains $n$ agents, hence the system's convex-hull has $ m \le n$ corners. We denote the upper bound on the angle of the sharpest corner by $\varphi_*$, hence

$$\varphi_s \le \varphi_* = \pi(1 - 2/n)$$

Let us show that the "current" agent $s$ necessarily moves inside the convex-hull of the system at a speed which is strictly positive.\\

We present the proof of this claim by projecting $\dot{p}_s$ on $\hat{\varphi}_s$ which is the bisector unit-vector of the convex-hull's corner occupied by agent $s$ (see Figure \ref{BisectorOfSharpestAngle}).

\begin{equation} \label{EqBisectorOfSharpestAngle}
\|\dot{p}_s(t)\| \ge \hat{\varphi}_s^T\dot{p}_s(t) = \sigma\sum\limits_{j \in N_i(t)} a_{j}\hat{\varphi}_s^T(p_j(t) - p_s(t))
\end{equation}
\begin{figure}[h!]
\captionsetup{width=0.8\textwidth}
  \centering
    \includegraphics[width=60mm]{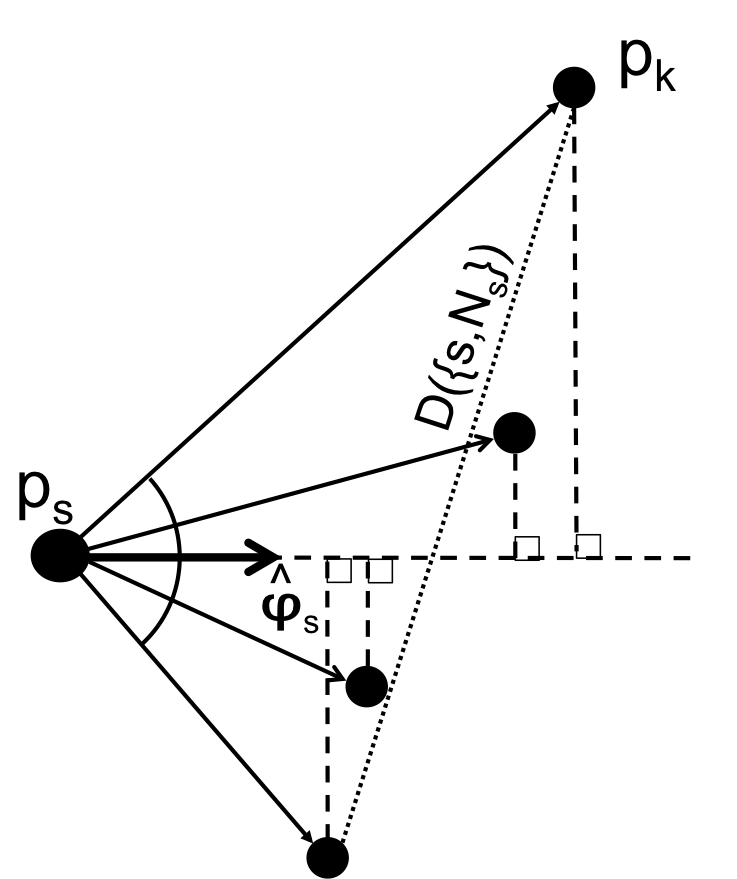}
    \caption{$\hat{\varphi}_s$ is the bisector unit vector of the sector created by the neighbours' positions of agent $s$}
   \label{BisectorOfSharpestAngle}
\end{figure}

where\\

\begin{equation}
	\begin{array}{l}  	
		a_j =  \left\{
					\begin{array}{ll}
						1  &  B_s(t) = \emptyset\\
						\\
						\frac{1}{\|p_j(t)-p_s(t)\|}  & B_s(t) \neq \emptyset  \quad and \quad (p_j(t) = p_{L_s}  \quad or \quad  p_j(t) = p_{R_s})\\
						\\
						0 & o.w
					\end{array}
				\right.
	\end{array}
\end{equation}



Let us first examine the simple case where $B_s(t) \neq \emptyset$. In this case all terms in (\ref{EqBisectorOfSharpestAngle}) are zero except two strictly positive terms which correspond to $L_s(t)$ and $R_s(t)$ the extremal left and right agents in the set $B_s(t)$:
$$  a_{L_s}\hat{\varphi}_s^\intercal(p_{L_s} - p_s) =  \hat{\varphi}_s^\intercal\frac{p_{L_s} - p_s}{\|p_{L_s} - p_s\|} \ge cos(\varphi_*/2) > 0$$ 
and
$$  a_{R_s}\hat{\varphi}_s^\intercal(p_{R_s} - p_s) =  \hat{\varphi}_s^\intercal\frac{p_{R_s} - p_s}{\|p_{R_s} - p_s\|} \ge cos(\varphi_*/2) > 0$$
hence in case $B_s(t) \neq \emptyset$ we have in (\ref{EqBisectorOfSharpestAngle}) that $\|\dot{p}_s(t)\| > 0$.\\

Now let us examine the other case where $B_s(t) = \emptyset$. Here all the coefficients $a_j$ equal $1$, and by geometry all the terms associated with the vectors $p_j(t)-p_s(t)$ have that the angle between $p_j(t)-p_s(t)$ and $\hat\varphi_s(t)$ is smaller than or equal to $\varphi_*/2$ (see figure \ref{BisectorOfSharpestAngle}), i.e.

$$\{\forall j \in N_s(t);\  \ \hat{\varphi}_s^\intercal(t)\frac{(p_j(t) - p_s(t))}{\|(p_j(t) - p_s(t))\|} > 0\}$$

Next we show that if the "diameter" of the set of points comprising agent $s$ and its neighbourhood, defined as the maximal distance between two points of the set, is bounded away from zero, then there exist at least one agent $j \in N_s$ whose distance from agent $s$ is strictly positive, yielding a strictly positive $\|\dot p_s\|$ in (\ref{EqBisectorOfSharpestAngle}).\\

We shall denote the diameter of a set of agents $A$ by $D(A) = \max\limits_{j,k \ \in A} \|p_j - p_k\|$. One of the following statements holds:

\begin{enumerate}
\item Agent $s$ is one of the extremal agents defining $D(\{s,N_s\})$.
\item Two other agents $i$ and $j$ are the extremal agents defining $D(\{s,N_s\})$, so that by the triangle inequality 
$$  \max \{ \| p_s - p_i \|,\| p_s- p_j \| \} \ge  \frac{D(\{s,N_s\})}{2} $$
\end{enumerate}

If $D(\{s,N_s\})$ has a strictly positive value, then there is at least one neighbour $k$ which is located at a strictly positive distance from agent $s$, i.e.:
$$ \{ \exists k \in \{s,N_s\} : \;  \| p_s - p_k \|  \ge  \frac{D(\{s,N_s\})}{2} \} $$



And what if $D(\{s,N_s\})$ tends to become infinitesimal? Since the initial constellation of the system is connected (but not yet fully connected by assumption), and by Lemma \ref{lemma55_2_1} we have that, if it once was connected, it stays connected forever, there must exist an agent $k \in N_s$ connected to an agent $j \not\in N_s$. In this case the distance between agents $j$ and $k$ is clearly very close to $V$ (i.e. $\| p_j - p_k \| \simeq V$). Hence, agent $j$ has a positive and bounded away from zero velocity, and instantly $D(\{s,N_s\})$ will immediately assume a strictly positive value. Hence in case $B_s(t) = \emptyset$ we also have in (\ref{EqBisectorOfSharpestAngle}) that $\|\dot{p}_s(t)\| > 0$.\\

We proved  that the velocity of agent $s$ is positive and bounded away from zero, and it moves into the convex-hull or along its perimeter, while all the other agents cannot move outside the convex-hull.\\

Using the observation above, we can give a formal proof that the system gathers, using the perimeter of the convex-hull as a Lyapunov function.\\

Define $\mathcal{L}(P(t))$ as the perimeter of $CH(P(t))$ and $l_i(t)$ as the current length of the convex-hull side connecting corners $i$ and $(i+1)mod(m)$ where $m$ is the current number of corners in the convex-hull.\\

Let $\varphi_i(t)$ be the angle of corner $i$ of $CH(P(t))$, let $\alpha_i(t)$ denote the direction of motion of the agent located at corner $i$ (relative to the direction of side $l_i(t)$ from $p_i(t)$ to $p_{i+1}(t)$), and let $v_i(t)=\|\dot p_i(t)\|$ be a scalar denoting the speed of the agent located at corner $i$ (as shown in Figure \ref{CHshrink_S5}).\\
\begin{figure}[h!]
\captionsetup{width=0.8\textwidth}
  \centering
    \includegraphics[width=100mm]{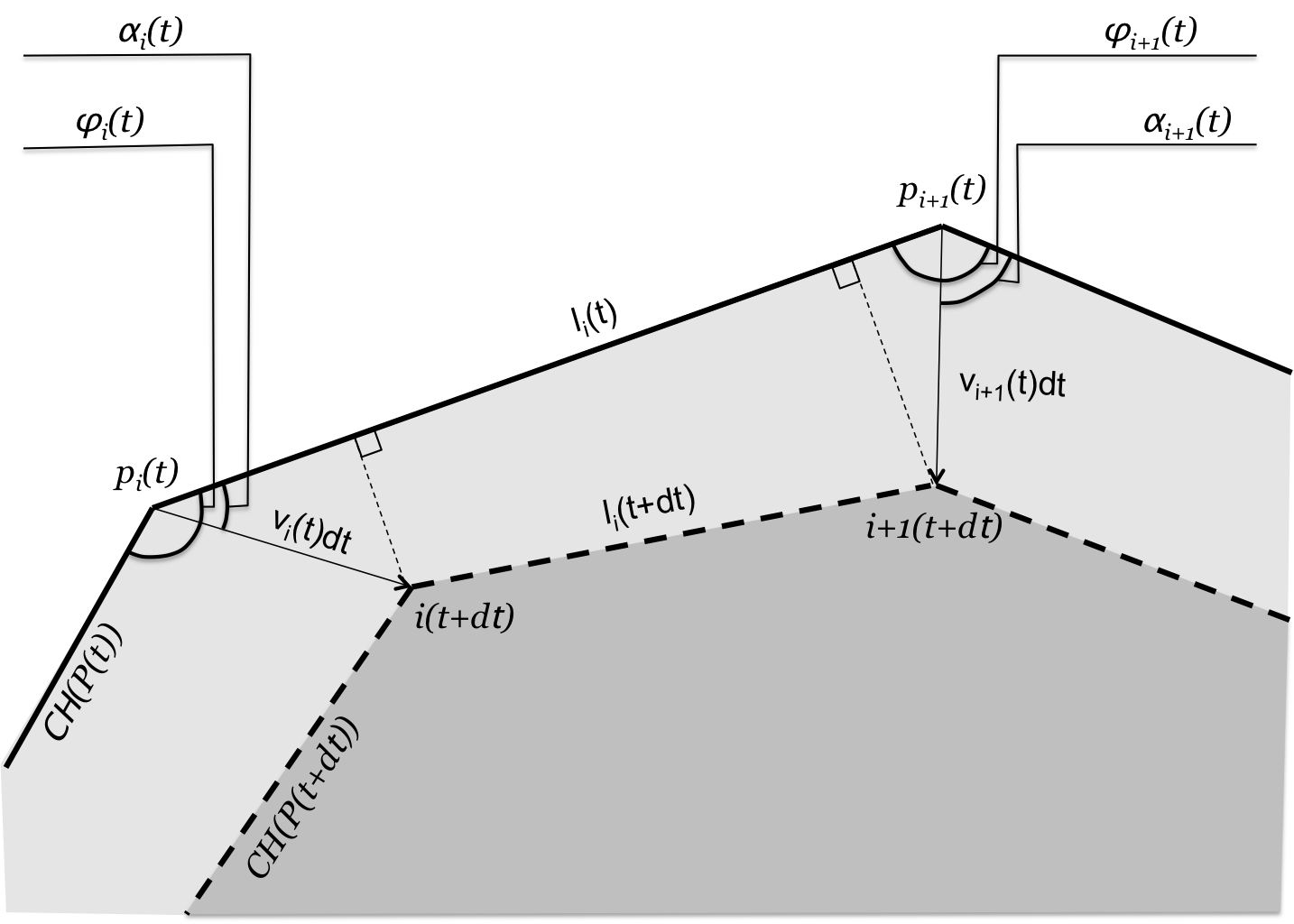}
     \caption{Convex-hull shrinkage.}
     \label{CHshrink_S5}
\end{figure}

By geometry, we have that
$$\dot l_i(t) = \left( \frac{p_{i+1}(t) - p_i(t)}{\|p_{i+1}(t) - p_i(t)\|}\right)^\intercal \left(\dot p_{i+1}(t) - \dot p_i(t) \right) = $$ 
$$=-(v_i(t) \cos \alpha_i(t) + v_{i+1}(t)\cos (\varphi_{i+1} (t) -\alpha_{i+1}(t)))$$
hence,
$$\dot{\mathcal{L}}(P(t))=-\sum\limits_{i=1}^{m} (v_i(t) \cos \alpha_i(t) + v_{i+1}(t)\cos (\varphi_{i+1} (t) -\alpha_{i+1}(t))) =$$
$$=-\sum\limits_{i=1}^{m} v_i(t)(\cos \alpha_i(t)+\cos (\varphi_{i} (t) -\alpha_{i}(t)))=$$
$$= -2\sum\limits_{i=1}^{m} v_i(t) \cos(\frac{\varphi_{i} (t)}{2}) \cos(\frac{\varphi_{i}(t) - 2\alpha_{i}(t)}{2}) $$

Since $v_i(t)$ represents the speed of movement of an agent, its value cannot be negative, and since $0 \le \alpha_i(t) \leq \varphi_i(t) \le \pi$, we have that both cosines above are positive. Therefore we have that $\dot{\mathcal{L}}(P(t)) \leq 0$, i.e. the perimeter of the convex-hull decreases.\\

Let $v_s(t)$, $\varphi_s(t)$ and $\alpha_s(t)$ be the relevant values associated with agent $s$.
Since $v_s(t) = \| \dot{p}_s(t) \|$ is positive and bounded away from zero, and since we have that $0 \le \alpha_{s}(t) \leq \varphi_s(t) \leq \pi(1-2/n)$, we have:
$$\dot{\mathcal{L}}(P(t)) \leq - v_s(t) \cos(\frac{\varphi_{s} (t)}{2}) \cos(\frac{\varphi_{s}(t) - 2\alpha_{s}(t)}{2}) $$

Therefore, while not all agents are within a range of $V-\delta$ from each other, $\dot{\mathcal{L}}(P(t))$ is strictly negative and bounded away from zero, and therefore the diameter of $CH(P(t))$ will reach $V-\delta$ in finite time.

\end{proof}

In summary, we proved that for system $\mathcal{S}_5$, any initial constellation with connected topology necessarily yields a complete graph where all the agents are within the range $V-\delta$ from each other, in finite time, and from this state the system asymptotically converges to a point, as claimed in Theorem \ref{Theorem55_2}.\\
\end{proof} 

\subsection{Discrete Time Dynamics (system $\mathcal{S}_6$)}

In this section we present the gathering process based on a dynamic rule proposed by Ando \textit{et.al} \cite{ando1999}. We survey below their results and often simplify the original proofs.\\

The system discussed here is a discrete time dynamic system which aims to gather the agents based on relative position sensing and finite visibility horizon $V$ in discrete time.\\

We assume that each agent moves according to the following dynamic law:\\

\begin{equation} \label{eq:Dynamics66}
p_i(k+1) = p_i(k) + Step_i(k) \frac{c_i(k) - p_i(k)}{\|c_i(k) - p_i(k)\|}
\end{equation}\\
where $p_i(k)$ is the position of agent $i$ at time-step $k$, and $c_i(k)$ and $Step_i(k)$ are defined below.\\
 
\begin{equation} \label{Move}
Step_i(k) = \min\{\sigma, \; Goal_i(k), \; Limit_i(k)\}
\end{equation}
is the step-size of agent $i$ at time-step $k$, which as seen in (\ref{Move}) is limited by the smallest of three factors:\\ \\
(Factor 1) The step-size limit - $\sigma$:\\ \\
The step-size limit $\sigma$ represents a physical constraint on the agents, i.e. the maximal step-size an agent can move in one time unit.\\ \\ 
(Factor 2) The current distance to goal - $Goal_i(k)$:\\ \\ 
\begin{equation} \label{Goal}
Goal_i(k)=\|p_i(k) - c_i(k)\|
\end{equation}
is the distance of agent $i$ from the point $c_i(k)$ at time-step $k$, where $c_i(k)$ is the center of circle $C_i(k)$, defined as the smallest circle enclosing agent $i$ and all its neighbours (see Figure \ref{FigAndo1}).
\begin{figure}[h!]
\captionsetup{width=0.8\textwidth}
  \centering
    \includegraphics[width=80mm]{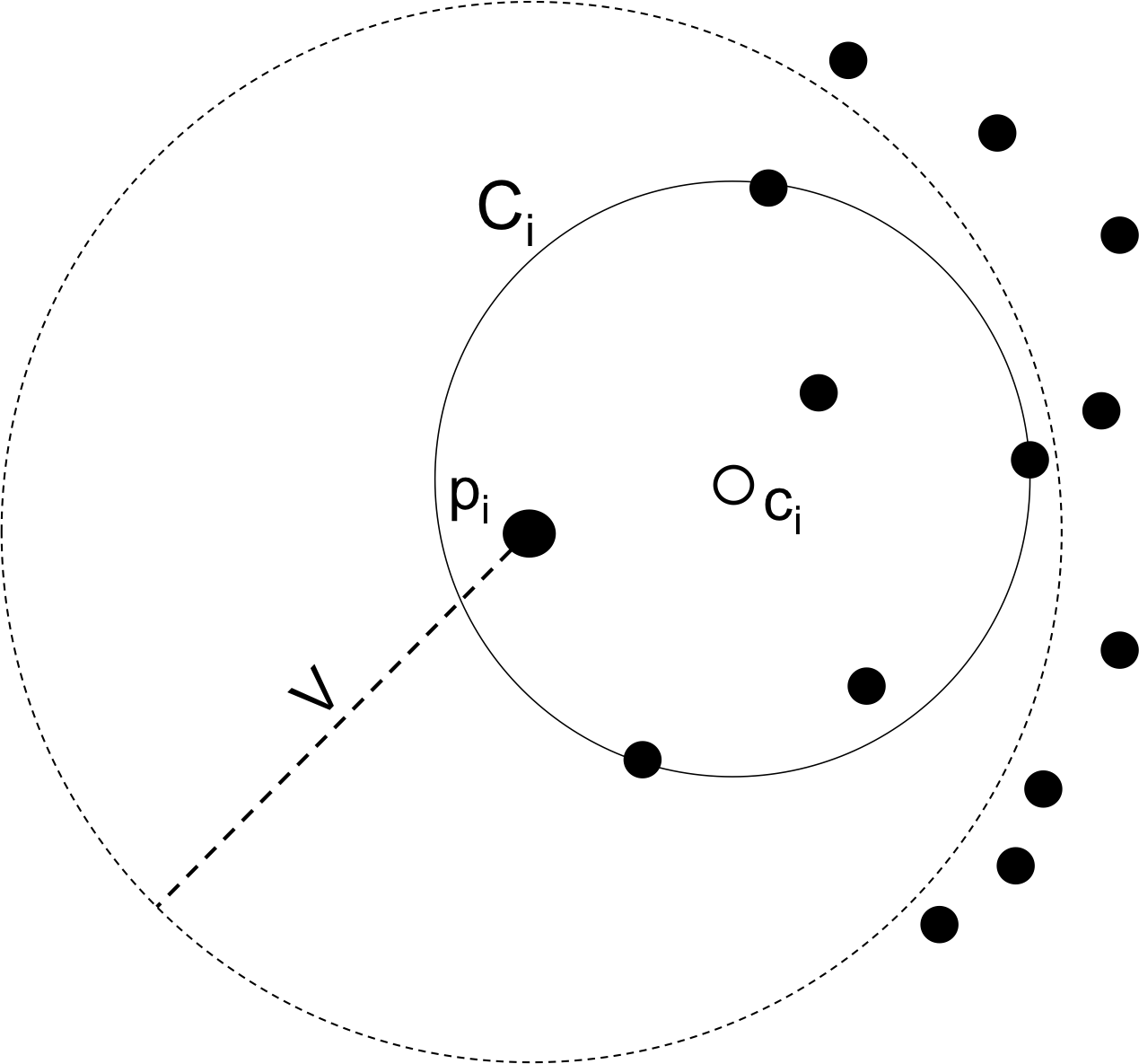}
      \caption{The dashed circle is the border line of agent $i$'s visible area, and the solid circle is the current smallest enclosing circle $C_i$ of agent $i$ and its neighbours $N_i$.}
        \label{FigAndo1}
\end{figure}\\ \\
(Factor 3) The "don't lose friends" constraint - $Limit_i(k)$:\\ \\ 
 
\begin{equation} \label{Limit}
Limit_i(k) = \min\limits_{j \in N_i}\{Limit_{ij}(k)\}
\end{equation}\\ 
is the maximal step-size allowed for agent $i$ towards $c_i(k)$ in order not to lose visibility to \textit{any} of its existing neighbours. This is defined by placing a circle of radius $V/2$ at $m_{ij}$, the midpoint between $p_i(k)$ and $p_j(k)$, limiting the area where agents $i$ and $j$ are currently allowed to move without losing visibility with each other (see Figure \ref{FigAndo3} (a)). By geometry we have that this constraint is satisfied by setting $Limit_i(k)$ to be the smallest value from the set 
$$\left\{Limit_{ij}(k)_{\substack {j \in N_i(k)}} = \frac{l_{ij}(k)}{2} cos(\theta_{ij}(k)) + \sqrt{(\frac{V}{2})^2 - (\frac{l_{ij}(k)}{2})^2sin^2(\theta_{ij}(k))}\right\}$$
where $Limit_{ij}(k)$ is the maximal step-size allowed for agent $i$ in order not to lose its neighbour $j$. The value of $l_{ij}(k)=\|p_j(k)-p_i(k)\|$ is the distance between agents $i$ and agent $j$, and $\theta_{ij}(k)$ is the angle $\angle{c_ip_ip_j}$.\\

\begin{figure}[h!]
\captionsetup{width=0.8\textwidth}
  \centering
    \includegraphics[width=100mm]{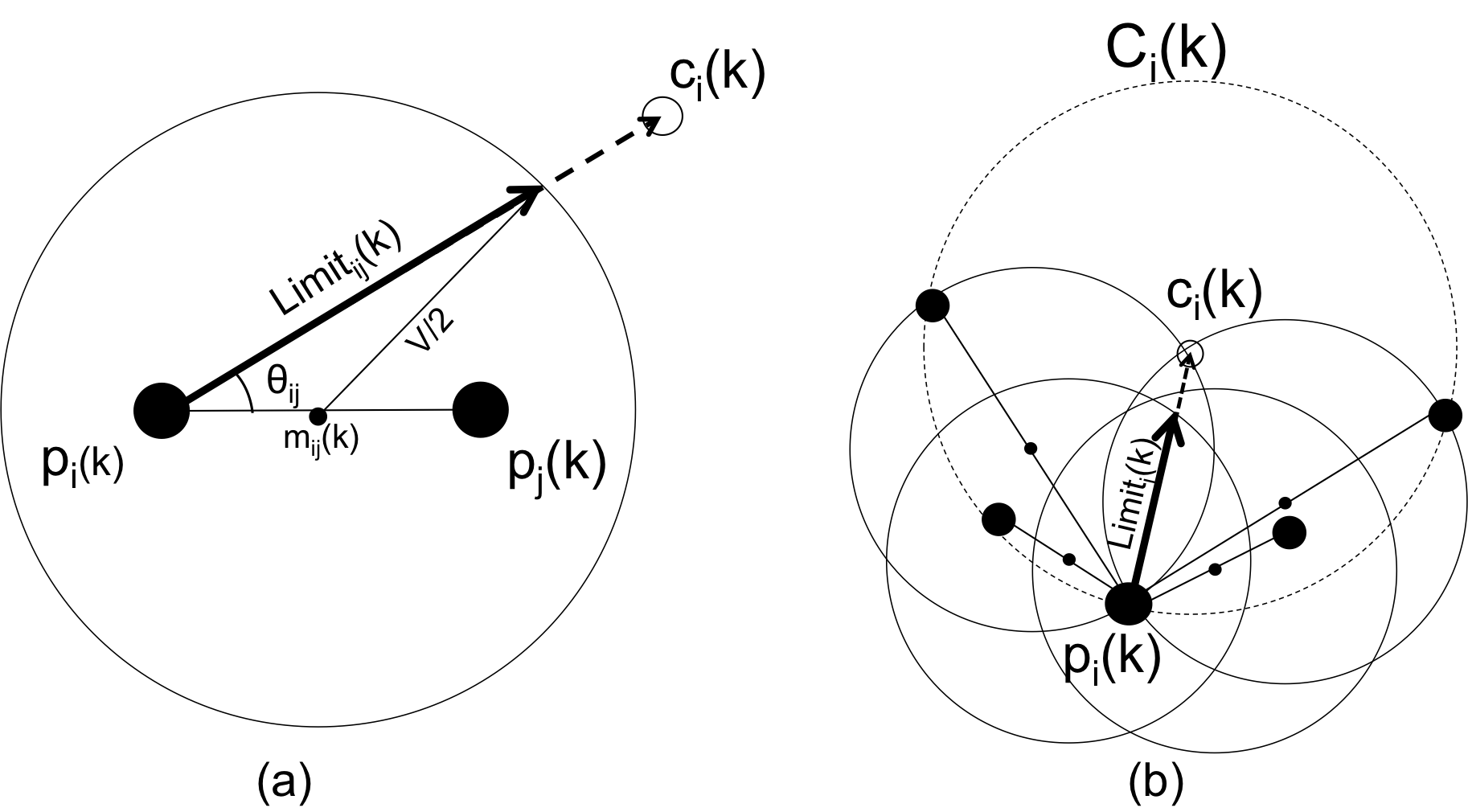}
      \caption{$Limit _{ij}(k)$ and $Limit_i(k)$:\\
      (a) $Limit_{ij}(k)$ - maximum distance agent $i$ can move towards $c_i(k)$ without leaving  a circle of radius $V/2$ centered at $m_{ij}(k)$, the average position of $p_i$ and $p_j$ at time-step $k$.\\
      (b) $Limit_i(k) = \min\limits_{j \in N_i}\{Limit_{ij}(k)\}$\\} 
          \label{FigAndo3}
\end{figure}

Hence we have that the motion law (\ref{eq:Dynamics66}) causes every agent $i$ to move toward $c_i(k)$, the center of the current circle $C_i(k)$, a step which is limited by the minimum of:

\begin{enumerate}
\item The agent's step-size limit $\sigma$.
\item The agent's distance from $c_i(k)$, the center of $C_i(k)$.
\item The intersection of all the disks of a radius $V/2$ centered at the average positions of agent $i$ and each one of the other agents within its visibility range (see Figure \ref{FigAndo3}(b)).\\
\end{enumerate}

\begin{theorem} \label{Theorem66}
For any initial constellation corresponding to a connected neighborhood graph, all the agents of system $\mathcal{S}_6$ gather to a point within finite number of time-steps.
\end{theorem}

\begin{proof} 
We shall prove Theorem \ref{Theorem66} using four lemmas related to system $\mathcal{S}_6$:
\begin{enumerate}
\item \textbf{"Never lose a neighbour":} Neighbouring agents will remain neighbours forever (Lemma \ref{NLF}).
\item \textbf{"The group stays close":} The convex-hull of system $\mathcal{S}_6$ is always contained in its previous ones, i.e. $CH(P(k+1)) \subseteq CH(P(k))$ (Lemma \ref{Lemma_ChNoGrows6}).\\
\item \textbf{"A clique gathers"}: If all agents are within the visibility range of all the other agents, so that the agents' interconnection topology is a complete graph, they will converge to a point within finite number of time-steps (Lemma \ref{GFG}).
\item \textbf{"Connected agents evolves to a clique"}: Any initial constellation with a connected neighborhood graph necessarily evolves to a complete graph within finite number of time-steps (Lemma \ref{CHshrinksLemma}).
\end{enumerate}
Using these lemmas we prove that system $\mathcal{S}_6$ gathers to a point within finite number of time-steps as claimed.\\

First we present three preliminary geometrical observations:\\
\begin{proposition} \label{PropCenerOfCircle}
Let $C$ be the smallest enclosing circle of $n>1$ points in $\mathbb{R}^2$. If only two points are located on the circumference of $C$, then the center of $C$ is located at the average position of these two points. Otherwise, there exist at least one set of three points located on the circumference of $C$, which creates an acute (or right) triangle that contains the center of the circle $C$. 
\end{proposition}

\begin{proof}
For the proof of this proposition see e.g. Reference \cite{elzinga1972geometrical}.\\
\end{proof}

\begin{proposition} \label{MinDistanceFromC}
Let $a$ and $b$ be two points on a circle of radius $R$ centered at $o$, which divide its circumference into a small arc and a big arc. Let point $c$ be the midpoint of the segment $ab$. Denote by $x$ any point on the smaller arc $ab$, and by $\beta$ the associated angle $\angle{axb}$. The minimal distance between the small arc to point $c$ is $\|a-b\|/(2\tan(\beta/2))$ being achieved when points $x$, $c$, and $o$ are collinear.
\end{proposition}

\begin{proof}
By the triangle inequality, any point $x$ described above satisfies $ \|x-c\| + \|c-o\| \ge R $, where $ \|x-c\| + \|c-o\| = R $ if and only if $\bar{xc}$ is the bisector of $\beta$, i.e. point $x$ is located so that $\triangle axb$ is an isosceles triangle (see $x'$ in Figure \ref{FigAndoProp2}).
\end{proof}     
\begin{figure}[!h]
\captionsetup{width=0.8\textwidth}
  \centering
  \includegraphics[width=60mm]{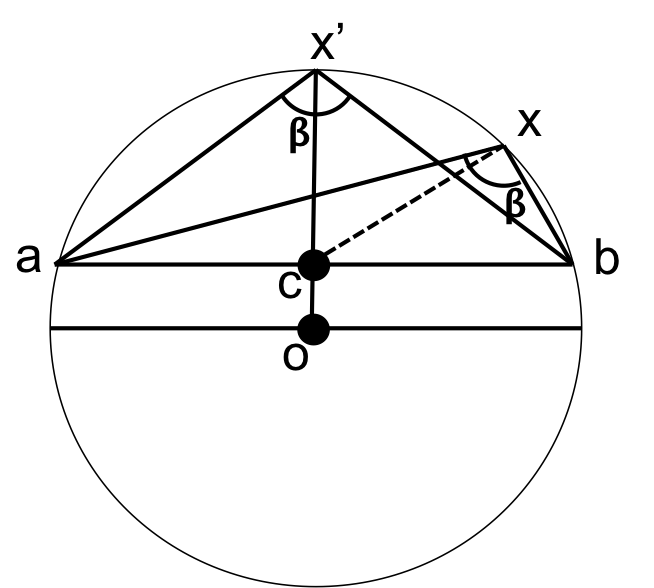}
    \caption{Proof of Proposition \ref{MinDistanceFromC}. Point $x'$ is located on the small arc between points $a$ and $b$, and has the minimal distance from point $c$. Note that $\beta  = \angle axb$ for any point $x$ on the smaller arc of the circle.}
         \label{FigAndoProp2}
\end{figure}


\begin{proposition} \label{CHinCH}
If a convex polygon $P$ is "strictly contained" in another polygon $G$ (so that the area of $P$ is strictly smaller than the area of $G$), then the length of the perimeter of $P$ is strictly smaller than the length of the perimeter of $G$.
\end{proposition}

\begin{figure}[h!]
\captionsetup{width=0.8\textwidth}
  \centering
  \includegraphics[width=50mm]{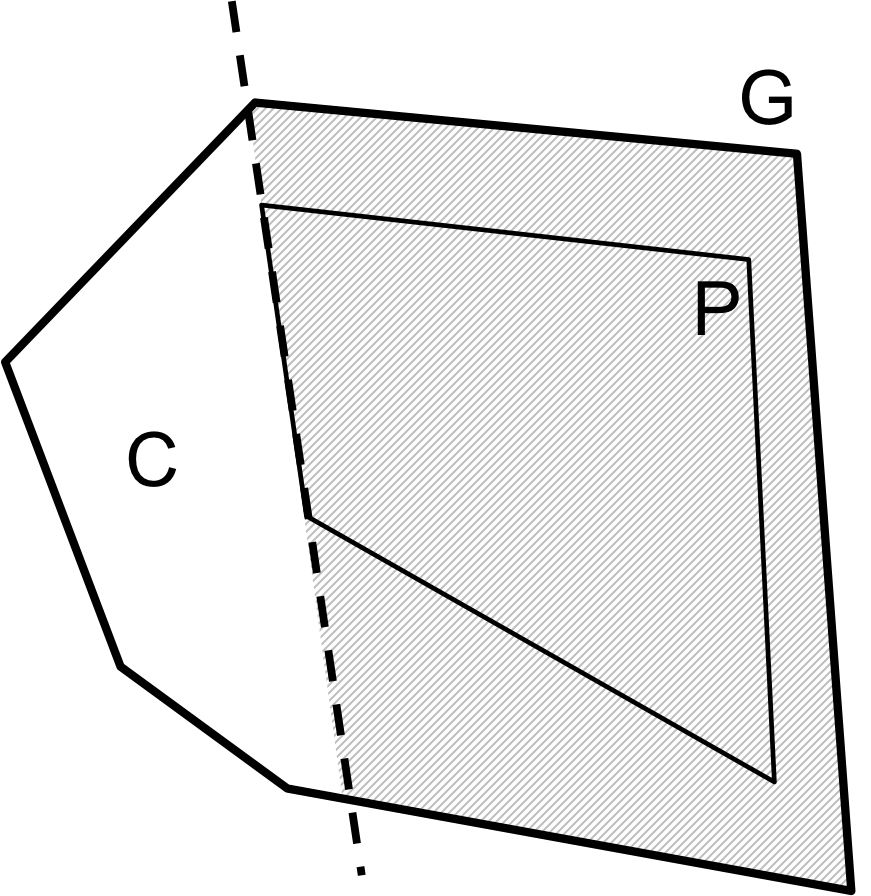}
    \caption{Intermediate polygons from $G$ to $P$. Polygon $P$ is "strictly contained" in polygon $G$. The intermediate polygon area is dashed, and the currently removed polygon is white and marked $C$.}
         \label{PolygonLineIntersecionAriel}
\end{figure}

\begin{proof}
In order to prove proposition \ref{CHinCH} we shall analyze the lengths of some intermediate polygons. We suggest a procedure to reach a convex polygon $P$ contained in another polygon $G$ as follows: Cut polygon $G$ into two polygons along a straight line defined by any side of the internal polygon $P$, mark the new polygon containing polygon $P$ as an intermediate polygon and remove the other polygon (as presented in Figure \ref{PolygonLineIntersecionAriel} by a dashed area and by $C$ respectively). Repeat this cutting procedure over and over again, beginning each iteration from the previous intermediate polygon along a straight line defined by a new side of the original polygon $P$, until only polygon $P$ remains. Since the polygon $P$ is convex, such a process is always doable.\\

Denote by $C$ the polygon removed obtained at an arbitrary intermediate step from $G$ to $P$. By the triangle inequality we have that the length of the side of the intermediate polygon $C$ defined by the current cutting line, is either equal to the sum of lengths of all other sides of $C$ (in case the area of $C$ equals zero), or smaller than the sum of lengths of all other sides of $C$ (in case $C$ has a positive area). Therefore the length of the perimeter of each intermediate polygon is either equal to or smaller than the perimeter of the previous intermediate polygon. But since the area of $P$ is smaller than the area of $G$ (being "strictly contained" in it), at least one intermediate $C$ must have a significant area, hence the length of the perimeter of the intermediate polygon associated with that intermediate $C$ will be smaller than its previous polygon. Therefore the length of the perimeter of polygon $P$ is strictly smaller than the length of the perimeter of polygon $G$.\\ 

\end{proof}

Using these propositions let us now prove Theorem \ref{Theorem66}:

\begin{lemma}  \label{NLF}
In system $\mathcal{S}_6$, if agents $i$ and $j$ are visible to each other they will forever remain visible to each other.
\end{lemma}

\begin{proof}
The step-size of any two visible agents $i$ and $j$ is limited by a disk of radius $V/2$ centered in their average position (see Figure \ref{FigAndo3}), therefore the distance between them at the next time-step, no matter where they moved inside that disc, will never exceed $V$.\\
\end{proof}


\begin{lemma} \label{Lemma_ChNoGrows6}
Let  $CH(P(k))$ be the convex-hull of all the agents in system $\mathcal{S}_6$ at time-step $k$. Then $$CH(P(k+1)) \subseteq CH(P(k))$$
\end{lemma}

\begin{proof} 
Using Proposition \ref{PropCenerOfCircle}, the center point $c_i(k)$ corresponding with the enclosing circle $C_i(k)$ of agent $i$ and its neigbours, is located in $CH(P(k))$. Since each agent $i$ is located either on or inside $CH(P(k))$ and moves towards its associated point $c_i(k)$ which is obviously inside $C_i(k)$, its next position will necessarily be either inside $CH(P(k))$ or along its perimeter. Therefore the convex-hull $CH(P(k+1))$ is contained in $CH(P(k))$.\\
\end{proof}

Note that Lemma \ref{Lemma_ChNoGrows6} only proves that $CH(P(k))$ does not expand. To prove that it necessarily shrinks to a single point within finite number of time-steps, we first verify this for complete graph topology, using the geometric Claim \ref{LimitFinitepositive} and Lemma \ref{GFG} (due to Ando \textit{et.al} \cite{ando1999}), then we prove that a connected (but not  fully connected) topology will necessarily become fully connected.\\


\begin{claim} \label{LimitFinitepositive}
If $Goal_i(k)$ has a strictly positive value, then necessarily $Limit_i(k)$ has a bounded bellow by a positive constant.
\end{claim} 

\begin{proof}
Assume $Goal_i(k)$ has a strictly positive value, but for some $j$ $Limit_{ij}(k)$ is infinitesimal. This can happen only if the distance between $p_i(k)$ and $p_j(k)$ is nearly equal to $V$, and the segment $[p_i(k)c_i(k)]$ intersects the circle of radius $V/2$ (centered at $m_{ij}(k)=(p_i(k)+p_j(k))/2$) at a point $\epsilon$-distanced from $p_i(k)$, see Figure \ref{FigAndoLimits}.\\

Note that the distance from $p_i(k)$ to $c_i(k)$ must be smaller than or equal to $V$ due to the definition of $C_i(k)$ as the \textit{minimal} enclosing circle of $p_i(k)$ and its neighbours.\\
\begin{figure}[h!]
\captionsetup{width=0.8\textwidth}
  \centering
    \includegraphics[width=65mm]{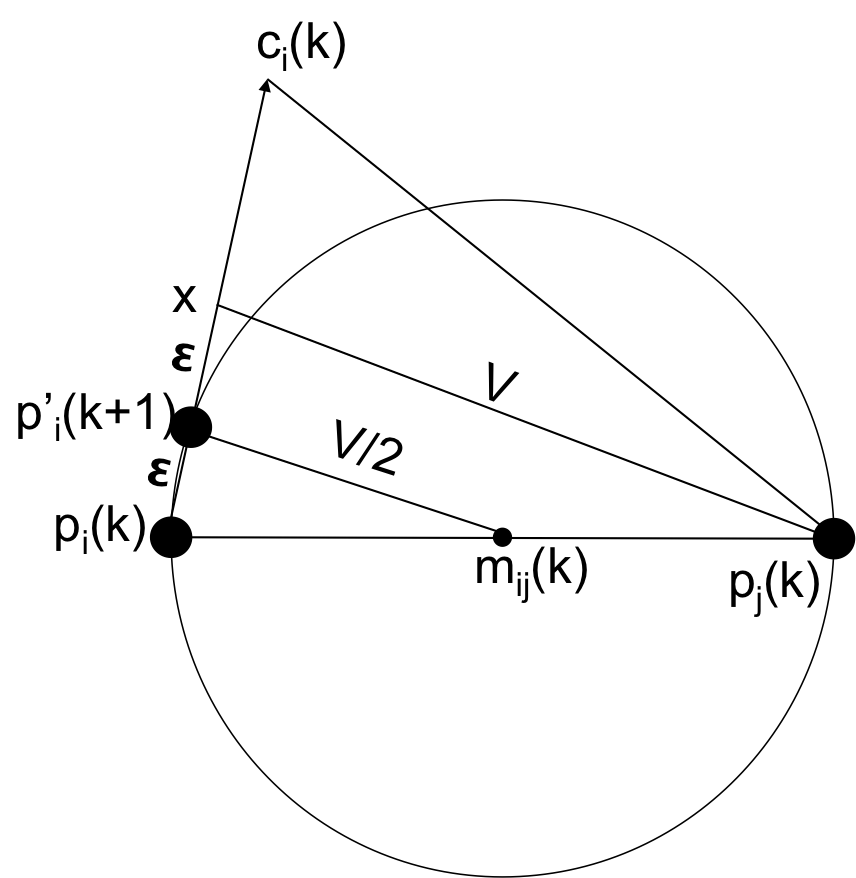}
     \caption{Geometry explaining the proof of  Claim \ref{LimitFinitepositive}.}
         \label{FigAndoLimits}
\end{figure}

Assume $Limit_{ij}(k) << V$ is infinitesimal, then there exists an agent $j$ associated with a circle of a radius $V/2$ centered at $m_{ij}$ (the average position of agents $i$ and $j$) so that $p'_i(k+1) = p_i(k) + \epsilon\frac{c_i(k) - p_i(k)}{\|c_i(k) - p_i(k)\|}$ would have been on the circumference of that circle, where $\epsilon$ has a positive infinitesimal value (see Figure \ref{FigAndoLimits}). Let $x$ be the position $p_i(k) + 2\epsilon\frac{c_i(k) - p_i(k)}{\|c_i(k) - p_i(k)\|}$. By geometry, $\| p_j(k) - x \| \cong V$ and by geometry the angle $\angle p_jxc_i \ge \pi/2$. This yields $\| p_j(k) - c_i(k) \| > V$ which is not possible since the maximal radius of $C_i(k)$ is $V$.\\

Therefore if $Goal_i(k)$ has a strictly positive value, then $Limit_{ij}(k)$ necessarily has a strictly positive and bounded away from zero by a constant (independent of $k$) value for all $j \in N_i$, hence by (\ref{Limit}) we have that $Limit_i(k)$ has a strictly positive and bounded away from zero by a constant value as well.

\end{proof} 

\begin{lemma}  \label{GFG}
If all the agents of system $\mathcal{S}_6$ are within the visibility range of each other, they will converge to a point within finite number of time-steps.\\
\end{lemma}

\begin{proof} If all the agents are visible to each other, all circles $C_i(k)$ coincide, hence their centers are identical. Denote this circle by $C(k)$ and its center by $c(k)$ i.e. $\{\forall i \; : \; c_i(k) \triangleq c(k)\}$. By Claim \ref{LimitFinitepositive} we know that in this case all the agents jump a step of a strictly positive size towards point $c(k)$, and therefore the radius of $C(k)$ decreases by a strictly positive value at each time-step.\\
\end{proof}

Notice that when the radius of $C(k)$ is less than or equal to $V/2$, all points $m_{ij}$ (the average position of any two agents $i$ and $j$) are inside $C(k)$ (see Fig. \ref{AndoCompleteFig}), i.e $\| m_{ij} - c(k)\|  \le V/2$, hence point $c(k)$ is necessarily inside the circle that defines $Limit_{ij}(k)$ . Therefore, for all agents $i$, $Goal_i(k) \le Limit_{i}(k)$ for any agent $i$. Hence, $Limit_{i}(k)$ does not play a role defining the step-size of agent $i$, and we have that

$$Step_i (k) = \min\{\sigma, \; Goal_i(k)\}$$

Since all the agents share the same goal point $c_i$, and move with the same finite step-size $\sigma$ (or $Goal_i(k)$ in case $Goal_i(k)<\sigma$) toward the same point $c_i$, the radius of $C(k)$ decreases in each time-step by $\sigma$ (or $Goal_i(k)$). Therefore all agents converge to a point within at most $V/(2\sigma)$ time-steps. This is true since we deal with synchronous dynamics, i.e all agents jump simultaneously.\\
\begin{figure}[h!]
\captionsetup{width=0.8\textwidth}
  \centering
    \includegraphics[width=65mm]{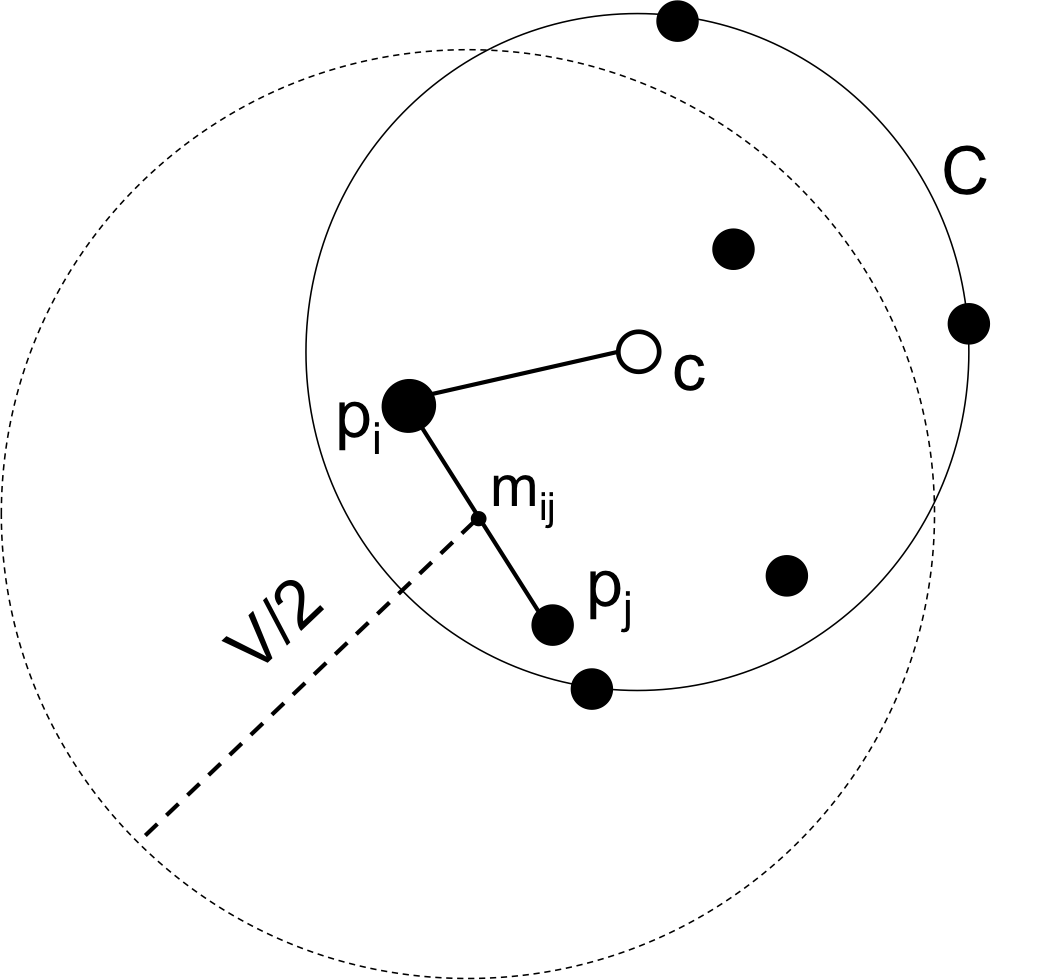}
     \caption{For a complete graph the solid circle is the smallest enclosing circle $C(k)$ of the agents. The dashed circle of radius $V/2$ represents the limit required for agent $i$ to maintain connectivity with agent $j$.}
         \label{AndoCompleteFig}
\end{figure}

We now discuss the case where the neighborhood graph topology is connected but not fully connected, showing it necessarily becomes fully connected within finite number of time-steps.  

\begin{lemma}\label{CHshrinksLemma}
The perimeter of the convex-hull of system $\mathcal{S}_6$ decreases until all agents gather to a fully connected constellation within finite number of time-steps.
\end{lemma}

\begin{proof}
To prove Lemma \ref{CHshrinksLemma} we shall use three claims. By Claim \ref{LimitFinitepositive} above we already have that for every agent $i$ if $Goal_i$ is bounded away from zero, then $Limit_i$ is bounded away from zero as well. Next we show that there is always at least one corner of the convex-hull occupied by an agent whose current or next associated $Goal$ is bounded away from zero (Claims \ref{Prop61} and \ref{GoalAlphaIsFinitePositive}). Therefore, the perimeter of the convex-hull necessarily decreases within finite number of time-steps, until the agents' constellation has a fully connected visibility graph.\\


Let us examine the dynamics of the agent denoted by $s$, the current sharpest corner of the convex-hull:\\
  
Let $Q(k)$ be the set of vertices that define $CH(P(k))$ the convex-hull of the system at time-step $k$, and $\varphi_i(k)$ be the associated angle of each such vertex $p_i(k) \in Q(k)$. By geometry, the sum of all $\varphi_i(k)$ is $\pi (|Q(k)|-2)$. Let $s$ be the agent located at $p_s(k)$, the current sharpest corner of the convex-hull, and let $\varphi_s(k)$ be the associated inner angle of that corner, so that $\varphi_s(k) \le \pi - \frac{2\pi}{|Q(k)|} \le \pi - \frac{2\pi}{n}=\varphi_*$.

\begin{claim} \label{Prop61}
Let $R(C_s(k))$ be the current radius of the smallest enclosing circle of $s$ and its neighbours. If $R(C_s(k))$ is bounded away from zero, then $Goal_{s}(k)$ is bounded away from zero as well.
\end{claim} 

\begin{proof}

If agent $s$ is currently located on the circumference of $C_s(k)$ then trivially
$$Goal_s(k) = R(C_s(k))$$
otherwise, referring to Figure \ref{AndoAlpha}, let points $a$ and $b$ be the intersection points of circle $C_s(k)$ and the rays directed from $p_s(k)$ to its neighbouring corners in the convex-hull $p_{s-1}(k)$ and $p_{s+1}(k)$ respectively. By Proposition \ref{PropCenerOfCircle} we have that only two geometrical cases may exist:\\ \\
(Case 1): Only two agents $i$ and $i'$ define $C_s(k)$ as they are located on two sides of its diameter, where then obviously $\angle isi' \ge \pi/2$ (since agent $s$ located inside $C_s(k)$ and $\varphi_s \ge  \angle isi'$)\\ \\
(Case 2): Three agents $i$, $j$ and $q$ define $C_s(k)$ as they are located on its perimeter, and they create an acute triangle that contains the center of $C_s(k)$. Here too, obviously, $\angle p_ip_sp_j \ge \pi/2$ since agent $s$ is located inside $C_s(k)$, and angle $\angle p_ip_sp_j$ contains at least a pair of points which define a diameter of $C_s(k)$ (See Figure \ref{AndoAlpha}).
Let $a'$ be the point diametrically opposed to $a$, so that the segment $aa'$ is a diameter of $C_s(k)$, then by geometry point $a'$ is located on the bigger arc $ab$, hence the following is true: 
$$\pi/2 \le \angle asa' \le \angle asb = \varphi_s(k) \le \varphi_*$$
By Proposition \ref{MinDistanceFromC}, the minimal distance between $p_s$ and $c_s(k)$ is equal to or greater than ${R(C_s(k))}/{\tan(\frac{\angle asa'}{2})}$. Therefore,
$$ Goal_s(k) \geq \frac{R(C_s(k))}{tan(\frac{\angle asa'}{2})} \geq \frac{R(C_s(k))}{tan(\frac{\varphi_*}{2})} $$
hence if $R(C_s(k))$ is bounded away from zero, $Goal_s(k)$ is bounded away from zero as well.\\
\end{proof}
\begin{figure}[h!]
\captionsetup{width=0.8\textwidth}
  \centering
    \includegraphics[width=58mm]{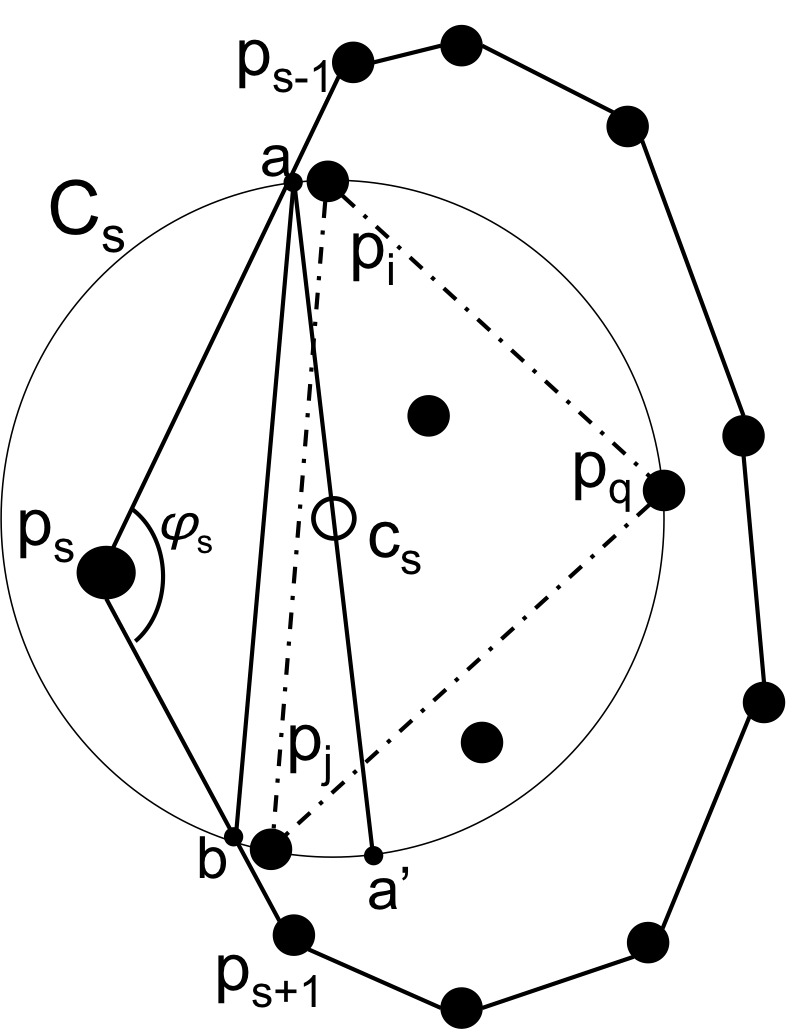}
     \caption{Convex-hull corners with agent $s$ occupying the sharpest corner; Circle $C_s$ is defined by agents $i$, $j$ ans $q$, and its center $c_s$ is confined in $\triangle p_ip_jp_q$; Points $a$ and $a'$ define the diameter of $C_s$ and are located on the bigger arc $ab$.}
         \label{AndoAlpha}
\end{figure}
And what if $R(C_s(k))$ has an infinitesimal value? In this case the following is true:\\

\begin{claim} \label{GoalAlphaIsFinitePositive}
While the neighbourhood graph topology of system $\mathcal{S}_6$ is connected but not fully connected, if at time-step $k$ the radius of $C_s(k)$ has an infinitesimal length $\epsilon$, then, at the next time-step, the radius of $C_s(k+1)$ will necessarily be bounded away from zero. 
\end{claim} 

\begin{proof}
If the radius of $C_s(k)$ is infinitesimal, and the agents interconnection topology is connected but not fully connected, then necessarily there is an agent $j$ in the neighbourhood of agent $s$ which is also connected to agent $i$, who is not in the neighbourhood of agent $s$. Therefore, the distance between agents $j$ and $i$ is between $V-\epsilon$ to $V$ (See Figure \ref{InfineighbourhoodOfs})\\ 

Since the distance between $p_j(k)$ and $p_s(k)$ is infinitesimal, we may consider that agent $j$ is located at the same location as agent $s$. But we know that the radius of $C_j(k)$ is greater than or equal to $(V-\epsilon)/2$, hence Claim \ref{Prop61} holds for agent $j$, and therefore it moves a step of a strictly positive length inside the system's convex-hull, while agent $s$ stays put.\\

\end{proof}

\begin{figure}[h!]
\captionsetup{width=0.8\textwidth}
  \centering
    \includegraphics[width=75mm]{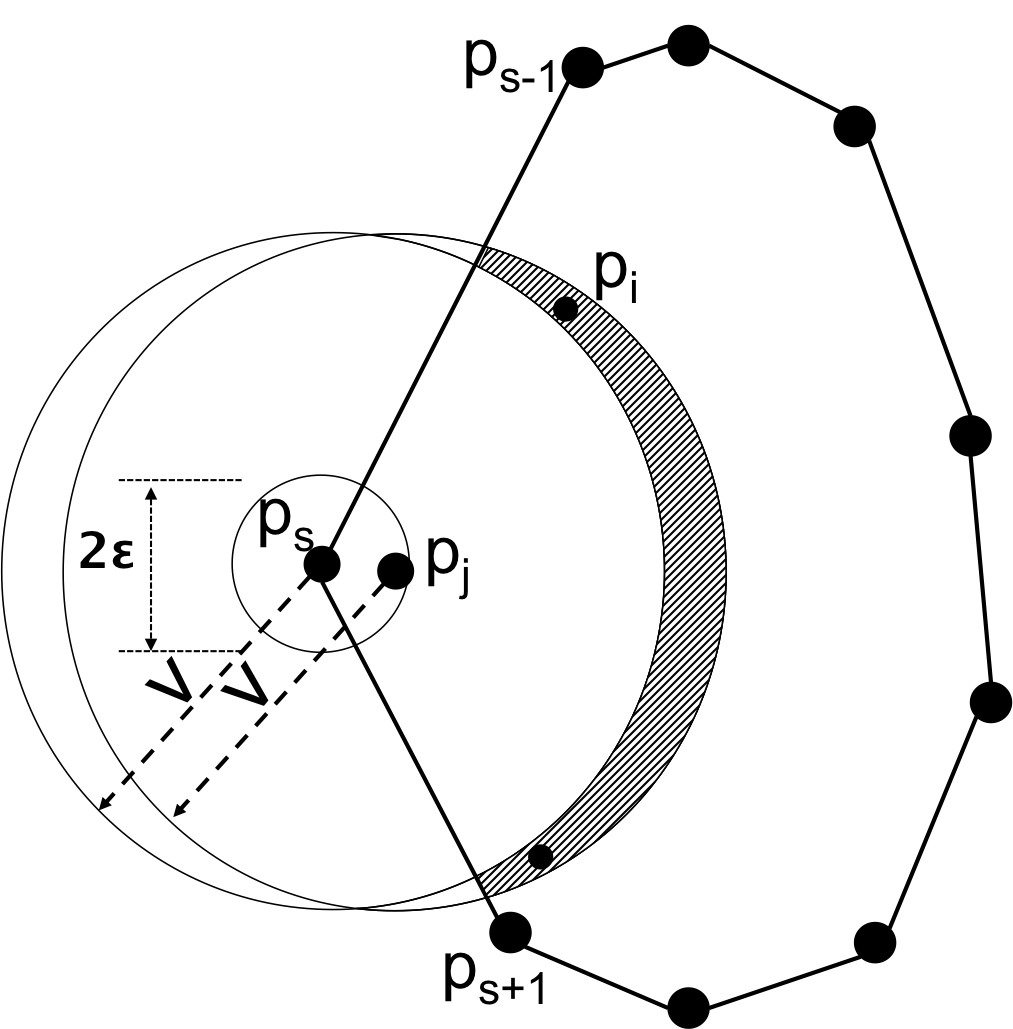}
     \caption{Proof of claim \ref{GoalAlphaIsFinitePositive}. The small circle represents the $\epsilon$-neighbourhood of agent $s$, the dashed region is area $J_N(k)$ which contains the neighbours of $j$ which are not neighbours of agent $s$.}
         \label{InfiNeighbourhoodOfs}
\end{figure}

Returning to the proof of Lemma \ref{CHshrinksLemma}, we have that if $R(C(k))$ is strictly positive, then by Claim \ref{Prop61} we have that $Goal_s(k)$ is strictly positive, or if currently $Goal_s(k)$ has an infinitesimal value, by Claim \ref{GoalAlphaIsFinitePositive} at the next time step $Goal_s(k+1)$ will be strictly positive. Using Claim \ref{LimitFinitepositive} and Lemma \ref{Lemma_ChNoGrows6}, agent $s$ and all the agents located within an infinitesimal distance from it jump a step of a strictly positive size at the current time-step or at the next one. Therefore, the current sharpest corner jumps a strictly positive distance inside the convex-hull at the current time-step or at the next.\\

If the area of $CH(P(k))$ is strictly positive, then, following its current sharpest corner (which was proved to move inside the convex-hull a strictly positive distance) and by Lemma \ref{Lemma_ChNoGrows6} (which proved that for any time step $k$ we have that $CH(P(k+1)) \subseteq CH(P(k))$), the area of $CH(P(k+1))$ and/or the area of $CH(P(k+2))$ is significantly smaller than the area of $CH(P(k))$. Hence, by Proposition \ref{CHinCH}, the length of the perimeter of $CH(P(k+1))$ and/or $CH(P(k+2))$ is significantly smaller than that of $CH(P(k))$.\\ \\And if all agents are collinear, so that $CH(P(k))$ is one dimensional with an area equal zero (yet a perimeter length bigger than zero), obviously the length of the perimeter of $CH(P(k+1))$ and/or $CH(P(k+2))$ is significantly smaller than that of $CH(P(k))$ since the sharpest corners are located at both ends of $CH(P(k))$). This proves Lemma \ref{CHshrinksLemma}\\

\end{proof}

In summary, we proved that the convex-hull of system $\mathcal{S}_6$ necessarily shrinks until all agents evolve to a constellation of points having a complete visibility graph within finite number of time-steps from any initial constellation of connected visibility topology (Lemma \ref{CHshrinksLemma}). We also proved that if the interconnection graph topology is complete, all the agents converge to a point within finite number of time-steps (Lemma \ref{GFG}). Therefore, for any initial constellation of connected topology all the agents converge to a point within finite number of time-steps as claimed in Theorem \ref{Theorem66}\\
\end{proof} 


\subsection{Discussion}



\textbf{Agents' trajectories:} Simulation results of the dynamics of system $\mathcal{S}_6$ for an arbitrary initial constellation is shown in Figure \ref{FigAndoCHSim}. Note that as long as the agents interconnection topology is not complete, the agents' trajectories are neither straight nor smooth in the sense of having significant discontinuities in the directions of motion of agents, then at some $k = \kappa$ when the topology becomes complete and all agents are contained in a disc of radius $V/2$, the trajectories of all agents become straight. This is explained by the fact that from  time step $\kappa$ on, all agents share the same smallest enclosing circle $C(k)$, hence the same center point $c(k)$ as their destination. Since all agents jump the same step-size ($\sigma$, or $Goal$ if $Goal$ is less than $\sigma$) at each time-step, from this time-step on the interior agents remain during their motion in the circle $C(k)$, and the exterior agents define circle $C(k)$ for all $k \ge \kappa$. Those agents jump exactly the same step-size ($\sigma$, or $Goal$ if $Goal$ is less than $\sigma$) towards $c(k)$ like all other agents, while the point $c(k)$ remains stationary, i.e. $c(k) \triangleq c(\kappa)$.\\

\begin{figure}[H]
\captionsetup{width=0.8\textwidth}
\centering
    \includegraphics[width=120mm]{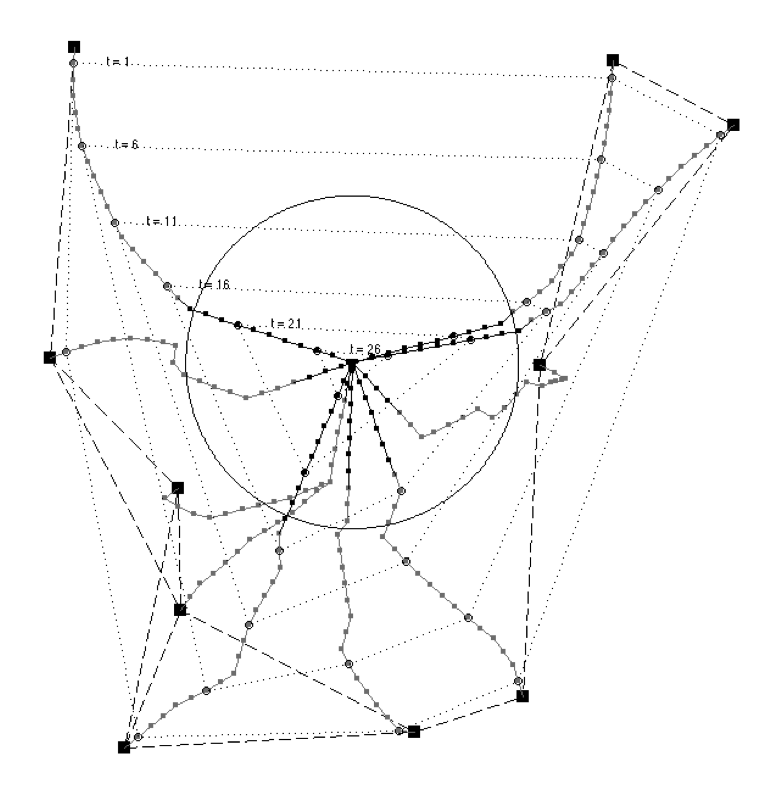}
    \caption{Simulation result of system $\mathcal{S}_6$ with 10 agents. The bold squares represent the initial configuration of the agents, and the dashed lines connecting them represent the initial agents interconnection topology. The small squares (here connected with short straight lines) represent the positions of the agents at each time step, simulating their trajectories. Once all the agents gathered into a circle of radius $V/2$ shone above, their interconnection topology becomes complete, causing each one of them to move along a straight line (in black).}
      \label{FigAndoCHSim}
\end{figure}


Notice that the assumption in system $\mathcal{S}_6$ is that all agents are capable of carrying out significant geometric computation at each step of the algorithm, while keeping the obliviousness of their behaviour (i.e. agents respond to the present situation only, while being, of course, identical and anonymous).\\ 

\newpage
\section{Finite Visibility, Bearing-only Sensing}

In this section we assume that every agent has information only about the \textit{relative bearing direction} to its neighbours within a visibility range $V$ (but cannot measure relative distances).

\subsection{Continuous Time Dynamics (system $\mathcal{S}_7$)}

First we present the gathering process based on bearing-only sensing with limited visibility in the continuous time framework proposed by Gordon \textit{et.al} in \cite{gordon2004}.\\

Assume each agent moves with a constant speed $\sigma$ unless it stops, according to the following dynamic law:

\begin{equation} \label{eq:Dynamics77}
 \dot p_i(t)=\left\{\begin{alignedat}{2}
    & \sigma\ \hat p_{i_{bisector}}(t), && \quad\psi_i(t)<\pi \\
    & 0, && \quad o.w.
  \end{alignedat}\right. 
 \end{equation}
where the unit-vector $\hat p_{i_{bisector}}(t)$ is defined via:
$$\hat p_{i_{bisector}}(t) =\frac{\hat p_{i_{R}}(t)+\hat p_{i_{L}}(t)}{\|\hat p_{i_{R}}(t)+\hat p_{i_{L}}(t)\|}$$
the points $p_{i_{R}}(t)$ and $p_{i_{L}}(t)$ being the positions of the extreme right and left agents defining the minimal angular sector containing all neighbours of agent $i$. This sector angle will be denoted by $\psi_i = \angle p_{i_{R}} p_i p_{i_{L}}$ (see Figure \ref{FigGordon2}).\\

The vectors $\hat p_{i_{R}}(t)$ and $\hat p_{i_{L}}(t)$ are unit vectors from $p_i(t)$ to $p_{i_{R}}(t)$ and to $p_{i_{L}}(t)$ respectively, denoted by

$$\hat p_{i_{R}}(t) = \frac{p_{i_{R}}(t)-p_i(t)}{\|p_{i_{R}}(t)-p_i(t)\|}$$
and

$$\hat p_{i_{L}}(t) = \frac{p_{i_{L}}(t)-p_i(t)}{\|p_{i_{L}}(t)-p_i(t)\|}$$\\
hence, agent $i$ moves at a constant velocity $\sigma$ along the bisector of $\psi_i(t)$, unless $\psi_i(t) \geq \pi$ then agent $i$ doesn't move (shown in Figure \ref{FigGordon2}).\\

\begin{figure}[h!]
\captionsetup{width=0.8\textwidth}
  \centering
    \includegraphics[width=100mm]{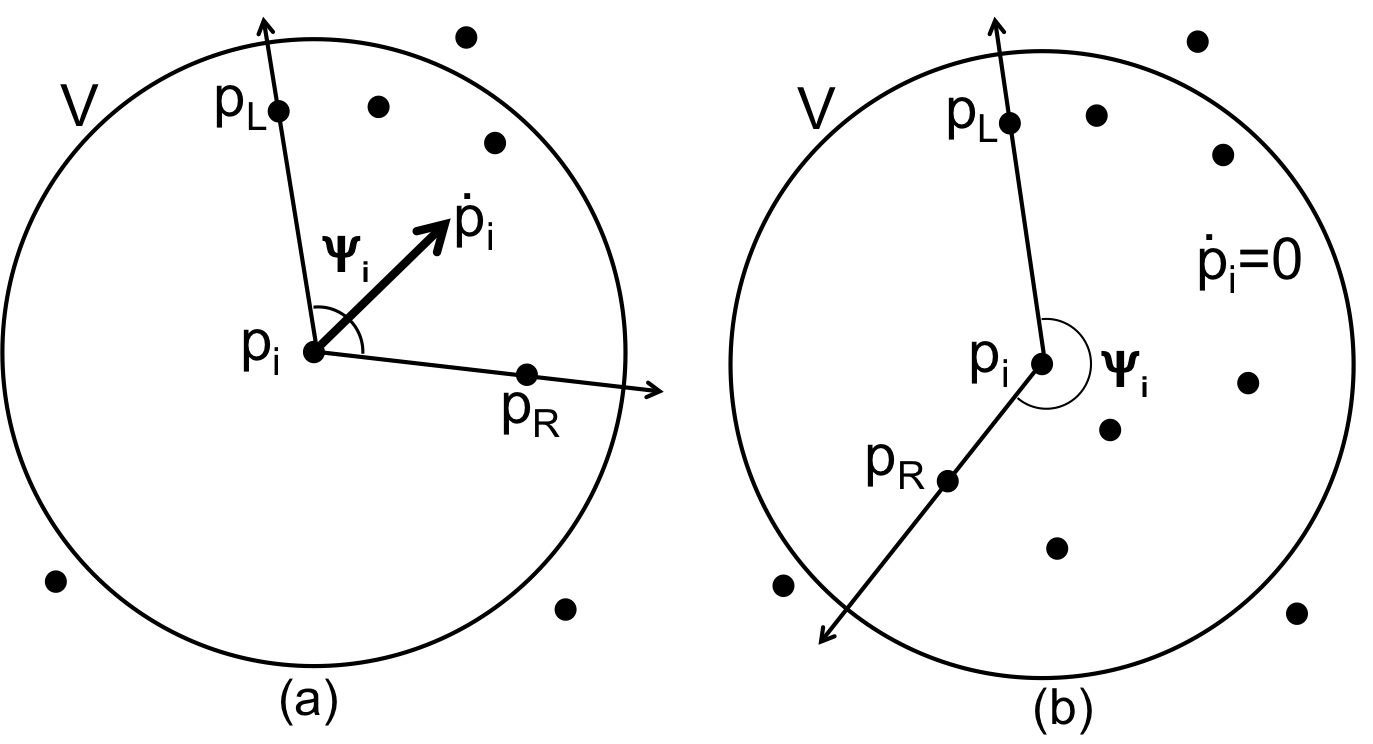}
    \caption{(a) Agent $i$ moves in the direction of $\psi_i$'s bisector. (b) Agent $i$ does not move since $\psi_i > \pi$}
    \label{FigGordon2}
\end{figure}

Note that there can be a discontinuity in the velocity of the agents, e.g. the movement direction of an agent may change in zero time and an agent may decide to stop.\\

An alternative law of motion with proportional velocity dictates that
\begin{equation} \label{eq:Belaish}
 \dot p_i(t)=\left\{\begin{alignedat}{2}
    & \sigma(p_{i_{R}}(t)+p_{i_{L}}(t)), && \quad\psi_i(t)<\pi \\
    & 0, && \quad o.w.
  \end{alignedat}\right. 
 \end{equation}
This dynamic rule is analysed in details in \cite{bellaiche2015}.
\begin{theorem} \label{Theorem77}

For an initial constellation of a connected neighborhood graph, all the agents of system $\mathcal{S}_7$ converge to a point in finite time.

\end{theorem}

\begin{proof} 
In the following section we shall prove Theorem \ref{Theorem77} using three steps:
\begin{enumerate}
\item Any two neighbours of system $\mathcal{S}_7$ will remain neighbours (Lemmas \ref{NeverLoseFriendLemma} and \ref{NeverLoseFriendUnderDynm}).
\item At any time $t$ the convex-hull of system $\mathcal{S}_7$ is contained in its previous (Lemma \ref{GordonCHnotGrow}).
\item While the convex-hull of the agents locations has a perimeter greater than zero, the perimeter decreases at a finite speed (Lemma \ref{GordonCHshrink}).
\end{enumerate}

These three lemmas prove that system $\mathcal{S}_7$ converges to a point in finite time as claimed.\\

Let us first define an \textit{allowable region} $AR_i(t)$ where each agent $i$ can move without losing any existing neighbour, i.e. its distance from every existing neighbour will stay smaller than $V$.\\

Denote by $D_r(c)$ a disc of radius $r$ centered at point $c$, and denote by $c_{ij}(t) = p_i(t) + \frac{V}{2}\frac{p_j(t)-p_i(t)}{\|p_j(t)-p_i(t)\|}$ a point distanced $\frac{V}{2}$ from $p_i(t)$ in the relative direction to $p_j(t)$.\\

The current allowable region, where agent $i$ may move without losing any of its existing neighbours, is denoted by (see Figure \ref{FigGordon1}):

\begin{equation} \label{NeverLoseFriendRig}
AR_i(t) \triangleq  \left( \bigcap \limits_{j \in N_i(t)}  D_{\frac{V}{2}}(c_{ij}(t)) \right) \cap D_{\frac{V}{2}}(p_i(t))
\end{equation}

\begin{figure}[h!]
\captionsetup{width=0.8\textwidth}
  \centering
    \includegraphics[width=120mm]{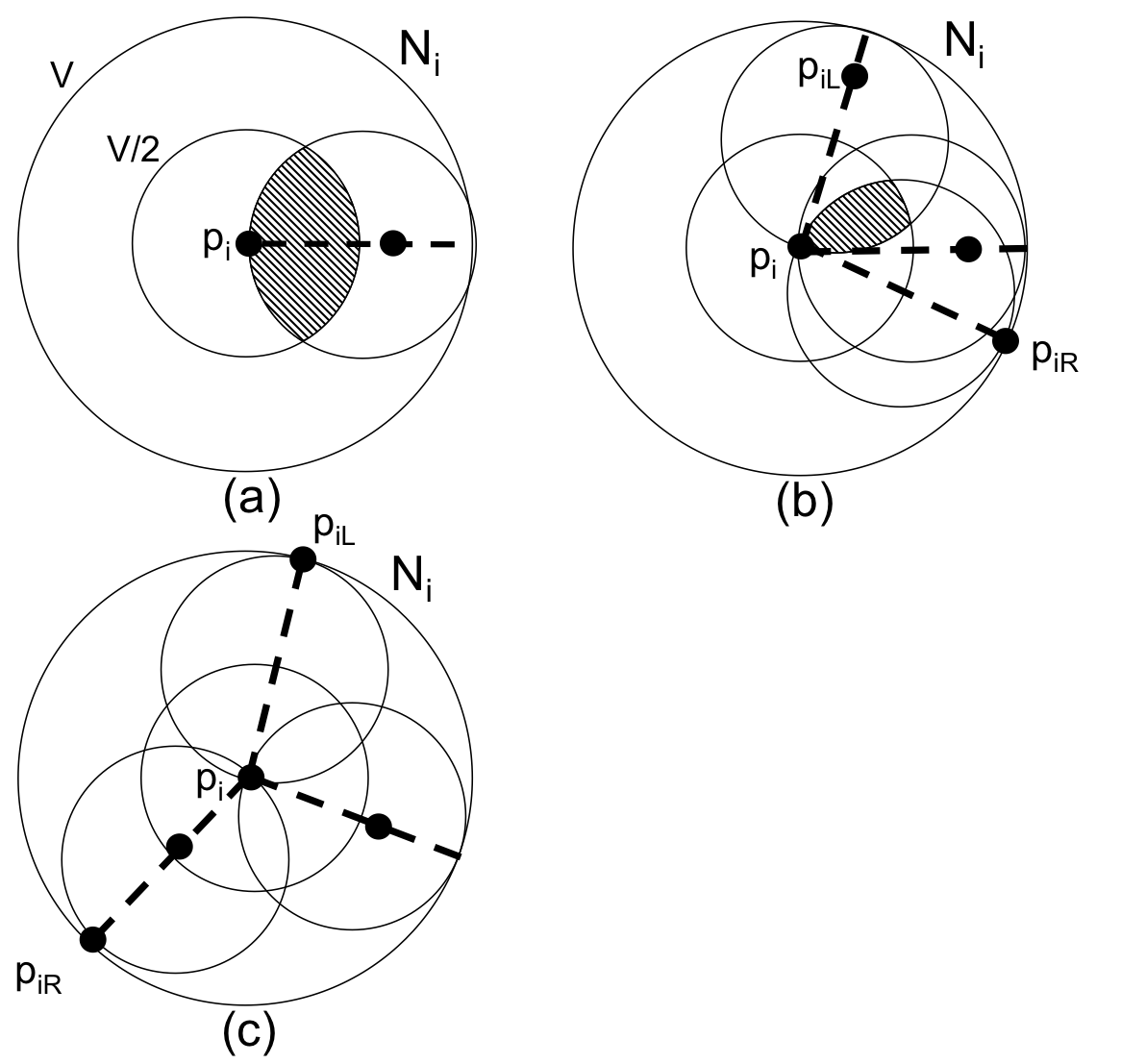}
    \caption{Allowable regions for agent $i$. (a) Single neighbour. (b) Intersection between the extreme left agent's disc, the extreme right agent's disc, and the disc $D_{\frac{V}{2}}(p_i)$. (c) No allowable region since the intersection yields an empty area.}
      \label{FigGordon1}
\end{figure}

\begin{lemma} \label{NeverLoseFriendLemma}
If all agents move to an arbitrary location inside their allowable regions they will not lose any of their neighbours.
\end{lemma}

\begin{proof}
Considering an agent $i$, we realize that if it sees an agent $j$ in a given direction, the agent $j$ will be somewhere at a distance less than $V$ from it. If the agent is at a distance $V$, then clearly both $i$ and $j$ can move into a disc of radius $V/2$ centered at their average location $(p_i(t)+p_j(t))/2$ without losing mutual visibility. If $j$ will be at a distance less than $V$ from $i$ then they can again move into a disc of a radius $V/2$ centered at the average of their locations. Hence we have that the intersection of all possible moves for agent $i$ due to possible locations of agent $j$ within $r<V$ distance from agent $i$, in the direction to $j$ (known to $i$) is given by
$$
AR_{ij}(t)=\bigcap \limits_{r=o}^{V}  D_{\frac{V}{2}}\left( p_i(t)+\frac{1}{2}\frac{p_j(t)-p_i(t)}{\|p_j(t)-p_i(t)\|}r \right)=
$$
$$
=D_{\frac{V}{2}}( p_i(t))\bigcap D_{\frac{V}{2}}\left(p_i(t)+\frac{1}{2}\frac{p_j(t)-p_i(t)}{\|p_j(t)-p_i(t)\|}V\right)=
$$
The allowable region for $i$ to move will be
$$
AR_i(t)= \bigcap \limits_{j \in N_i(t)} AR_{ij}(t)
$$
hence we get the formula \ref{NeverLoseFriendRig}.\\

Therefore, for any pair of neighbours $i$ and $j$, if both $i$ and $j$ move into their allowable region, we have that $AR_i(t)$ and $AR_j(t)$ are contained in $D_{V/2}((p_i(t)+p_j(t))/2)$, hence the distance between them remains within $V$.\\

Note that if agents $j \in N_i(t)$ surround agent $i$ (i.e. $\psi_i(t)>\pi$), the allowable region will be empty hence agent $i$ cannot move without risking disconnecting visibility to some of its neighbours.\\
\end{proof}

We still have to show that agents with dynamics defined by (\ref{eq:Dynamics77}) necessarily move into the allowable region.
\begin{lemma}\label{NeverLoseFriendUnderDynm}
Motion according to the dynamics rule (\ref{eq:Dynamics77}) ensures that all agents move into their allowable regions.
\end{lemma}

\begin{proof}
We provide two different argument to prove this.\\ \\
\textbf{First argument:}\\ \\
Let us examine dynamics (\ref{eq:Dynamics77}) using Lemma \ref{NeverLoseFriendLemma}. Let $Limit_{i}(t)$ and $Limit_{ij}(t)$ be the length of the section segments of $AR_i(t)$ and $AR_{ij}(t)$, starting at $p_i(t)$ with the direction $p_{i_{bisector}}(t)$, so that
$$ Limit_i(t) =  \min\limits_{j \in N_i}\{Limit_{ij}(t)\} $$
By geometry,
$$ Limit_{ij}(t) =  \min\{V/2,V\cos(\theta_{ij}(t))\} $$
where $\theta_{ij}(t)$ is the angle between the vectors $p_j(t)-p_i(t)$ and $\hat p_{i_{bisector}}(t)$ (see Figure \ref{Limit_ij_Gordon}).
\begin{figure}[H]
\captionsetup{width=0.8\textwidth}
  \centering
    \includegraphics[width=70mm]{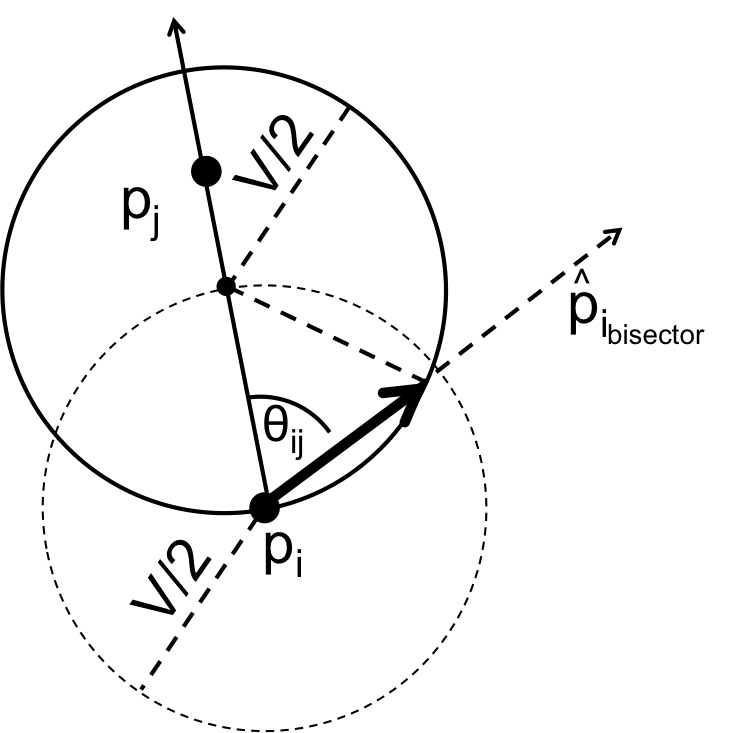}
    \caption{Allowable region $AR_{ij}$ section line along agent's $i$ movement direction. The thick arrow is $Limit_{ij}$.}
      \label{Limit_ij_Gordon}
\end{figure}
By dynamics (\ref{eq:Dynamics77}) an agent $i$ moves only when $\psi_i(t)$ is less than $\pi$, and therefore for each neighbour $j$ of agent $i$ the angle $\theta_{ij}(t)$ is less than $\pi/2$. Thereby,
$$  Limit_i(t) =  \min\limits_{j \in N_i}\{Limit_{ij}(t)\} = \min\{V/2, \; \min\limits_{j \in N_i}\{V\cos(\theta_{ij}(t))\}\} > 0$$
and when agent $i$ moves, it moves toward the direction $p_{i_{bisector}}(t)$, hence it moves inside its allowable region $AR_i(t)$.\\ \\
\textbf{Second argument:}\\ \\
According to the motion law (\ref{eq:Dynamics77}) agent $i$ travels from $p_i(t)$ in the direction  $\hat p_{i_{bisector}}(t)$ or does not move. By geometry (see Figure \ref{FigGordon1}(b)), the length of the section that runs from point $p_i(t)$ along the bisector of $\psi_i(t)$ to the boundary of the allowable region of $p_i(t)$ is $V/2$ if $\psi_i(t) \leq 2\pi/3$ due to $D_{\frac{V}{2}}p_i(t)$, and $V\cos(\psi/2)$ if $2\pi/3 < \psi_i(t) < \pi$, running to zero as $\psi_i(t)$ runs to $\pi$. Therefore if $\psi_i(t) < \pi$ agent $i$ moves, otherwise it stays put. We next provide a general proof for such a geometric constellation of circles forming the allowable region (see Figure \ref{GordonShortestBisector}).\\

Without loss of generality, let two circles of radius $r$, whose centers are at a significant but smaller than $2r$ distance from each other, be symmetrically positioned in a $2d$ cartesian coordinate system $XY$ so that one of their meeting points is at the origin, and the other meeting point is along the positive half of the $Y$ axis. Denote the smallest angle between $X$ axis and the line connecting the origin to the center of one of those circles by $\alpha$. By geometry, the non trivial meeting point of those circles (denoted in Figure \ref{GordonShortestBisector} by $Y_{\alpha}$) is at $[o;2r\sin \alpha]$.\\


\begin{figure}[h!]
\captionsetup{width=0.8\textwidth}
  \centering
    \includegraphics[width=90mm]{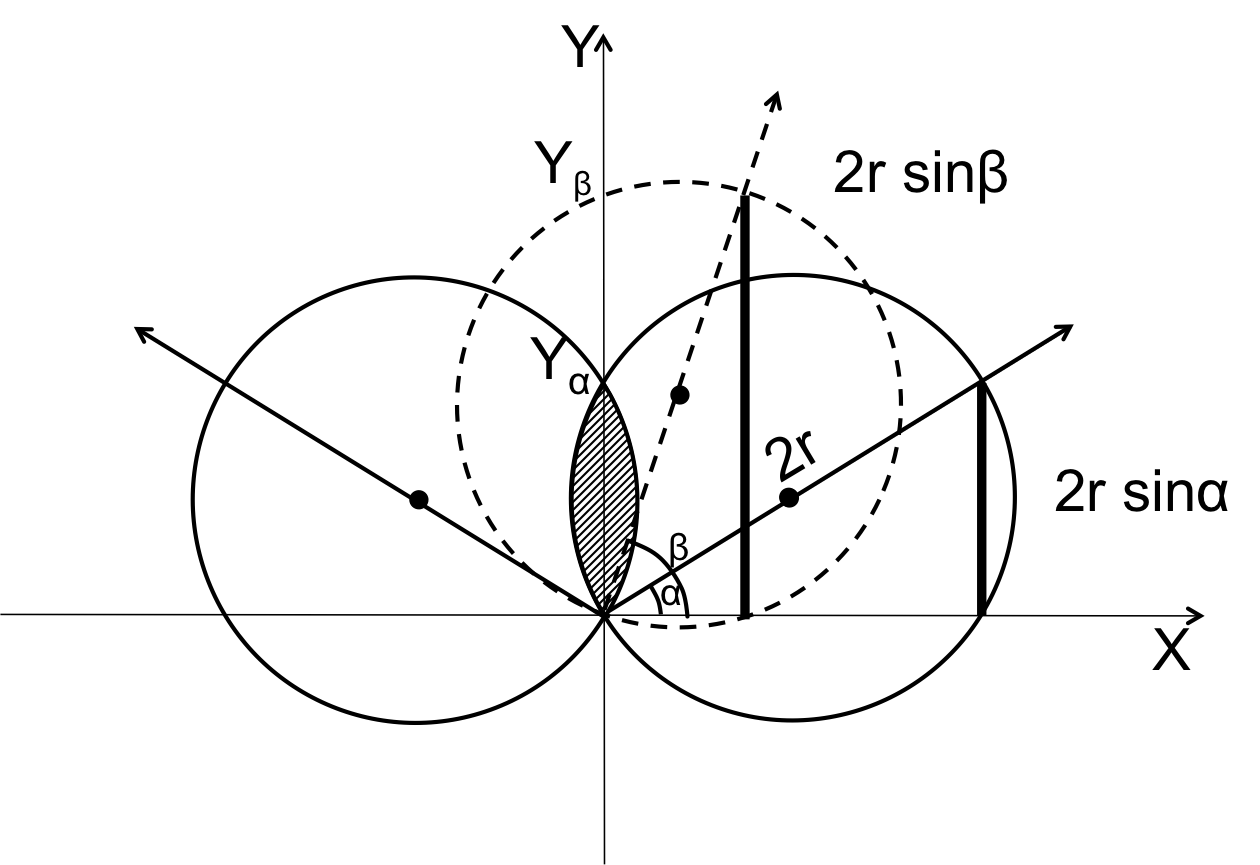}
    \caption{A proof that the segment along the bisector of the extreme rays  is contained in all other circles inside the wedge angle.}
      \label{GordonShortestBisector}
\end{figure}


Let a third circle of the same radius $r$ be positioned at that coordinate system, so that its center is at a distance $r$ from the origin, and its center is anywhere inside the minimal angular section defined by the centers of the two original circles and the origin, as shown in Figure \ref{GordonShortestBisector} in dashed lines. Denote the smallest angle between $X$ axis and the line connecting the origin and the center of that circle by $\beta$. By geometry, the non trivial meeting point of that circle with $Y$ axis (denoted in Figure \ref{GordonShortestBisector} by $Y_{\beta}$) is at $[o;2r\sin \beta]$.\\

Since $0 < \alpha < \beta \leq \pi/2$, we have that $\sin \alpha < \sin \beta$, hence $Y_{\alpha} < Y_{\beta}$. Therefore, point $Y_{\alpha}$ is necessarily included in any circle of radius $r$, whose center, distanced $r$ from the origin, is located inside the minimal angular section defined by the centers of the two external circles and the origin.\\

Since by dynamics (\ref{eq:Dynamics77}) all the agents in system $\mathcal{S}_7$ move inside their allowable region along their associated bisectors, by Lemma \ref{NeverLoseFriendLemma} they will not lose any of their neighbours, as claimed in Lemma \ref{NeverLoseFriendUnderDynm}.\\
\end{proof}

\begin{lemma}\label{GordonCHnotGrow}
Let $CH(P(t))$ be the convex-hull of agents' positions in system $\mathcal{S}_7$ at time $t$. Then for all $t \geq 0$ and $\Delta t > 0$
$$CH(P(t+\Delta t)) \subseteq CH(P(t))$$
\end{lemma}

\begin{proof}
By the dynamic law (\ref{eq:Dynamics77}) each agent moves along the bisector of $\psi_i(t)$, where $\psi_i(t)$ is the angle of the smallest wedge containing all the neighbours of agent $i$. And since there is no agent located outside the convex-hull of the system, no agent moves out of the convex-hull.\\
\end{proof}

\begin{lemma} \label{GordonCHshrink}\label{CHdecreases77}
If the graph topology of system $\mathcal{S}_7$ is connected and the perimeter of its convex-hull is greater than zero, then the perimeter of its convex-hull decreases at a rate bounded away from zero.\\
\end{lemma}

\begin{proof} 

We shall show that the perimeter of $CH(P(t))$ drops at a strictly positive rate as long as the diameter of the system is strictly positive.\\
 
The proof is based on the dynamics of the agent (or agents) $s$, located at a current sharpest corner of the system's convex-hull. Let $\varphi_s$ be the inner angle of this corner.\\

The sum of angles of any convex polygon is  $\pi(m-2)$, where $m$ is the number of its corners, therefore the angle of its sharpest corner $\varphi_s$ is at most $\pi(1-\frac{2}{m})$. System $\mathcal{S}_7$ contains $n$ agents, hence the system's convex-hull has $ m \le n$ corners. We denote the upper limit on the sharpest corner of the convex-hull by $\varphi_*$, so
$$\varphi_s \le \varphi_* = \pi(1 - 2/n)$$

Define $\mathcal{L}(P(t))$ as the perimeter of $CH(P(t))$ and $l_i(t)$ as the length of the convex-hull side connecting corners $i$ and $i+1$ at time $t$.\\

Let $\varphi_i(t)$ be the angle of corner $i$ of $CH(P(t))$, let $\alpha_i(t)$ denote the direction of motion of the agent located at corner $i$ relative to the direction of corner $i+1$, and let $v_i(t)$ be the speed of the agent located at corner $i$ (as shown in Figure \ref{CHshrink}).\\
\begin{figure}[h!]
\captionsetup{width=0.8\textwidth}
  \centering
    \includegraphics[width=100mm]{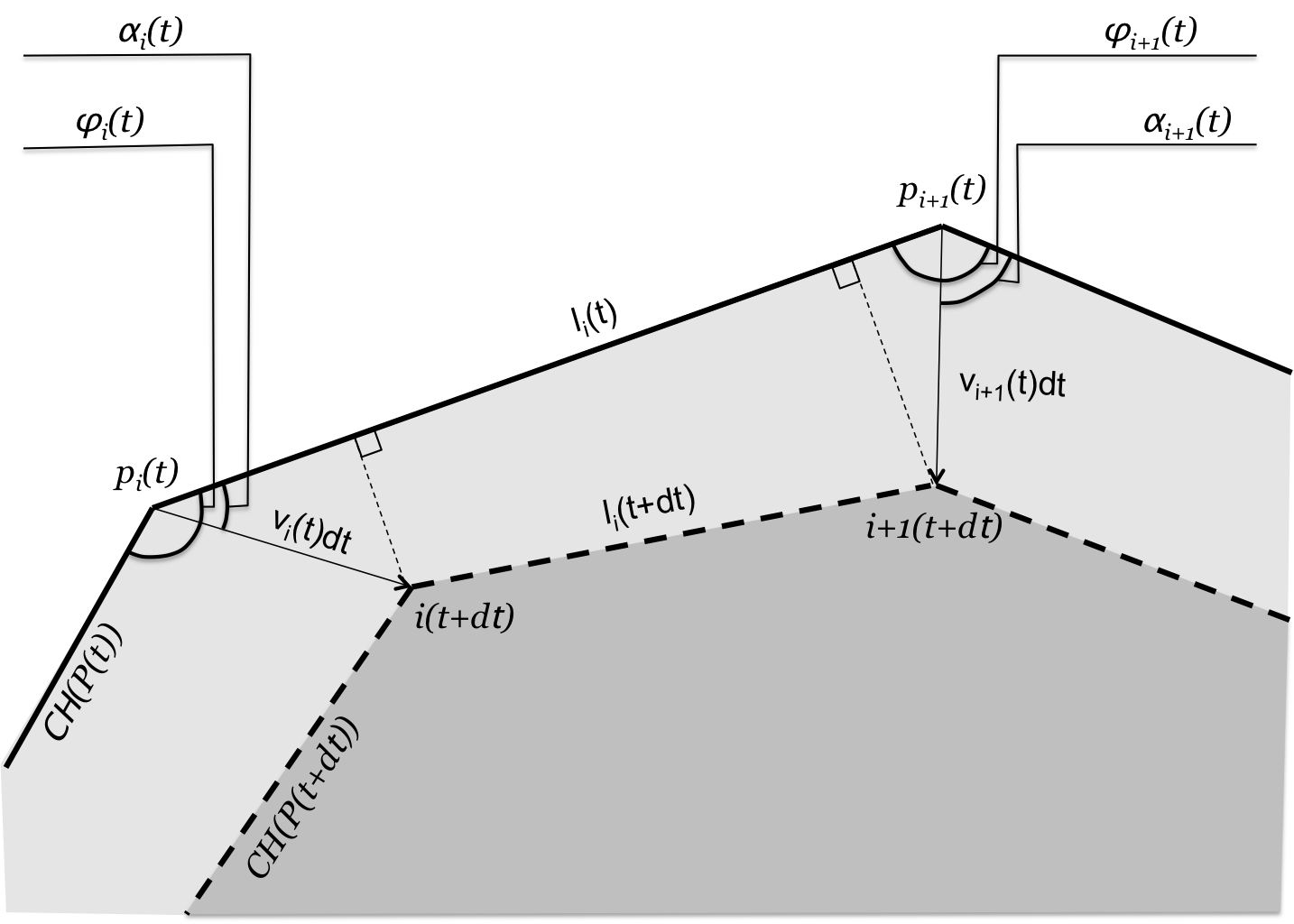} 
     \caption{Convex-hull shrinkage.}
     \label{CHshrink}
\end{figure}

We have that 
$$\lim\limits_{dt \to 0^+}\{ l_i(t+dt)-l_i(t) \} = $$ $$-(v_i(t) \cos \alpha_i(t) + v_{i+1}(t)\cos (\varphi_{i+1} (t) -\alpha_{i+1}(t))) dt + \mathcal{O}(dt)$$
hence,
$$\dot{\mathcal{L}}(P(t))=-\sum\limits_{i=1}^m v_i(t)(\cos \alpha_i(t)+\cos (\varphi_{i} (t) -\alpha_{i}(t)))$$
$$= -\sum\limits_{i=1}^m v_i(t) \cos(\frac{\varphi_{i} (t)}{2}) \cos(\frac{\varphi_{i}(t) - 2\alpha_{i}(t)}{2}) $$

Let $v_s(t)$, $\varphi_s(t)$ and $\alpha_s(t)$ be the relevant values associated with agent $s$.
Since $v_s = \| \dot{p}_s(t) \|$ is positive and bounded away from zero by the assumed rule of motion (\ref{eq:Dynamics77}), and by Lemma \ref{CHdecreases77} we have $\alpha_{i}(t) \le \varphi_i(t) < \pi$ we have 
$$\dot{\mathcal{L}}(P(t)) \leq - v_s(t) \cos(\frac{\varphi_{s} (t)}{2}) \cos(\frac{\varphi_{s}(t) - 2\alpha_{s}(t)}{2}) \le -\sigma \cos^2(\frac{\varphi_*}{2}) $$
Hence, the perimeter decreases at a rate of at least $\sigma \cos^2(\varphi_*/2)$ proving Lemma \ref{GordonCHshrink}.\\
\end{proof}

To prove Theorem \ref{Theorem77} we have by Lemma (\ref{CHdecreases77}) that the length of the perimeter of the convex-hull of system $\mathcal{S}_7$ decreases at a bounded away form zero rate, therefore it necessarily converges to a point in finite time as claimed.\\
\end{proof}

\subsection{Discrete Time Dynamics (system $\mathcal{S}_8$)}

In the discrete time models discussed by Gordon et. al. \cite{gordon2004, gordon2005}, the step-size of each agent varies, and is limited by the maximal value $\sigma$. A semi-synchronous motion rule is considered, so that at each time-step, agents have a certain positive probability to be active, and each active agent $i$ jumps a step-size limited by $\sigma$ to a point in its current allowable region $AR_i(k)$, similar to the allowable region defined in (\ref{NeverLoseFriendRig}) for the continuous case of system $\mathcal{S}_7$. To prove gathering, the \textit{strong asynchronicity} property is used. One needs the fact that at each time-step, there is a strictly positive probability that only one agent is active, and that probability is at least $\delta$. In addition to the deterministic move into the allowable region, Gordon also considers a randomized motion model \cite{gordon2005}, which allows the agents to jump to any point inside their current allowable region (rather then strictly along the bisector of the region's wedge angle).\\

Why is a semi-synchronization motion schedule necessary for this type of system and not for the previous ones?\\

The agents of system $\mathcal{S}_8$ not only lack information about the distance to their neighbours, but their visibility range is limited to $V$, so that their interconnection topology graph changes whenever new pairs of agents become visible to each other. In order to maintain visibility (i.e. once two agents are visible, hence neighbours to each other, they remain neighbours), the dynamics corresponding to rule (\ref{eq:Dynamics77}) postulates that any agent with a wedge angle $\psi_i(k)$ greater than or equal to $\pi$, is "locked" and cannot move until its associated wedge angle becomes sharp.\\

To see how this rule of motion causes trouble in the case of synchronized, discrete time schedule, let us analyze the following example. Consider the constellation of agents presented in Figure \ref{GordonLocked}:
\begin{figure}[H]
\captionsetup{width=0.8\textwidth}
  \centering
    \includegraphics[width=60mm]{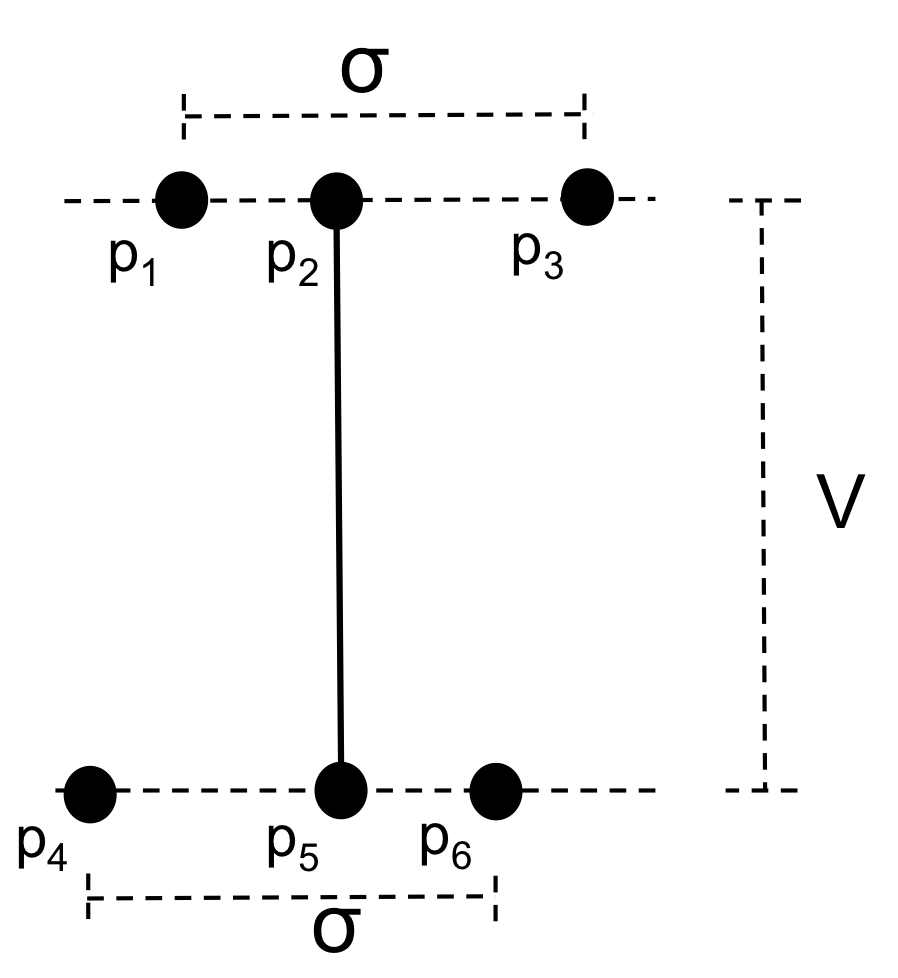}
    \caption{A special constellation that prevents system $\mathcal{S}_8$ without randomization from gathering to a point.}
      \label{GordonLocked}
\end{figure}
Note that in Figure \ref{GordonLocked}, the agents $p_1$, $p_2$, $p_3$ and $p_4$, $p_5$, $p_6$ are located on parallel lines so that $\bar{p_1p_3} \| \bar{p_4p_6}$. These parallel lines are at a distance $V$ from each other, and only $\bar{p_2p_5}$ is  perpendicular to $\bar{p_1p_3}$ (and to $\bar{p_4p_6}$), so that only $\|p_2-p_5\| = V$. Assume $\|p_1-p_3\| = \|p_4-p_6\|=\sigma < V$. By geometry both $p_1$ and $p_3$ are not visible to $p_4$, $p_5$ and $p_6$ since they are at a distance of more than $V$ from them, and both $p_4$ and $p_6$ are not visible to $p_1$, $p_2$ and $p_3$. Considering dynamic rule similar to (\ref{eq:Dynamics77}) but in a discrete time framework, where all agents are active at each time step, we have that at time-step $k$ the wedge angles of $p_2$ and $p_5$ are $\psi_2(k) =\psi_5(k)= \pi$, therefore both $p_2$ and $p_5$ are locked. At time-step $k+1$ both agents $p_1$ and $p_3$ must move a step of size $\sigma$ towards each other, so that they switch positions, and so do $p_4$ and $p_6$. The same switching phenomenon occurs over and over again simultaneously, leaving $p_2$ and $p_5$ locked forever, preventing the system from gathering.\\

This example shows that a deterministic motion schedule can lead to non-gathering constellations, hence some randomization is needed. Indeed adding randomization to the motion schedule may break this "locked" situation and "free" the agents to move. For example, in the constellation above, if, once in a while, an agent "sleeps" and doesn't move (resulting, due to the jumps of $p_1$ to $p_3$ while $p_3$ sleeps or due to the jump of $p_4$ while $p_6$ sleeps, in $\psi_2(k)=\pi/2$ or $\psi_5(k)=\pi/2$), agents $p_2$ and $p_5$ will approach each other, and eventually more agents become visible to each other.\\

Let us use the definition of allowable region (\ref{NeverLoseFriendRig}), where each agent $i$ can move without losing any existing neighbour. In order to obey the constraint of having a limited step (to $\sigma$) in discrete time, the maximal step-size of an agent can never exceed
$$\mu = \min\{\frac{V}{2}, \sigma\}$$
and therefore in the discrete case, the allowable region is given by (see Figure \ref{GordonDescreteAR}):

\begin{equation} \label{DiscreteAlowableRegionRig}
AR_i(k) \triangleq  \left( \bigcap \limits_{j \in N_i(k)}  D_{\frac{V}{2}}(c_{ij}(k)) \right) \cap D_{\mu}(p_i(k))
\end{equation}
similar to (\ref{NeverLoseFriendRig}).\\

\begin{figure}[h!]
\captionsetup{width=0.8\textwidth}
  \centering
    \includegraphics[width=70mm]{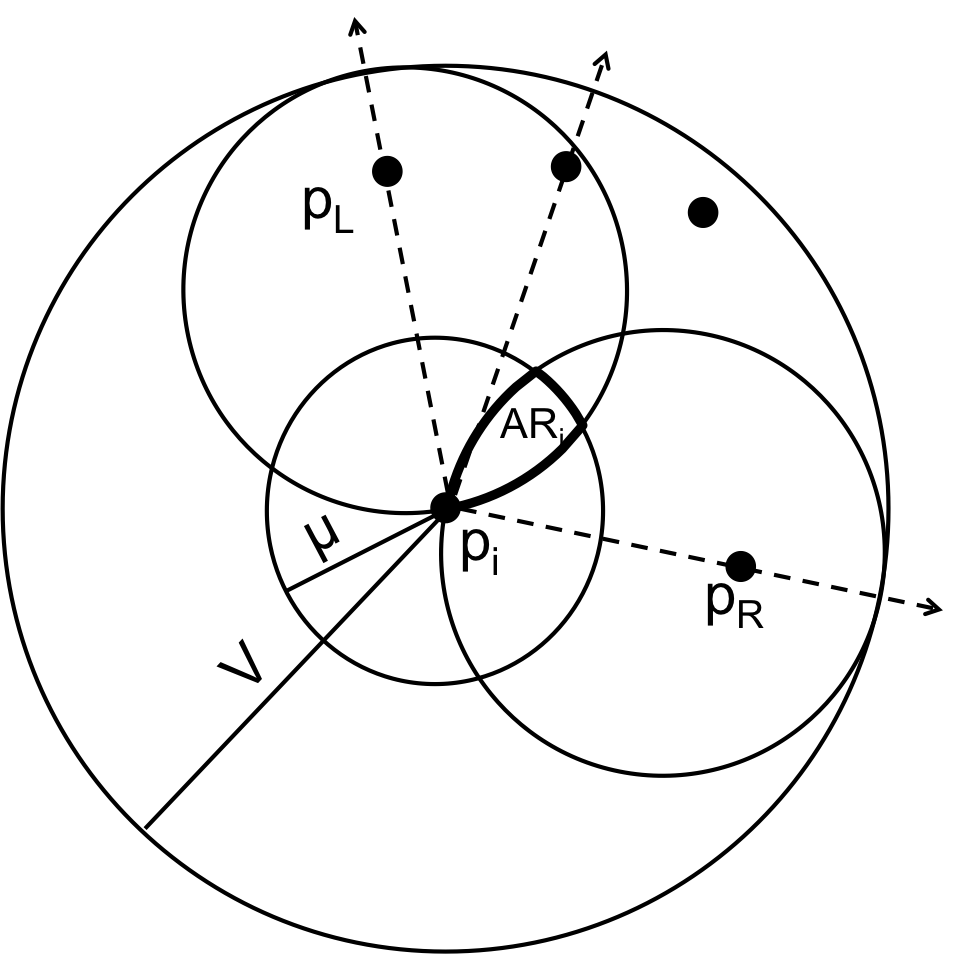}
    \caption{Allowable region of agent $i$ in system $\mathcal{S}_8$. Allowable region for agent $i$ of system $\mathcal{S}_8$ (in thick line). points $p_L(k)$ and $p_R(k)$ are the position of the current extreme left and right agents in agent's $i$ neighbours wedge.}
      \label{GordonDescreteAR}
\end{figure}

For system $\mathcal{S}_8$ we shall assume, following Gordon et. al. \cite{gordon2005}, that each agent has a strictly positive probability to be active at each time-step, and that each active agent $i$ moves to a uniformly distributed random point in its allowable region $AR_i(k)$. Note that due to the \textit{strong asynchronicity}, at each time-step there is a strictly positive probability ($>\delta$) that only a single agent is active.\\

\begin{theorem} \label{Theorem88}
For any initial constellation having a connected graph topology, all agents of the system $\mathcal{S}_8$ gather to a disk of diameter $V$ in a finite expected number of time-steps.
\end{theorem}

The outline of the proof is as follows: at each time-step there is always a probability $\rho$ bounded away from zero for an active agent on the convex-hull of the agents' locations to reduce its distance from the current average position of all the agents $\bar p$ by a strictly positive amount $s^*$. Therefore, there is a probability of at least $\delta \rho$ that the sum of squared distances of all agents from $\bar p$, will decrease by at least ${s^*}^2/n$ (where $\delta$ is positive and bounded away from zero value as well by the strong asynchronicity assumption).\\

As long as the agents' interconnection graph is not complete, there is always a bounded away from zero probability that it becomes complete within finite number of time-steps $M$. Using Lemma \ref{NeverLoseFriendLemma}, once the agents interconnection graph is complete it stays complete forever, so that the maximal distance between any two agents remains smaller than or equal to $V$. Therefore, all agents are henceforth confined to a disc of diameter $V$.\\

We have already seen that, in case of discrete time dynamics, when agents lack information about the distance to their neighbours, overshoot phenomena may occur, so instead of a gathering to single points, agents gather to a bounded region (as for example seen in system $\mathcal{S}_4$).\\

\begin{proof}
Let $CH(P(k))$ be the current convex-hull of the agents' constellation, $\varphi_i(k)$ the current internal angle of the convex-hull vertex denoted by $i$, and $D(P(k))$ the current diameter of the convex-hull defined by:
$$D(P(k)) \triangleq max_{i,j} \|p_i(k)- p_j(k)\|$$
Denote by $\mathcal{L}(P(k))$ the sum of squared distances of all agents from their current average position $\bar{p}(k)$:
$$\mathcal{L}(P(k)) = \sum_{i=1}^{n} \|p_i(k) - \bar{p}(k)\|^2$$

\begin{lemma} \label{DistFormAvg}
For any agent $i$ located at a corner of $CH(P(k))$, the distance between $p_i(k)$ and $\bar{p}(k)$ is bounded as follows:
$$\|p_i(k)-\bar{p}(k)\| \geq \frac{D(P(k))}{2 n \tan(\varphi_i(k)/2)}$$
where $D(P(k)) \triangleq max_{i,j} \|p_i(k) - p_j(k)\|$ is the current diameter of the convex-hull, and $\varphi_i(k)$ is the current internal angle of the convex-hull angle denoted by $i$.
\end{lemma}

\begin{proof}
Any agent $i$, located at a corner of $CH(P(k))$, either defines the convex-hull diameter together with another agent $j$ so that
$$D(P(k))=\|p_j(k)-p_i(k)\|$$
or there are two other agents $j_1$ and $j_2$ defining its diameter, so that
$$D(P(k))=\|p_{j_1}(k)-p_{j_2}(k)\|$$

By Proposition \ref{MinDistanceFromC} (see system $\mathcal{S}_6$) we have that

\begin{equation} \label{MinimalSizeOfEdge}
max\{\|p_i(k)-p_{j_1}(k)\|, \|p_i(k)-p_{j_2}(k)\|\} \geq \frac{D(P(k))}{2 \sin(\varphi_i(k)/2)}
\end{equation}
We also have that:
$$\|\bar{p}(k)-p_i(k)\| = \|\frac{1}{n}\sum_{j=1}^{n} (p_j(k)-p_i(k))\| = \|\frac{1}{n}\sum_{j=1}^{n} (p_j(k)-p_i(k))\|$$
Let $\theta_{ij}(k)$ be the angle between the vector ($p_j(k) - p_i(k)$) and the bisector of $\varphi_i(k)$.\\

\begin{figure}[h!]
\captionsetup{width=0.8\textwidth}
  \centering
    \includegraphics[width=70mm]{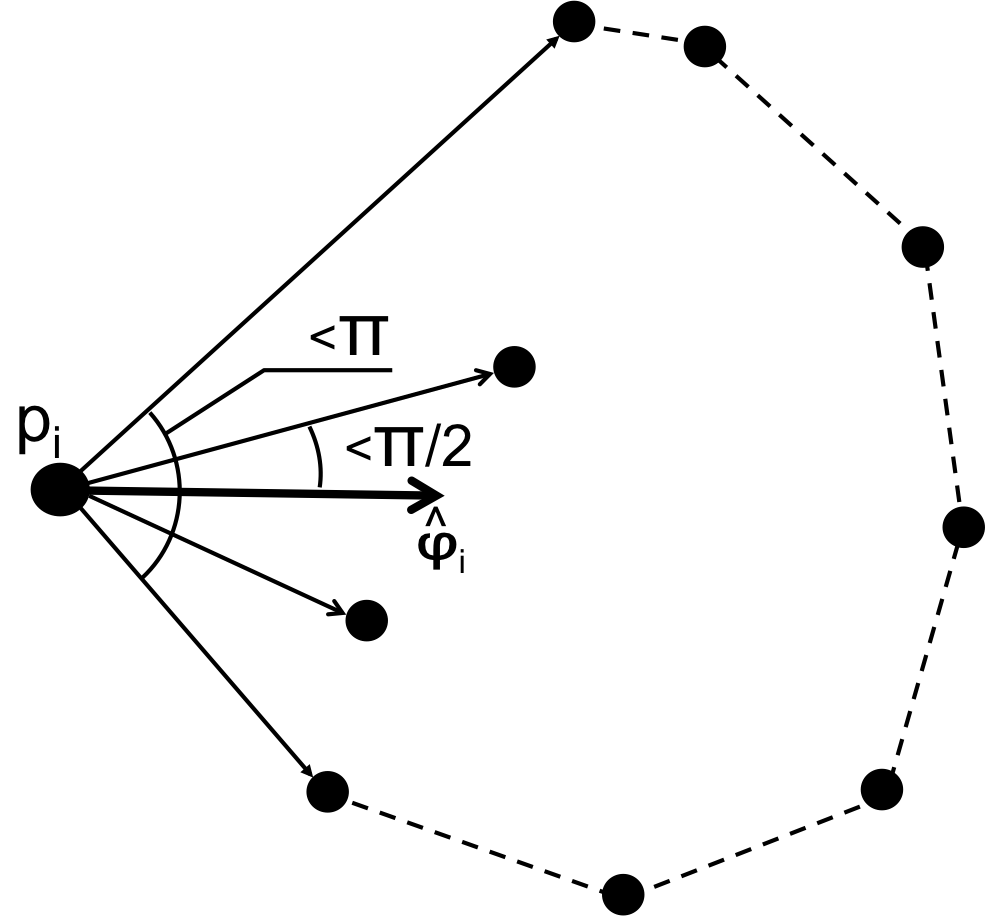}
    \caption{Since each angle of the convex-hull is smaller than $\pi$, any angle defined by an internal agent, a convex-hull corner and its associated bisector is smaller than $\pi/2$.}
      \label{HalfSharpAngle}
\end{figure}

Since $\cos(\theta_{ij}(k))>0$ (see Figure \ref{HalfSharpAngle}), we have that
$$\|\bar{p}(k)-p_i(k)\| \geq \frac{1}{n}\sum_{j=1}^{n} \|p_j(k)-p_i(k)\|\cos{\theta_{ij}(k)}$$
and using (\ref{MinimalSizeOfEdge}) we get:\\ 
$$\|\bar{p}(k)-p_i(k)\| \ge \frac{1}{n}\sum_{j=1}^{n} \|p_j(k)-p_i(k)\|\cos{\theta_{ij}(k)} \geq \frac{D(P(k))}{2n \sin(\varphi_i(k)/2)}\cos{(\varphi_i(k)/2)}$$
as claimed.\\
\end{proof}

Let us show using Lemma \ref{DistFormAvg} that if the diameter of the convex-hull is bounded away from zero, it has at least one corner whose distance from the current average position of the agents $\bar{p}(k)$ is bounded away from zero as well.\\

\begin{corollary}  \label{Finite}
If the length of the diameter of $CH(P(k))$ is bounded away from zero, the distance between $\bar{p}(k)$ to agent $s$, the agent located at the sharpest corner of the system's convex-hull, is bounded away from zero as well.
\end{corollary}

\begin{proof}
By Lemma \ref{DistFormAvg} we have that  $$\|p_i(k)-\bar{p}(k)\|\geq \frac{D(P(k))}{2 n \tan(\varphi_i(k)/2)}$$
and since $\varphi_{s}(k) \leq \varphi_* < \pi$ (as presented in the proof of Lemma \ref{CHdecreases77}), we have that
$$ \|p_s(k)-\bar{p}(k)\| \geq \frac{D(P(k))}{2 n \tan(\varphi_{s}(k)/2)} \geq \frac{D(P(k))}{2 n \tan(\varphi_*/2)} $$
Hence $\|p_s(k)-\bar{p}(k)\|$ is bounded away from zero as claimed.\\
\end{proof}

\begin{lemma} \label{GordonAgentsMovment}
There exist strictly positive constants $\rho$ and $s^*$, so that for any constellation $P(k)$, while $D(P(k))$ is bounded away from zero, if agent $s$ is active, the probability that at the next time-step it will be closer to $\bar{p}(k)$ by a distance of at least $s^*$ is at least $\rho$.
\end {lemma}

\begin{proof}
Let $\psi_s(k)$ be the angle of the minimal sector where agent $s$ is located at the head of the sector and it contains all neighbours of $s$, such that we have $\psi_s(k) \le \varphi_{s}(k) \leq \varphi_*$. Hence, by geometry, we have that the allowable region of agent $s$ has an area bounded away from zero (see Figure \ref{GordonDisc_Allowable_RegionArea}):
$$
\begin{array}{l}
AR_s(k) =
\left\{
\begin{array}{ll}
\frac{1}{2} V^2 \left( \pi - \psi_s - \frac{1}{2}\sin(\psi_s) \right ) &  \psi_s \ge 2\alpha \quad :(a)\\
\frac{1}{2} V^2 \left( \pi - 2\alpha - \frac{1}{2}\sin(2\alpha) \right ) + \frac{1}{2}\mu ^2(2\alpha - \psi_s) & \psi_s < 2\alpha \quad :(b)\\
\end{array}
\right.
\end{array}
$$
where $\mu = \min\{\frac{V}{2}, \sigma\}$, and $\alpha \triangleq \cos^{-1}(\mu/V)$.\\\\
(Note that the time-step index $(k)$ has been omitted in the equation above for convenience of reading).\\

\begin{figure}[h]
\captionsetup{width=0.8\textwidth}
  \centering
    \includegraphics[width=120mm]{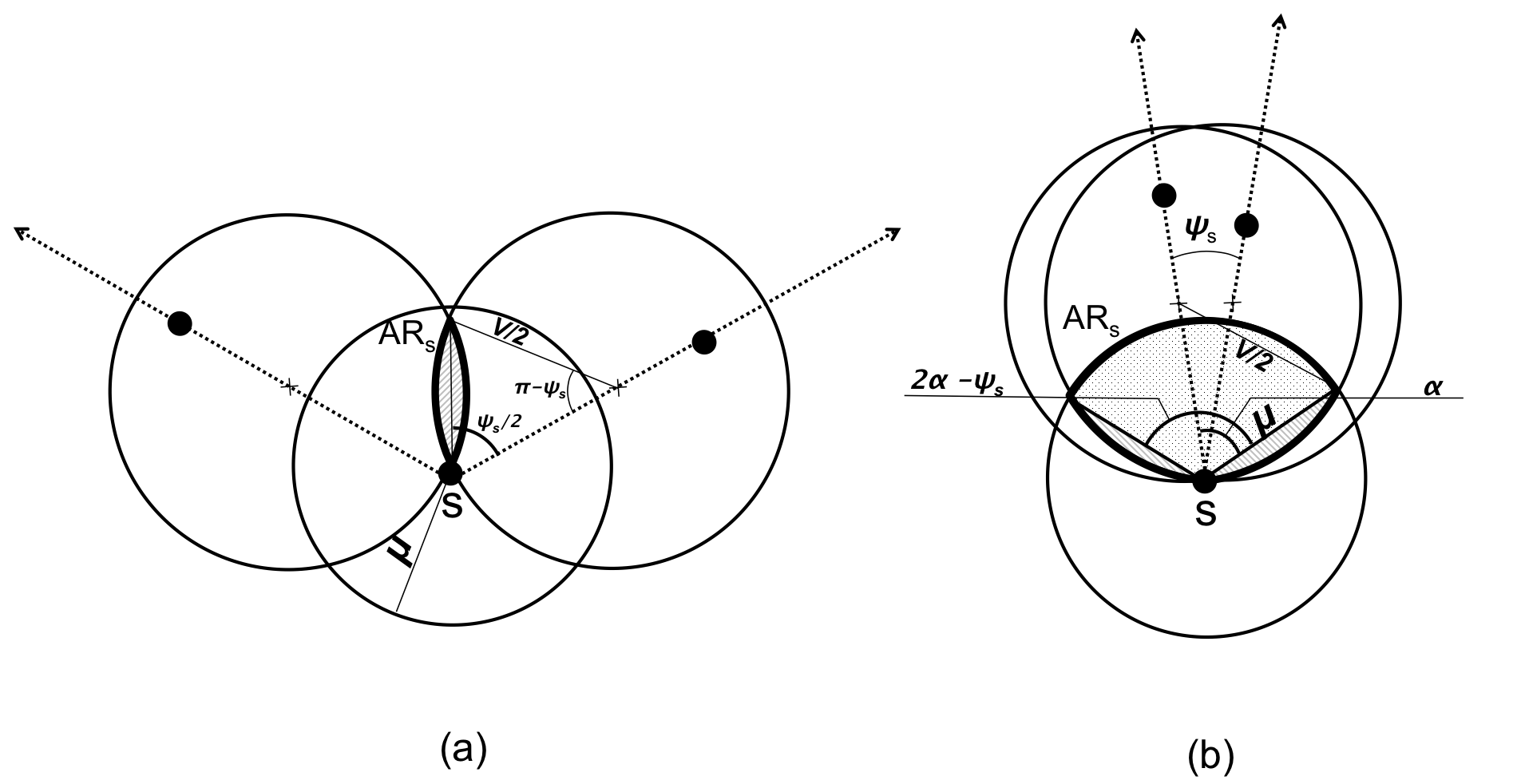}
    \caption{Allowable region of agent $s$. (a) In case the allowable region does not reach a distance of $\mu$ from agent $s$, its area may be calculated as the area of two segments of circles denoted by the dashed area. (b) In case the allowable region reaches $\mu$, its area may be calculated as the area of two segments of circles denoted by the dashed area and a sector denoted by the dotted area.}
      \label{GordonDisc_Allowable_RegionArea}
\end{figure}

Let $D_{\|\bar{p}(k)-p_s(k)\| - s^*}(\bar{p}(k))$ be a disc centered at $\bar{p}(k)$ with a radius set to be $\|\bar{p}(k)-p_s(k)\| - s^*$, where $s^*$ is a small and bounded away from zero value. If agent $s$ jumps inside that disc, it is guaranteed to be closer to $\bar p(k)$ (compared to where it was before the jump) by at least $s^*$.\\

The current agents' average position $\bar p(k)$ is located inside $CH(P(k))$, hence for any agent $i \neq s$ the angle $\angle p_ip_s\bar{p}$ is smaller than $\varphi_s < \varphi_*$. Furthermore, by Corollary \ref{Finite}, if the diameter of the system is bounded away from zero, the distance from agents $s$ to $\bar p(k)$ is bounded away from zero as well, therefore for a small enough yet strictly positive $s^*$, the area of the intersection of $AR_{s}(k)$ and $D_{\|\bar{p}(k)-p_s(k)\| - s^*}(\bar{p}(k))$ is bounded away form zero as well (see Figure \ref{GordonRandomize}).\\

\begin{figure}[h]
\captionsetup{width=0.8\textwidth}
  \centering
    \includegraphics[width=100mm]{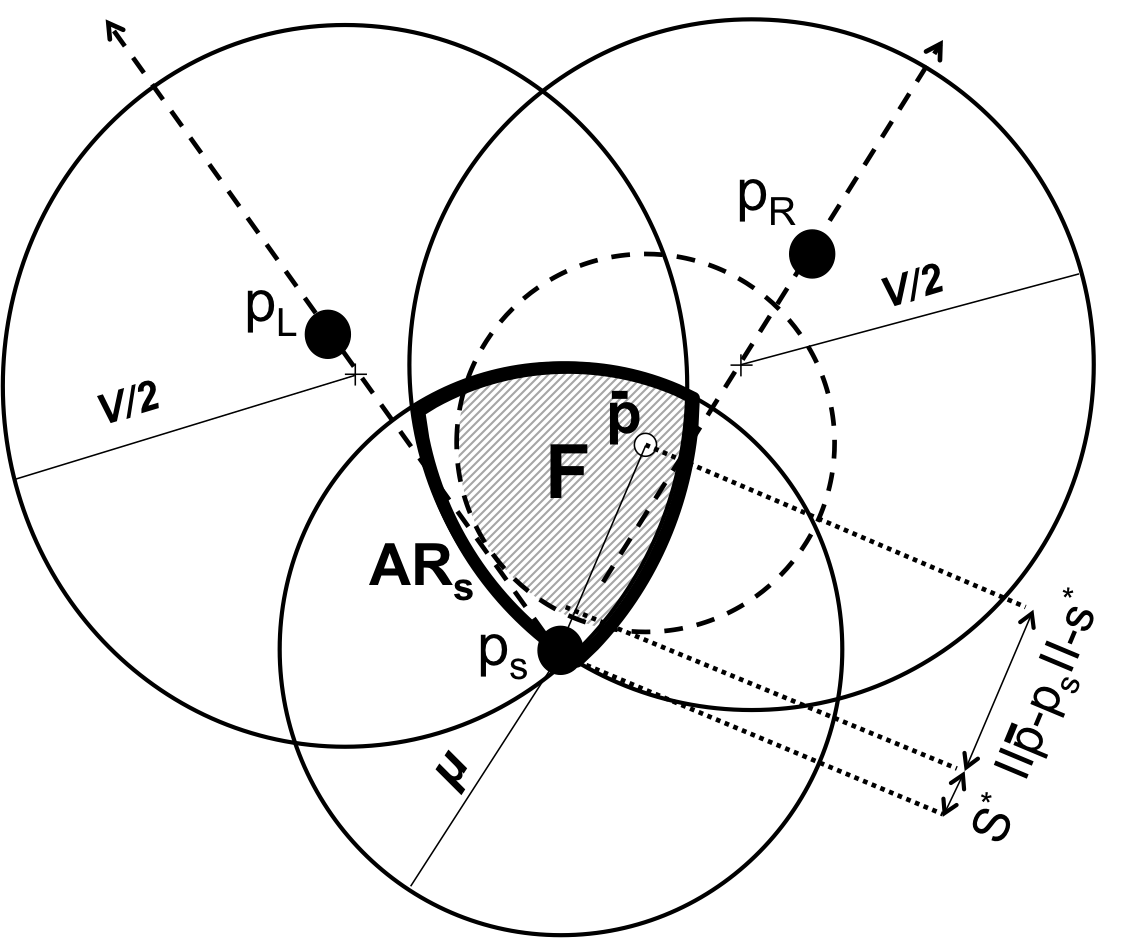}
    \caption{The probability that the current distance of agent $s$ from $\bar{p}$ will decrease in the next time-step is the proportion between the grey area $F$ and the current allowable region $AR_s$ of agent $s$ (bounded by a thick line).}
      \label{GordonRandomize}
\end{figure}

Denote this intersection region by $F(k)$, i.e.
$$ F(k) \triangleq AR_{s}(k) \cap D_{\|\bar{p}(k)-p_s(k)\| - s^*}(\bar{p}(k))$$
We have that the area of region $F(k)$ is bounded away from zero, and we have that each agent moves to a uniformly distributed random point in its allowable region. Using that, let us denote the probability that agent $s$ moves inside area $F(k)$ by
$$\rho = \frac{\|F(k)\|}{\|AR_{s}(k)\|}$$
where $\|F(k)\|$ and $\|AR_{s}(k)\|$ are the areas of regions $F(k)$ and $AR_{s}(k)$ respectively, and therefore $\rho$ is strictly positive.\\

We have that whenever agent $s$ is active and moves into area $F(k)$, it moves closer to $\bar{p}(k)$ by a bounded away from zero value. Therefore, $\rho$, the probability that agent $s$ will be closer to $\bar{p}(k)$ is bounded away from zero as claimed.

\end{proof}

\begin{lemma}
There is at least $\delta \rho$ probability for $\mathcal{L}(P(k))$ to decrease by at least $\|s^*\|^2/n$ at each time-step.
\end{lemma}

\begin{proof}
The constant $\delta$ is defined as the lower bound for the probability that at any time-step $k$ only agent $s$ becomes active. The probability that agent $s$ makes a step $\Delta p_s(k)$ of a size $\tilde{s} \ge s^*$ inside region $F$ is $\rho$ as shown in Lemma \ref{GordonAgentsMovment}. In this case the value of $\mathcal{L}(P(k))$ decreases as follows:  
\begin{equation} \label{DeltaVar}
\mathcal{L}(P(k+1)) - \mathcal{L}(P(k))=\tilde{s}(\tilde{s} - 2\|p_s(k) - \bar p(k)\|\cos(\theta_{s}(k)) - \frac{{\tilde{s}}^2}{n}
\end{equation}
where 
$$\theta_{s}(k) = \angle\bar{p}(k)p_{s}(k)p_{s}(k+1))$$\\
To prove this we proceed as follows:\\ \\
Let $\Delta p_s(k)$ the position displacement of agent $s$ between time steps $k$ and $k+1$, i.e. $\Delta p_s(k)=p_s(k+1)-p_s(k)$
$$ \mathcal{L}(P(k+1)) = \sum\limits_{i=1}^{n} \| p_i(k+1) - \bar p(k+1) \|^2 = $$
$$ \sum\limits_{\substack{i=1 \\ i \neq s}}^{n} \| p_i(k) - (\bar p(k)+\frac{\Delta p_s(k)}{n}) \|^2 + \| p_s(k) + \Delta p_s(k) - (\bar p(k)+\frac{\Delta p_s(k)}{n}) \|^2 = $$
$$ \mathcal{L}(P(k)) + 2\frac{1}{n}\Delta p^\intercal_s(k) \left( (n-1)(p_s(k) - \bar p(k)) -\sum\limits_{\substack{i=1 \\ i \neq s}}^{n}(p_i(k) - \bar p(k))  \right) +$$ $$+\frac{(n-1)+(n-1)^2}{n^2} \|\Delta p_s(k)\|^2 = $$
$$ \mathcal{L}(P(k)) + \|\Delta p_s(k)\|^2  + 2\Delta p^\intercal_s(k) (p_s(k) - \bar p(k)) - \frac{\|\Delta p_s(k)\|^2}{n} =$$
$$ \mathcal{L}(P(k)) + \tilde{s}(\tilde{s} - 2\|p_s(k) - \bar p(k)\|\cos(\theta_{s}(k)) - \frac{{\tilde{s}}^2}{n} $$

Now, if agent $s$ moves inside disc $D_{\|\bar{p}(k)-p_s(k)\|}(\bar{p}(k))$, then
$$\tilde{s}(\tilde{s} - 2\|p_s(k) - \bar p(k)\|\cos(\theta_{s}(k)) < 0 $$
hence,
$$\mathcal{L}(P(k+1)) - \mathcal{L}(P(k)) < - \frac{{\tilde{s}}^2}{n}$$
Therefore, the probability that $\mathcal{L}(P(k))$ will decrease by at least $\frac{{s^*}^2}{n}$ is bounded from below by $\delta \rho$ as claimed.\\
\end{proof}

Back to the proof of Theorem \ref{Theorem88}. Since the initial agents' interconnection graph is connected, $D(P) \leq (n-1)V$. Note that $D(P)$ gets this maximal value when the agents are evenly distributed along a straight line, with a distance V between neighbours. Therefore $\mathcal{L}(P) < n((n-1)V)^2$.\\

In addition, if $\mathcal{L}(P) \le \left(\frac{V}{2}\right)^2$ then the agents' interconnection graph is necessarily fully connected, since the maximal distance of an agent from $\bar{p}$ is $\frac{V}{2}$, and hence all inter-agent distances are necessarily less than $V$.\\

The transition from any arbitrary constellation to a fully connected constellation may be achieved in finite number of possible steps $M$, where 
$$ M < \frac{n((n-1)V)^2 - V^2/4}{\|s^*\|^2/n} +1$$

Let us examine the evolution of the agents' constellation every $M$ steps. At the end of each series of $M$ steps, the probability that $\mathcal{L}(P) < V^2/4$ is at least $(\delta \rho)^M$. Let $q$ be this minimal probability, $q = (\delta \rho)^M$. The expected number of series of $M$ steps for gathering to a clique is at most:
$$\sum\limits_{i=1}^{\infty} i(1 - q)^{i-1}q =  -q\frac{d}{dq}\sum\limits_{i=1}^{\infty} (1 - q)^i = -q\frac{d}{dq}\frac{1}{q} = \frac{1}{q}$$
hence, the expected number of steps for gathering is at most:

$$M\frac{1}{(\delta \rho)^M}$$

By Lemma \ref{NeverLoseFriendLemma}, once a fully connected constellation is reached the system remains fully connected. Therefore gathering to a bounded region is achieved in a finite expected number of time-steps.

\end{proof}
 
\subsection{Discussion}

\textbf{Simulation results}: Simulation results for $\mathcal{S}_7$ are shown in Figure \ref{FigGordonCHSim}. Note that the agents' trajectories are rather complex, and agents meet and "travel together" toward the gathering point.\\
 
\begin{figure}[h!]
\captionsetup{width=0.8\textwidth}
  \centering
    \includegraphics[width=100mm]{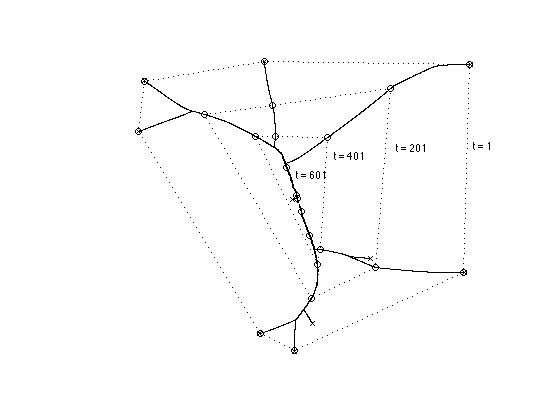}
    \caption{Simulation result of system $\mathcal{S}_7$ with 10 agents and an arbitrary initial constellation.}
      \label{FigGordonCHSim}
\end{figure}

Simulation results for $\mathcal{S}_8$ are shown in Figure \ref{GordonDisc}. It is easy to notice that the agents do not gather to a point.\\
 
\begin{figure}[h!]
\captionsetup{width=0.8\textwidth}
  \centering
    \includegraphics[width=100mm]{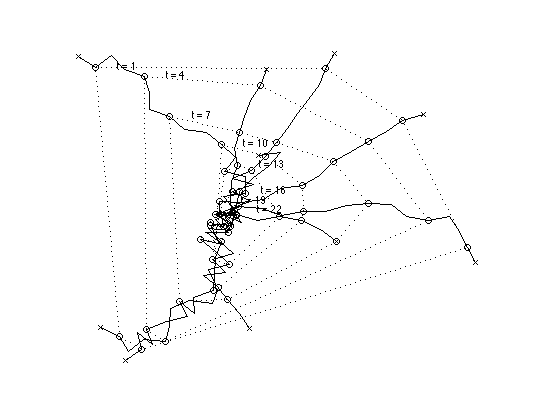}
    \caption{Simulation result for system $\mathcal{S}_8$ with 10 agents and an arbitrary initial constellation.}
      \label{GordonDisc}
\end{figure}

\textbf{Comparison: Ando et. al. vs. Gordon et. al. methods:}\\

It is interesting to compare the discrete time algorithm suggested by Gordon et. al. (system $\mathcal{S}_8$) to the one suggested by Ando et. al. (system $\mathcal{S}_6$). Both are based on an allowable region $AR_i(k)$ where any agent $i$ can move without losing visibility to any of its existing neighbours, but while Ando deals with agents capable of measuring the relative location of their neighbours (including both bearing angle and relative distance), Gordon's allowable region is based on agents lacking the capability to measure the relative distance, and therefore, in order not to lose an existing neighbour, Gordon needs to consider all possible distances to the visible neighbour in the range $[0,V]$. This reduces the allowable region size and limits the longest step an agent can make, which in turn may affect the speed of convergence.\\

The allowable region for agent $i$ by Ando's method is:
$$AR_i^{Ando} (k)= \bigcap \limits_{j \in N_i(k)} D_{\frac{V}{2}}(\frac{p_i(k)+p_j(k)}{2})$$
as compared to Gordon's which is given by (\ref{NeverLoseFriendRig}):
$$AR_i^{Gordon}(t) \triangleq  \left( \bigcap \limits_{j \in N_i(t)}  D_{\frac{V}{2}}(c_{ij}(t)) \right) \cap D_{\frac{V}{2}}(p_i(t))$$

\newpage
\addcontentsline{toc}{section}{Conclusions}
\section*{Conclusions}
This report surveys gathering of multi robotic systems in eight cases defined by the sensor capabilities and motion timing assumed for the robotic agents. The survey clearly showcases the wide verity of approaches that can be used for proving the correctness of the dynamic processes of geometric consensus or gathering, and pointed out the weak points of the existing analysis, and the lack appropriate tools for analyzing the speed of convergence and size of gathering cluster in almost all practical cases.\\

We hope that the collection of tools and results described here will in the future be extended and used in novel ways to  obtain better results in terms of realistic convergence rates and area estimates for the region of gathering, as well as in suggesting novel distributed dynamics. \\

Most importantly, the survey we present exhibits the considerable gap that exists between the classical control communities that address gathering problems via algebraic and Lyapunov-functions based methods,  and the computer science community which addresses the same problems via geometric and algorithmic approaches. We believe that the best results will emerge from a synergy between the methods employed so far and perhaps some new ones yet to be invented.\\

\newpage
\addcontentsline{toc}{section}{References}
\bibliography{MARS_group_new}\

\begin{thebibliography}{10}

\bibitem{reif1999social}
John~H Reif and Hongyan Wang.
\newblock Social potential fields: A distributed behavioral control for
  autonomous robots.
\newblock {\em Robotics and Autonomous Systems}, 27(3):171--194, 1999.

\bibitem{jadbabaie2003}
A.~Jadbabaie, Jie Lin, and A.S. Morse.
\newblock Coordination of groups of mobile autonomous agents using nearest
  neighbor rules.
\newblock {\em Automatic Control, IEEE Transactions on}, 48(6):988--1001, 2003.

\bibitem{gazi2003stability}
Veysel Gazi and Kevin~M. Passino.
\newblock Stability analysis of swarms.
\newblock {\em IEEE Transactions on Automatic Control}, 48:692--697, 2003.

\bibitem{gazi2004stability}
Veysel Gazi and Kevin~M Passino.
\newblock Stability analysis of social foraging swarms.
\newblock {\em Systems, Man, and Cybernetics, Part B: Cybernetics, IEEE
  Transactions on}, 34(1):539--557, 2004.

\bibitem{moreau2004}
Luc Moreau.
\newblock Stability of continuous-time distributed consensus algorithms.
\newblock In {\em Decision and Control, 2004. CDC. 43rd IEEE Conference on},
  volume~4, pages 3998--4003. IEEE, 2004.

\bibitem{ren2005consensus}
Wei Ren, Randal~W Beard, et~al.
\newblock Consensus seeking in multiagent systems under dynamically changing
  interaction topologies.
\newblock {\em IEEE Transactions on automatic control}, 50(5):655--661, 2005.

\bibitem{olfati2007consensus}
Reza Olfati-Saber, J~Alex Fax, and Richard~M Murray.
\newblock Consensus and cooperation in networked multi-agent systems.
\newblock {\em Proceedings of the IEEE}, 95(1):215--233, 2007.

\bibitem{ji2007}
Meng Ji and Magnus~B Egerstedt.
\newblock Distributed coordination control of multi-agent systems while
  preserving connectedness.
\newblock {\em Robotics, IEEE Transactions on}, 23(4):693--703, Aug 2007.

\bibitem{cucker2007emergent}
Felipe Cucker and Steve Smale.
\newblock Emergent behavior in flocks.
\newblock {\em Automatic Control, IEEE Transactions on}, 52(5):852--862, 2007.

\bibitem{motsch2014heterophilious}
Sebastien Motsch and Eitan Tadmor.
\newblock Heterophilious dynamics enhances consensus.
\newblock {\em SIAM review}, 56(4):577--621, 2014.

\bibitem{reynolds1987flocks}
Craig~W Reynolds.
\newblock Flocks, herds and schools: A distributed behavioral model.
\newblock In {\em ACM Siggraph Computer Graphics}, volume~21, pages 25--34.
  ACM, 1987.

\bibitem{chazelle2014convergence}
Bernard Chazelle.
\newblock The convergence of bird flocking.
\newblock {\em Journal of the ACM (JACM)}, 61(4):21, 2014.

\bibitem{chazelle2015algorithmic}
Bernard Chazelle.
\newblock An algorithmic approach to collective behavior.
\newblock {\em Journal of Statistical Physics}, 158(3):514--548, 2015.

\bibitem{feynman1985surely}
Richard Feynman.
\newblock {\em Surely You're Joking, Mr. Feynman!}
\newblock W. W. Norton \& Company, 1985.

\bibitem{bruckstein1991ants}
Alfred~M Bruckstein, N~Cohen, and A~Efrat.
\newblock {\em Ants, crickets and frogs in cyclic pursuit}.
\newblock Technion-Israel Institute of Technology. Center for Intelligent
  Systems, 1991.

\bibitem{bruckstein1993ant}
Alfred~M Bruckstein.
\newblock Why the ant trails look so straight and nice.
\newblock {\em The Mathematical Intelligencer}, 15(2):59--62, 1993.

\bibitem{wagner1997row}
Israel~A Wagner and Alfred~M Bruckstein.
\newblock Row straightening via local interactions.
\newblock {\em Circuits, Systems and Signal Processing}, 16(3):287--305, 1997.

\bibitem{bruckstein1997probabilistic}
AM~Bruckstein, CL~Mallows, and IA~Wagner.
\newblock Probabilistic pursuits on the grid.
\newblock {\em American Mathematical Monthly}, pages 323--343, 1997.

\bibitem{marshall2004formations}
Joshua Marshall, Mireille~E Broucke, Bruce Francis, et~al.
\newblock Formations of vehicles in cyclic pursuit.
\newblock {\em Automatic Control, IEEE Transactions on}, 49(11):1963--1974,
  2004.

\bibitem{lin2005necessary}
Zhiyun Lin, Bruce Francis, and Manfredi Maggiore.
\newblock Necessary and sufficient graphical conditions for formation control
  of unicycles.
\newblock {\em Automatic Control, IEEE Transactions on}, 50(1):121--127, 2005.

\bibitem{belkhouche2005modeling}
Fethi Belkhouche and Boumediene Belkhouche.
\newblock Modeling and controlling a robotic convoy using guidance laws
  strategies.
\newblock {\em Systems, Man, and Cybernetics, Part B: Cybernetics, IEEE
  Transactions on}, 35(4):813--825, 2005.

\bibitem{martinez2006optimal}
Sonia Mart{\'\i}nez and Francesco Bullo.
\newblock Optimal sensor placement and motion coordination for target tracking.
\newblock {\em Automatica}, 42(4):661--668, 2006.

\bibitem{sinha2006generalization}
Arpita Sinha and Debasish Ghose.
\newblock Generalization of linear cyclic pursuit with application to
  rendezvous of multiple autonomous agents.
\newblock {\em Automatic Control, IEEE Transactions on}, 51(11):1819--1824,
  2006.

\bibitem{sinha2007generalization}
Arpita Sinha and Debasish Ghose.
\newblock Generalization of nonlinear cyclic pursuit.
\newblock {\em Automatica}, 43(11):1954--1960, 2007.

\bibitem{hristu2007bio}
Dimitrios Hristu-Varsakelis and Changguo Shao.
\newblock A bio-inspired pursuit strategy for optimal control with partially
  constrained final state.
\newblock {\em Automatica}, 43(7):1265--1273, 2007.

\bibitem{oggier2012cyclic}
Fr{\'e}d{\'e}rique Oggier and Alfred Bruckstein.
\newblock On cyclic and nearly cyclic multiagent interactions in the plane.
\newblock In {\em A Panorama of Modern Operator Theory and Related Topics},
  pages 513--539. Springer, 2012.

\bibitem{suzuki1999distributed}
Ichiro Suzuki and Masafumi Yamashita.
\newblock Distributed anonymous mobile robots: Formation of geometric patterns.
\newblock {\em SIAM Journal on Computing}, 28(4):1347--1363, 1999.

\bibitem{ando1999}
Hideki Ando, Yoshinobu Oasa, Ichiro Suzuki, and Masafumi Yamashita.
\newblock Distributed memoryless point convergence algorithm for mobile robots
  with limited visibility.
\newblock {\em Robotics and Automation, IEEE Transactions on}, 15(5):818--828,
  1999.

\bibitem{cieliebak2003solving}
Mark Cieliebak, Paola Flocchini, Giuseppe Prencipe, and Nicola Santoro.
\newblock Solving the robots gathering problem.
\newblock In {\em Automata, Languages and Programming}, pages 1181--1196.
  Springer, 2003.

\bibitem{schlude2003robotics}
Konrad Schlude.
\newblock From robotics to facility location: contraction functions, weber
  point, convex core.
\newblock Technical Report 403, Computer Science, ETHZ, 2003.

\bibitem{schlude2003point}
Konrad Schlude.
\newblock {\em Point Formation on a line: Contraction Functions and Weber
  point}.
\newblock PhD thesis, ETH, Eidgen{\"o}ssische Technische Hochschule Z{\"u}rich,
  Information Security Group, 2003.

\bibitem{gordon2004}
Noam Gordon, Israel~A Wagner, and Alfred~M Bruckstein.
\newblock Gathering multiple robotic a (ge) nts with limited sensing
  capabilities.
\newblock In {\em Ant Colony Optimization and Swarm Intelligence}, volume 3172
  of {\em Lecture Notes in Computer Science}, pages 142--153. Springer, 2004.

\bibitem{gordon2005}
Noam Gordon, Israel~A Wagner, and Alfred~M Bruckstein.
\newblock A randomized gathering algorithm for multiple robots with limited
  sensing capabilities.
\newblock In {\em Proc. of MARS 2005 workshop at ICINCO Barcelona}, 2005.

\bibitem{flocchini2005gathering}
Paola Flocchini, Giuseppe Prencipe, Nicola Santoro, and Peter Widmayer.
\newblock Gathering of asynchronous robots with limited visibility.
\newblock {\em Theoretical Computer Science}, 337(1):147--168, 2005.

\bibitem{cohen2005convergence}
Reuven Cohen and David Peleg.
\newblock Convergence properties of the gravitational algorithm in asynchronous
  robot systems.
\newblock {\em SIAM Journal on Computing}, 34(6):1516--1528, 2005.

\bibitem{agmon2006fault}
Noa Agmon and David Peleg.
\newblock Fault-tolerant gathering algorithms for autonomous mobile robots.
\newblock {\em SIAM Journal on Computing}, 36(1):56--82, 2006.

\bibitem{cortes2006robust}
Jorge Cort{\'e}s, Sonia Mart{\'\i}nez, and Francesco Bullo.
\newblock Robust rendezvous for mobile autonomous agents via proximity graphs
  in arbitrary dimensions.
\newblock {\em Automatic Control, IEEE Transactions on}, 51(8):1289--1298,
  2006.

\bibitem{martinez2007motion}
Sonia Mart{\'\i}nez, Jorge Cortes, and Francesco Bullo.
\newblock Motion coordination with distributed information.
\newblock {\em Control Systems, IEEE}, 27(4):75--88, 2007.

\bibitem{gordon2008}
Noam Gordon, Yotam Elor, and AlfredM. Bruckstein.
\newblock Gathering multiple robotic agents with crude distance sensing
  capabilities.
\newblock In {\em Ant Colony Optimization and Swarm Intelligence}, volume 5217
  of {\em Lecture Notes in Computer Science}, pages 72--83. Springer Berlin
  Heidelberg, 2008.

\bibitem{gordon2010fundamental}
Noam Gordon.
\newblock {\em Fundamental Problems in the Theory of Multi-Agent Robotics}.
\newblock PhD thesis, Technion, 2010.

\bibitem{cieliebak2012distributed}
Mark Cieliebak, Paola Flocchini, Giuseppe Prencipe, and Nicola Santoro.
\newblock Distributed computing by mobile robots: Gathering.
\newblock {\em SIAM Journal on Computing}, 41(4):829--879, 2012.

\bibitem{bellaiche2015}
Levi-Itzhak~Bellaiche Alfred~Bruckstein.
\newblock Continuous time gathering of agents with limited visibility and
  bearing-only sensing.
\newblock Technical report, CIS Technical Report, TASP, 2015.

\bibitem{mamei2006field}
Marco Mamei and Franco Zambonelli.
\newblock {\em Field-based coordination for pervasive multiagent systems}.
\newblock Springer Science \& Business Media, 2006.

\bibitem{bullo2009distributed}
Francesco Bullo, Jorge Cort{\'e}s, and Sonia Martinez.
\newblock {\em Distributed Control of Robotic Networks: A Mathematical Approach
  to Motion Coordination Algorithms}.
\newblock Princeton University Press, 2009.

\bibitem{mesbahi2010graph}
Mehran Mesbahi and Magnus Egerstedt.
\newblock {\em Graph theoretic methods in multiagent networks}.
\newblock Princeton University Press, 2010.

\bibitem{gazi2011swarm}
Veysel Gazi and Kevin~M Passino.
\newblock {\em Swarm stability and optimization}.
\newblock Springer Science \& Business Media, 2011.

\bibitem{bonato2011game}
Anthony Bonato and Richard~J Nowakowski.
\newblock {\em The game of cops and robbers on graphs}, volume~61.
\newblock American Mathematical Society Providence, 2011.

\bibitem{flocchini2012distributed}
Paola Flocchini, Giuseppe Prencipe, and Nicola Santoro.
\newblock Distributed computing by oblivious mobile robots.
\newblock {\em Synthesis Lectures on Distributed Computing Theory},
  3(2):1--185, 2012.

\bibitem{okubo1986dynamical}
Akira Okubo.
\newblock Dynamical aspects of animal grouping: swarms, schools, flocks, and
  herds.
\newblock {\em Advances in biophysics}, 22:1--94, 1986.

\bibitem{flierl1999individuals}
G~Flierl, D~Gr{\"u}nbaum, S~Levin, and D~Olson.
\newblock From individuals to aggregations: the interplay between behavior and
  physics.
\newblock {\em Journal of Theoretical Biology}, 196(4):397--454, 1999.

\bibitem{camazine2003self}
Scott Camazine.
\newblock {\em Self-organization in biological systems}.
\newblock Princeton University Press, 2003.

\bibitem{couzin2003self}
Iain~D Couzin and Jens Krause.
\newblock Self-organization and collective behavior in vertebrates.
\newblock {\em Advances in the Study of Behavior}, 32:1--75, 2003.

\bibitem{sumpter2006principles}
David~JT Sumpter.
\newblock The principles of collective animal behaviour.
\newblock {\em Philosophical Transactions of the Royal Society B: Biological
  Sciences}, 361(1465):5--22, 2006.

\bibitem{hildenbrandt2010self}
Hanno Hildenbrandt, Cladio Carere, and Charlotte~K Hemelrijk.
\newblock Self-organized aerial displays of thousands of starlings: a model.
\newblock {\em Behavioral Ecology}, 21(6):1349--1359, 2010.

\bibitem{ben1992adaptive}
Eshel Ben-Jacob, Haim Shmueli, Ofer Shochet, and Adam Tenenbaum.
\newblock Adaptive self-organization during growth of bacterial colonies.
\newblock {\em Physica A: Statistical Mechanics and its Applications},
  187(3):378--424, 1992.

\bibitem{vicsek1995novel}
Tam{\'a}s Vicsek, Andr{\'a}s Czir{\'o}k, Eshel Ben-Jacob, Inon Cohen, and Ofer
  Shochet.
\newblock Novel type of phase transition in a system of self-driven particles.
\newblock {\em Physical review letters}, 75(6):1226, 1995.

\bibitem{mogilner1996spatio}
Alex Mogilner and Leah Edelstein-Keshet.
\newblock Spatio-angular order in populations of self-aligning objects:
  formation of oriented patches.
\newblock {\em Physica D: Nonlinear Phenomena}, 89(3):346--367, 1996.

\bibitem{vicsek2012collective}
Tam{\'a}s Vicsek and Anna Zafeiris.
\newblock Collective motion.
\newblock {\em Physics Reports}, 517(3):71--140, 2012.

\bibitem{cavagna2014bird}
Andrea Cavagna and Irene Giardina.
\newblock Bird flocks as condensed matter.
\newblock {\em Annu. Rev. Condens. Matter Phys.}, 5(1):183--207, 2014.

\bibitem{mirollo1990synchronization}
Renato~E Mirollo and Steven~H Strogatz.
\newblock Synchronization of pulse-coupled biological oscillators.
\newblock {\em SIAM Journal on Applied Mathematics}, 50(6):1645--1662, 1990.

\bibitem{strogatz2000kuramoto}
Steven~H Strogatz.
\newblock From {K}uramoto to {C}rawford: exploring the onset of synchronization
  in populations of coupled oscillators.
\newblock {\em Physica D: Nonlinear Phenomena}, 143(1):1--20, 2000.

\bibitem{wang2005partial}
Wei Wang and Jean-Jacques~E Slotine.
\newblock On partial contraction analysis for coupled nonlinear oscillators.
\newblock {\em Biological cybernetics}, 92(1):38--53, 2005.

\bibitem{dorfler2014synchronization}
Florian D{\"o}rfler and Francesco Bullo.
\newblock Synchronization in complex networks of phase oscillators: A survey.
\newblock {\em Automatica}, 50(6):1539--1564, 2014.

\bibitem{cybenko1989dynamic}
George Cybenko.
\newblock Dynamic load balancing for distributed memory multiprocessors.
\newblock {\em Journal of parallel and distributed computing}, 7(2):279--301,
  1989.

\bibitem{xiao2004fast}
Lin Xiao and Stephen Boyd.
\newblock Fast linear iterations for distributed averaging.
\newblock {\em Systems \& Control Letters}, 53(1):65--78, 2004.

\bibitem{olshevsky2009convergence}
Alex Olshevsky and John~N Tsitsiklis.
\newblock Convergence speed in distributed consensus and averaging.
\newblock {\em SIAM Journal on Control and Optimization}, 48(1):33--55, 2009.

\bibitem{elzinga1972geometrical}
Jack Elzinga and Donald~W Hearn.
\newblock Geometrical solutions for some minimax location problems.
\newblock {\em Transportation Science}, 6(4):379--394, 1972.

\end{thebibliography}
\bibliographystyle{unsrt}

\end{document}